\documentclass[11pt]{article}
\usepackage[utf8]{inputenc}
\usepackage[margin=1in]{geometry}
\usepackage[T1]{fontenc}
\usepackage{lmodern}

\usepackage{microtype}
\usepackage{graphicx}
\usepackage{subfigure,wrapfig}
\usepackage{booktabs} 
\usepackage{natbib}
\setcitestyle{open={(},close={)}}

\usepackage[colorlinks=true,citecolor=blue]{hyperref}

\usepackage{amsmath,amsthm,amsfonts,amssymb,mathdots,array,mathrsfs,bm,bbm,stmaryrd,graphicx,subfigure,xcolor}
\usepackage[vlined,linesnumbered,ruled,resetcount]{algorithm2e}

\usepackage{breakcites}

\usepackage[T1]{fontenc}
\usepackage{enumerate}
\usepackage{inputenc}

\usepackage{graphicx} 
\usepackage{subfigure}

\usepackage{booktabs,balance}
\usepackage{rotating}
\usepackage{boldline}
\usepackage{makecell}
\usepackage{multirow}
\usepackage{balance}

\usepackage{tikz}

\newtheorem{theorem}{Theorem}[section]

\newtheorem{lemma}[theorem]{Lemma}

\newtheorem{definition}{Definition}[section]

\newtheorem{remark}{Remark}[section]

\newcommand{\bsmat}{\begin{bmatrix} }
\newcommand{\esmat}{\end{bmatrix} }

\SetKwInput{KwInput}{Input}
\SetKwInput{KwData}{Data}

\DeclareMathOperator*{\argmin}{argmin}

\begin{document}

\title{\bf\Huge Best Subset Selection with Efficient Primal-Dual Algorithm}

\author{\vspace{0.5in}\\\textbf{Shaogang Ren, Guanhua Fang,  Ping Li} \\\\
Cognitive Computing Lab\\
Baidu Research\\
10900 NE 8th St. Bellevue, WA 98004, USA\\\\
  \texttt{\{renshaogang, fanggh2018,\ pingli98\}@gmail.com}
}
\date{}
\maketitle

\begin{abstract}\vspace{0.2in}
\noindent\footnote{The initial version of the paper was submitted in February 2020.}Best subset selection is considered the `gold standard' for many sparse learning problems. A variety of optimization techniques have been proposed to attack this non-convex and NP-hard problem.  In this paper, we investigate the dual forms of a family of $\ell_0$-regularized problems. An efficient primal-dual method has been developed based on the primal and dual problem structures. By leveraging the dual range estimation along with the incremental strategy, our algorithm  potentially reduces redundant computation and improves the solutions of best subset selection. Theoretical analysis and experiments on synthetic and real-world datasets validate the efficiency and statistical properties of the proposed solutions.
\end{abstract}

\newpage

\section{Introduction}

Sparse learning is a standard approach to alleviate model over-fitting issues when the feature dimension is larger than the number of training samples. With a training set $\{x_i, y_i\}_{i=1}^n$ where $x_i \in \mathbb{R}^p $ is the sample feature and $y_i$ is the corresponding label, this paper focuses on the following generalized best subset selection problem,
\begin{align}\label{eq:primal1}
&\min_{\beta \in \mathbb{R}^p} F(\beta) = f(\beta) + \lambda_0 ||\beta||_0,  \\\notag
\text{where }\hspace{0.2in} &f(\beta) =   \sum_{i=1}^n l(\beta^{\top} x_i, y_i)   +  \lambda_1 ||\beta||_1 +\lambda_2 ||\beta||_2^2.
\end{align}
Here $l(\cdot)$ is a convex function, $\beta$ is the model parameter,  and $\lambda_0$,  $\lambda_1$ and  $\lambda_2$ are hyper-parameters/tuning parameters. It is well-known that an $\ell_0$ solver ($\lambda_0>0,\lambda_1=0$) has superior statistical properties  when the \textit{signal-to-noise ratio} (SNR) is high, but it may suffer over-fitting issues when SNR is low~\citep{pmlr-v65-david17a,mazumder2022subset}. The continuous-shrinkage solvers  e.g., ridge/Lasso ($\lambda_0=0,\lambda_1>0$), can perform better in this case compared with $\ell_0$ solver~\citep{mazumder2022subset,hastie2017extended}. Combinations of there hyper-parameters may adjust the model to work well in different noise levels. \cite{mazumder2022subset,Hazimeh18} used ridge/Lasso to improve $\ell_0$ solutions and achieve better or comparable solutions~with~less~nonzeros.

 Solving the vanilla $\ell_0$ regularized problem  is known to be NP-hard~\citep{Natarajan95}.  By leveraging the significant computational advances in mixed integer optimization (MIO), ~\cite{Bertsimas2015BestSS}  performed near optimal solutions to a special case of  problem~\eqref{eq:primal1}, for $\lambda_1 = 0$ and $\lambda_2 = 0$. This method scaled up solutions to  cases where feature sizes are much larger than what were considered possible in the community~\citep{Furnival74,Hazimeh18}. Their approach can achieve approximate optimality via dual bounds but with the cost of longer computation times. ~\cite{Bertsimas17} showed that cutting plane methods for subset selection can work well with mild sample correlations and a succinctly large $n$.

 Different from  the soft regularized ridge/Lasso problem given by~\eqref{eq:primal1},  Iterative Hard Thresholding~(IHT)~\citep{blumensath2009iterative,foucart2011hard,yuan2014gradient,jie2017tight,yuan2020nearly} has often been used to solve  $k$-sparse problems~\eqref{eq:hard_k},
 \begin{align}\label{eq:hard_k}
&\min_{||\beta||_0 \leq k} \sum_{i=1}^n l(\beta^{\top} x_i, y_i)    +\lambda_2 ||\beta||_2^2 \ .
\end{align}
In~\cite{blumensath2009iterative,foucart2011hard}, the authors demonstrated that IHT can be applied to compute the compressed sensing problem. IHT-based approaches have been studied by many researchers in the context of sparse learning problems~\citep{yuan2014gradient,jain2014iterative,jain2016structured,yuan2020nearly}. IHT methods require a specific value of the   features number~($k$) to start the algorithm.  Many other approaches have also been developed to tackle the $\ell_0$ regularized problems~\citep{Mazumder15,mazumder2022subset,soussen2015homotopy,bian2020smoothing,dedieu2020learning,yang2020fast,dong2015regularization,hazimeh2020sparse,zhu2020polynomial}.

Apart from $\ell_0$ solvers, extremely efficient and optimized   $\ell_1$-regularization (Lasso) solvers can solve an entire regularization path (with a hundred values of the tuning parameter) in usually less than a second~\citep{friedman2010regularization}. Screening and coordinate incremental techniques~\citep{Fercoq2015,GAP,Celer,ren2020thunder} can further scale the solutions to large datasets.  Compared to popular efficient solvers for Lasso, it seems that the high computation cost for using $\ell_0$ regularized models ~\citep{Bertsimas2015BestSS} might discourage practitioners from adopting global optimization-based solvers of~\eqref{eq:primal1} to daily analysis applications~\citep{hastie2017extended,mazumder2022subset,Hazimeh18}.  However, it is known~\citep{Loh14,Hazimeh18} that there is a significant gap in the statistical quality of solutions that can be achieved via Lasso (and its variants) and near-optimal solutions to non-convex subset-selection type procedures. The choice of algorithm can significantly affect the quality of solutions obtained. On many instances, algorithms that~do a better job in optimizing the non-convex subset-selection criterion ~\eqref{eq:primal1} result in superior-quality statistical estimators (for example, in terms of support recovery~\citep{Hazimeh18}).

Several recent studies attempt to further improve the efficiency of $\ell_0$ solvers. Along the line of dual methods, ~\cite{Liu17,yuan2020dual} recently studied the duality of $k$-sparse problem, and they proved the strong duality of $k$-sparse problem.  With the derived dual form by~\cite{Liu17,yuan2020dual}, a dual space based hard-thresholding method has been proposed by the authors. In addition,  a screening method has been proposed by~\cite{atamturk2020safe}.  Following coordinate descent~(CD) methods~\citep{breheny2011coordinate,mazumder2011sparsenet,friedman2010regularization,nesterov2012efficiency} for linear regression problems,~\cite{Hazimeh18} proposed an efficient CD based method to scale up the solutions of problem~\eqref{eq:primal1}. Their method can be improved with the proposed switch techniques that aim to escape from local solutions. 
Additionally, the combination of $\ell_0$ and $\ell_q$ ($q = 1$ or 2) are considered in many existing literature. See~\cite{liu2007variable, soubies2017unified} for theoretical analyses and details.

Following the studies in~\cite{pilanci2015sparse,Liu17,yuan2020dual},  we investigate the dual form of the generalized sparse problem~\eqref{eq:primal1}.  Under mild conditions, a strong duality theory has been established for problem~\eqref{eq:primal1}. A primal-dual algorithm~is proposed to further improve the efficiency and quality of solutions by leveraging the exploration in the~dual space along with coordinate screening and active incremental techniques~\citep{Fercoq2015,GAP,Ndiaye2017,atamturk2020safe,Celer,ren2020thunder}.  In summary, our contributions on theoretical side are three-folds. We first establish the duality theorem for the generalized non-convex sparse learning problem~\eqref{eq:primal1}. We second demonstrate that the derived duality allows us to adopt the screening and coordinate incremental strategies~\citep{Fercoq2015,GAP,Ndiaye2017,Celer,ren2020thunder} in $\ell_1$ solvers to boost the efficiency of the proposed  $\ell_0$ algorithm.
Thirdly, we provide  theoretical study of the proposed algorithms. Our theoretical analysis shows that the generalized sparse problem~\eqref{eq:primal1} can be solved within polynomial computation complexity. 
Experiments on both synthetic and real-world datasets show the advantages of our~method.

The rest of paper is organized as follows. In Section~\ref{sec:duality}, we formulate the dual form of the generalized sparse learning problem  and also introduce the  duality theory.
In Section~\ref{sec:algorithm}, we propose the new primal-dual algorithm improved with coordinate incremental techniques. Section~\ref{sec:analysis} presents our algorithm analysis. 
Experimental results are provided in Section~\ref{sec:experiment}.
A discussion is given in Section~\ref{sec:discuss}, and the concluding remark is presented in Section~\ref{sec:conclusion}.

\vspace{0.2in}

\noindent \textbf{Notation.}\ 
Symbol $\beta \in \mathbb R^p$ is used for the primal variable and $\alpha \in \mathbb R^n$ is for the dual variable. 
We use $\|\beta\|$, $\|\beta\|_0$ and $\|\beta\|_1$ to denote the $\ell_2$, $\ell_0$ and $\ell_1$ norm of $\beta$, respectively.
Functions $P(\beta)$ and $D(\alpha)$ represent the primal objective and the dual objective correspondingly. For matrix $X$, $\sigma_{max}(X)$/$\sigma_{min}(X)$ denotes its largest/smallest singular value. $\mathrm{supp}(\beta)$ is the support set of vector $\beta$, i.e. $\mathrm{supp}(\beta)=\{j | \beta_j \neq 0\}$. $S^c$ represents the complement of set $S$.  

\newpage

\section{Duality Theory of Sparse Learning}\label{sec:duality} 

This section extends the duality studies  in~\cite{pilanci2015sparse,Liu17,yuan2020dual} to the generalized sparse learning problem~\eqref{eq:primal1}. 
Let $X =[x_1, x_2, ..., x_n]^{\top}$ be the feature matrix, $y= [y_1, y_2, ..., y_n]^{\top}$ be the response vector and $n$ is the number of samples. 
Let $l^*_i(\alpha_i) = \max_{u\in \mathcal{F}} \{\alpha_i u - l_i(u)\}$ be the Fenchel conjugate~\citep{fenchel1949conjugate} of convex loss function $l_i(u)$ and $\mathcal{F} \subseteq \mathbb{R} $ be the feasible set of $\alpha_i$ regarding $l^*_i()$. According to the expression $l_i(u) = \max_{\alpha_i \in \mathcal{F}}\{\alpha_i u - l_i^*(\alpha_i)\}$, the primal problem can be reformulated into 
\begin{align} \label{eq:Lminmax}
\min_{\beta} \sum_{i=1}^n \max_{\alpha_i \in \mathcal{F}} &\big(\alpha_i \beta^{\top}x_i - l_i^*(\alpha_i) \big) + \lambda_0 ||\beta||_0    +  \lambda_1 ||\beta||_1 + \lambda_2 ||\beta||^2 .
 \end{align}
We use $L(\beta, \alpha)$  to represent the following objective
\begin{eqnarray} ~\label{eq:lagrangian}
L(\beta, \alpha) =  & \sum_{i=1}^n \big(\alpha_i \beta^{\top}x_i - l_i^*(\alpha_i) \big) + \lambda_0 ||\beta||_0    +  \lambda_1 ||\beta||_1 + \lambda_2 ||\beta||^2 .
 \end{eqnarray}
 Similar to the studies in~\cite{Liu17,yuan2020dual},  the  RIP (restricted strong condition number) bound conditions are not  explicitly required here.  Without  specifying $k$ in~\eqref{eq:hard_k}, our  duality theory is close to the standard duality paradigm, and thus generic primal-dual methods  may be used to further improve the solvers.

\subsection{Strong Duality}

We establish the  duality theory that guarantees the original non-convex  in~\eqref{eq:primal1} can be  solved in a dual space.  Following~\cite{yuan2020dual,Liu17}, we define the saddle point for the Lagrangian~\eqref{eq:lagrangian} of the generalized sparse learning~\eqref{eq:primal1}. 

\begin{definition} 
(Saddle Point). A pair $(\bar{\beta}, \bar{\alpha} ) \in \mathbb{R}^p \times \mathcal{F}^n$ is said to be a  saddle point for $L$~\eqref{eq:lagrangian}  if  the following holds 
 \begin{align}
L(\bar{\beta}, \alpha) \leq L(\bar{\beta}, \bar{\alpha})  \leq L( \beta, \bar{\alpha}) . \label{eq:saddlepoint}
\end{align}
\end{definition}
Different from the sparse saddle point in~\cite{yuan2020dual,Liu17} that requires $k$-sparse regarding the primal variable,  the saddle point defined here can be taken as a generalized saddle point.  Given $\alpha \in \mathcal{F}^n$, we further define  $\eta(\alpha) := - \frac{1}{2\lambda_2} \sum_{i=1}^n \alpha_i x_i  = - \frac{1}{2\lambda_2}  X^{\top}\alpha $, $\eta_0 := \frac{ 2\sqrt{\lambda_0 \lambda_2} +  \lambda_1}{2\lambda_2}$, and 
\begin{align} \label{eq:B_frak}
\beta_j(\alpha):=  \mathfrak{B}(\eta_j(\alpha))  :=  \begin{cases}
\mathrm{sign}\big(\eta_j(\alpha)\big) \big(|\eta_j(\alpha)| -\frac{\lambda_1 }{2\lambda_2} \big) &\text{if} \  \    |\eta_j(\alpha)|  > \eta_0 \\
\big\{0, \mathrm{sign}\big(\eta_j(\alpha)\big) \big(|\eta_j(\alpha)| -\frac{\lambda_1 }{2\lambda_2} \big) \big\}  &\text{if}  \  \  |\eta_j(\alpha)|  = \eta_0  \\
0  &\text{if}  \  \   |\eta_j(\alpha)|  < \eta_0
\end{cases}
.
\end{align}
Moreover, we define 
\begin{align}\label{eq:Psi_1}
&\Psi(\eta_j(\alpha); \lambda_0,  \lambda_1, \lambda_2)   
:=  \begin{cases}
- \lambda_2 (|\eta_j(\alpha)| - \frac{\lambda_1}{2 \lambda_2})^2 + \lambda_0   &\text{if} \  \     |\eta_j(\alpha)|  > \eta_0  \\
\big\{0, - \lambda_2 (|\eta_j(\alpha)| - \frac{\lambda_1}{2 \lambda_2})^2 + \lambda_0  \big\} \quad  &\text{if}  \  \    |\eta_j(\alpha)|  = \eta_0 \\
0  &\text{if}  \  \   |\eta_j(\alpha)|  < \eta_0
\end{cases} .
\end{align}

\newpage

The following  theorem establishes the duality theory regarding the generalized sparse problem~\eqref{eq:primal1}. 

\begin{theorem}\label{Thm:minimax}
Let $\bar{\beta} \in \mathbb{R}^p $ be a primal vector and $\bar{\alpha} \in \mathcal{F}^n$  regarding L, then

\begin{enumerate} 
\item $ (\bar{\alpha}, \bar{\beta})$ is a  saddle point of $L$ if and only if the following conditions hold:
\begin{enumerate}
\item $\bar{\beta}$ solves the primal problem; 
 
\item  $\bar{\alpha} \in [\partial l_1(\bar{\beta}^{\top}x_1),  \partial l_2(\bar{\beta}^{\top}x_2), ..., \partial l_n(\bar{\beta}^{\top}x_n)]^{\top}$;
\item
$\bar{\beta}_j =   \mathfrak{B}(\eta_j(\bar{\alpha}))$.
\end{enumerate}

 \item The mini-max relationship 
\begin{align} \label{eq:th2}
\max_{\alpha \in \mathcal{F}^n} \min_{\beta} L(\beta, \alpha) = \min_{\beta} \max_{\alpha \in  \mathcal{F}^n} L(\beta, \alpha) .
\end{align}
holds if and only if there exists a saddle point $(\bar{\beta}, \bar{\alpha})$ for L. 

\item The corresponding dual problem of~\eqref{eq:primal1} is  written as
\begin{align} \label{eq:dual}
  &\max_{\alpha \in \mathcal{F}^n} D(\alpha) =  \max_{\alpha \in \mathcal{F}^n}-  \sum_{i=1}^n  l_i^*(\alpha_i)   + \sum_{j=1}^p \Psi(\eta_j(\alpha); \lambda_0,  \lambda_1, \lambda_2),
\end{align} 
where  $l^*$ is the conjugate function of $l$.  The primal dual link is written as $\beta_j(\alpha)  =  \mathfrak{B}(\eta_j(\alpha))$.

\item (Strong duality) $\bar{\alpha}$ solves the dual problem in~\eqref{eq:dual}, i.e., $D(\bar{\alpha}) \geq D(\alpha), \alpha \in \mathcal{F}^n$, and $P(\bar{\beta}) = D(\bar{\alpha})$ if and only if the pair $(\bar{\beta}, \bar{\alpha})$ satisfies the  three conditions given by (a)$\sim$(c).

\end{enumerate}
\end{theorem}

Here $\eta_0$ is the threshold which controls the sparsity of the solution. Larger $\lambda_0$ and $\lambda_1$ lead to  sparser estimator. 
The mini-max result in  Theorem~\ref{Thm:minimax}-1 gives the sufficient and necessary conditions to guarantee the existence of a  saddle point for the Lagrangian.  Theorem~\ref{Thm:minimax}-2 can be used to  establish the duality theory, and it is the min-max side of the problem, and it  provides sufficient and necessary conditions under which one can safely exchange a min-max for a max-min regarding $L$~\eqref{eq:lagrangian}. 
\begin{remark}
Applying Theorem~\ref{Thm:minimax}, we have the following mini-max relationship 
\begin{align} \label{eq:rmk1}
\max_{\alpha \in \mathcal{F}^n} \min_{\beta} L(\beta, \alpha) = \min_{\beta} \max_{\alpha \in  \mathcal{F}^n} L(\beta, \alpha) .
\end{align}s
holds if and only if there exists a primal vector $\bar{\beta} \in \mathbb{R}^p$ and a dual vector $\bar{\alpha} \in \mathcal{F}^n$ such that conditions (a) $\sim$  (c) in Theorem\ref{Thm:minimax}-1 are satisfied.
Moreover, by calculations, it can be checked that \eqref{eq:rmk1} holds automatically for $ l(\cdot)$ being the square loss function.
\end{remark}


We use $P(\beta) = F(\beta)$ to represent the primal objective, and $D(\alpha)$ for the dual objective given in~(\ref{eq:dual}). 
Theorem~\ref{Thm:minimax}-3 indicates that the dual objective function $D(\alpha)$ is concave and the following remark explicitly gives the expression of its super-differential.
\begin{remark}\label{rmk:sup-grad}
The super-differential of  the dual form~\eqref{eq:dual} at $\alpha$ is given by 
$\nabla D(\alpha) = X\beta(\alpha) - l^{*'}(\alpha) = [\beta(\alpha)^{\top} x_1 -   {l^{*}_1}^{'}(\alpha_1), ..., \beta(\alpha)^{\top} x_n - {l^{*}_n}^{'}(\alpha_n )]^{\top}$.
\end{remark}

\newpage

The super-gradient can be alternatively derived through the partial derivative of the Lagrangian L~\eqref{eq:lagrangian} regarding $\alpha$. The sparse strong duality theory in Theorem~\ref{Thm:minimax}-4 gives the sufficient and necessary conditions under which the optimal values of the primal and dual problems~coincide. We define the  primal-dual gap  as
\begin{align}\label{eq:dgap}
\xi(\beta, \alpha) = P(\beta) - D(\alpha).
\end{align}
According to Theorem~\ref{Thm:minimax}-4, the  primal-dual gap reaches zero at the primal-dual pair $(\beta, \alpha)$ if and only if the conditions (a) $\sim$ (c) in Theorem~\ref{Thm:minimax}-1 hold. The  duality theory developed in this section suggests a natural way for finding the global minimum of the sparsity-constrained minimization problem in~\eqref{eq:primal1} via  primal-dual optimization methods. Let $(l^{*'})^{-1}$ be the inverse of $l^{*'}$, we have the following remark with $0 \in \nabla D(\bar{\alpha})$ at $\bar{\alpha}$.
\begin{remark}\label{rmk:dual_condt}
If $(\bar{\beta}, \bar{\alpha})$ satisfies the conditions in Theorem~\ref{Thm:minimax}-1, we have $\bar{\alpha} \in (l^{*'})^{-1}(X\bar{\beta}) \cap \mathcal{F}^n$. 
\end{remark}

Strong duality holds when both the primal and dual variables reach the optimal values. Before attaining the optimal values, the duality gap value can be bounded  by the current dual variable estimations. The closer the current estimation $\alpha$ and the  optimal value $\bar{\alpha}$ are, the smaller duality gap will be. Under special cases that the support of $\bar \beta$ is recovered, then the objective function becomes a convex function since $\|\beta^t\|_0$ remains a constant. Then strong duality holds naturally. 
In practice, the condition (a) of Theorem~\ref{Thm:minimax}-1 is hard to be satisfied because of the non-convexity of the primal problem. However, as long as $\beta$ reaches its optimal value, all the conditions in Theorem~\ref{Thm:minimax}-1 can be met because the dual problem is concave.

Different from~\cite{yuan2020dual,Liu17}, we study a generalized sparse problem. The methodology developed here can be easily extended to plain $\ell_1$ or $\ell_0$ problems (with  the $\ell_2$ term),  group sparse structures or fused sparse structures,  or even more  complex and  mixed sparse structures that we cannot or do not need to specify the active feature number $k$ value as in~\eqref{eq:hard_k}.

\subsection{Properties of Generalized Sparse Learning}\label{sec:sparse_prop}

In this paper, we study the duality of generalized sparse learning problem.  Based on the  strong duality of $\ell_0$ problem~\eqref{eq:primal1}, screening methods~\citep{Fercoq2015,GAP,Ndiaye2017} and coordinate increasing techniques~\citep{Celer,ren2020thunder} can be implemented to gain an improvement in algorithm efficiency. Following the Gap screening~\citep{Fercoq2015,GAP} for Lasso problem, we have the following theorem regarding the duality gap. 

\begin{theorem}\label{Thm:ball}
Assume that the primal loss functions $\{l_i(\cdot)\}_{i=1}^n$ are $1/\mu$-strongly smooth.  The range of the dual variable is  bounded via the duality gap value, i.e.,  $\forall \alpha \in \mathcal{F}^n, \beta \in \mathbb{R}^p$, $\{B(\alpha; r) : || \alpha -  \bar{\alpha}||_2  \leq r, r = \sqrt{ \frac{2(P(\beta) - D(\alpha)) } {\gamma}} \} $. 
Here $\gamma$ is a positive constant and $\gamma \geq \mu$. 
\end{theorem}

Let $x_{\cdot i}$ be the $i$th column of X, according to the definition of $\mathfrak{B}$ in~\eqref{eq:B_frak}, the activity of feature $j$ is determined by the magnitude of $\eta_j(\bar{\alpha})$, i.e.,  $|\eta_j(\bar{\alpha})| = \frac{1}{2\lambda_2} |x^{\top}_{\cdot j} \bar{\alpha}|$.  With the ball region estimation for $\bar{\alpha}$ in Theorem~\ref{Thm:ball}, we can estimate the activity of a feature with the value of current $\alpha$. Let  $r= \sqrt{ \frac{2(P(\beta) - D(\alpha)) } {\gamma}}$ be the radius of the estimated ball range for $\bar{\alpha}$ using current $\alpha$ and $\beta$ solutions. Then 
$\big||x^{\top}_{\cdot j} \alpha| - ||x_{\cdot j}||_2 r\big| \leq |x^{\top}_{\cdot j} \bar{\alpha}| \leq |x^{\top}_{\cdot j} \alpha| + ||x_{\cdot j} ||_2 r $ and
 we get $|\eta_i(\bar{\alpha})| \leq  \frac{1}{2\lambda_2}\big(|x^{\top}_{\cdot j} \alpha| + ||x_{\cdot j} ||_2 r\big) <\eta_0  \implies x_{\cdot j}  \  \text{is an inactive feature}.$
It implies
\begin{align*}
&\text{Upper Bound}: |x^{\top}_{\cdot j} \alpha| + ||x_{\cdot j} ||_2 r  <2\lambda_2 \eta_0 = 2\sqrt{\lambda_0 \lambda_2} +  \lambda_1   \implies j \notin \mathrm{supp}(\bar{\beta}), \\
&\text{Lower Bound}: \big| |x^{\top}_{\cdot j} \alpha| - ||x_{\cdot j} ||_2 r \big| > 2\lambda_2 \eta_0 = 2\sqrt{\lambda_0 \lambda_2} +  \lambda_1  \implies j \in \mathrm{supp}(\bar{\beta}). 
\end{align*}
According to the derived dual objective~\eqref{eq:dual} and equations~\eqref{eq:B_frak}-\eqref{eq:Psi_1}, a feature's activity is determined by its product with the optimal dual variable $\bar{\alpha}$, e.g., $\eta_j(\bar{\alpha})$ for feature $j$. The dual range estimation $\bar{\alpha}$ ($B(\bar{\alpha};r) $)  allows us to perform  feature screen in order to improve algorithm efficiency by following the approach for Lasso~\citep{Fercoq2015,GAP,Ndiaye2017}. As the support set $S$ is unknown, we just set $\gamma = \mu$ to ensure the safety of the feature screening. Here safety means the screening operation does not remove any feature belonging to $S=\mathrm{supp}(\bar{\beta})$. The framework proposed in this paper lays a broader bridge between screening methods and the solutions of $\ell_0$ regularized problems.

\section{Algorithm}\label{sec:algorithm}

With the mini-max relationship in Theorem~\ref{Thm:minimax} regarding Lagrangian form $L(\beta, \alpha)$~\eqref{eq:lagrangian} , we first develop a primal-dual algorithm to update both $\alpha$ and $\beta$. The dual objective $D(\alpha)$ is a non-smooth function as the term $\Psi()$ regarding $\alpha$ is non-smooth due to the truncation operation.  We focus on the following simplified dual form 
\begin{align} \label{eq:dual2}
   &\max_{\alpha \in \mathcal{F}^n} -  \sum_{i=1}^n  l_i^*(\alpha_i)   + \sum_{j=1}^p \Psi(\eta_j(\alpha); \lambda_0,  \lambda_1, \lambda_2),  \ \text{and}\\
 \label{eq:psi}
&\Psi(\eta_j(\alpha); \lambda_0,  \lambda_1, \lambda_2) 
 =   \begin{cases}
- \lambda_2 (|\eta_j(\alpha)| - \frac{\lambda_1}{2 \lambda_2})^2 + \lambda_0   &\text{if} \  \     |\eta_j(\alpha)|  \geq \eta_0 \\
0  &\text{if}  \  \   |\eta_j(\alpha)|  < \eta_0
\end{cases} .
\end{align}
The primal dual link is 
\begin{align} \label{eq:bigB}
&  \beta_j(\alpha)  =    \mathfrak{B}(\eta_j(\alpha)) = \begin{cases}
\mathrm{sign}\big(\eta_j(\alpha)\big) \big(|\eta_j(\alpha)| -\frac{\lambda_1 }{2\lambda_2} \big) &\text{if} \  \    |\eta_j(\alpha)|  \geq \eta_0 \\
0  &\text{if}  \  \   |\eta_j(\alpha)|  < \eta_0
\end{cases}
.
\end{align}
The super-gradient  regarding the dual variable  $\nabla_{\alpha} D(\alpha) =  X\beta(\alpha) -{l^*}^{'}(\alpha)$ can be taken as the partial derivative of $L(\beta, \alpha)$. We  give the dual problems of two objective functions in the supplements, and we will focus on linear regression to present the proposed algorithms.

\subsection{Primal-dual Updating for Linear Regression}

 We use linear regression as an example to illustrate the proposed primal-dual inner solver of $\ell_0$ regularized problems.
For least square problem, the primal form is
\begin{align*}
&\min_{\beta \in \mathbb{R}^p} P(\beta) = f(\beta) + \lambda_0 ||\beta||_0, \ \text{and} \ f(\beta) =   \frac{1}{2}|| y - X\beta||^2_2   +  \lambda_1 ||\beta||_1 +\lambda_2 ||\beta||_2^2.
\end{align*}
For least square problem, $l_i(u; y_i) = \frac{1}{2}(y_i - u)^2$, then
$l_i^*(\alpha_i)=  \frac{1}{2}\big((y_i + \alpha_i)^2 - y_i^2 \big) =\frac{1}{2} \alpha_i^2 +y_i\alpha_i $. Here $u= x_i^{\top} \beta$. Thus the dual problem is
\begin{align} \label{eq:dual3}
\max_{\alpha}  D(\alpha) =   & - \frac{1}{2}\alpha^{\top}\alpha  - y^{\top}\alpha  +  \sum_{j=1}^p \Psi( - \frac{1}{2\lambda_2} \sum_{i=1}^n \bar{\alpha}_i x_i; \lambda_0, \lambda_1, \lambda_2).
\end{align}
The corresponding super-gradient can be easily computed, i.e.,
$g_{\alpha} =  \nabla_{\alpha} D(\alpha) = X\beta - Y - \alpha.$ After the super gradient ascent for the dual variables, we apply the primal-dual link function to get the variable in the primal space.

The dual objective $D(\alpha)$ is non-smooth. The super-gradient $g_{\alpha} =  \nabla_{\alpha} D(\alpha)$ can be improved with a more accurate primal variable estimation for the Lagrangian form~\eqref{eq:lagrangian} with the mini-max relationship in Theorem~\ref{Thm:minimax}.  
We use coordinate descent~(CD)~\citep{Hazimeh18} to improve the estimation of  primal variable as
\begin{align} \label{eq:thresh}
&\beta_j = T(\mathbf{\beta}; \lambda_0, \lambda_1, \lambda_2)   =  \begin{cases}
\mathrm{sign}(\tilde{\mathbf{\beta}}_j) \frac{|\tilde{\mathbf{\beta}}_j| - \lambda_1}{1 + 2\lambda_2} \quad  &\text{if} \ \frac{|\tilde{\mathbf{\beta}}_j| - \lambda_1}{1 + 2\lambda_2}  \geq \sqrt{\frac{2\lambda_0}{1+2\lambda_2}} \\ 
0 \quad  &\text{if} \ \frac{|\tilde{\mathbf{\beta}}_j| - \lambda_1}{1 + 2\lambda_2}  < \sqrt{\frac{2\lambda_0}{1+2\lambda_2}} 
\end{cases} .
\end{align}

Here $\tilde{\mathbf{\beta}}_j = (y - X\mathbf{\beta})^{\top} x_{\cdot j} + \mathbf{\beta}_j  x_{\cdot  j}^{\top}x_{\cdot j}$, and $x_{\cdot j}$ is the $j^{th}$ column of $X$, and it is also named the $j^{th}$  feature. The operation $T()$ always decreases the primal objective, i.e. with $\beta_b = T(\beta_a)$, we always have $P(\beta_b) \leq P(\beta_a)$, and hence a smaller duality gap.

\begin{algorithm}[h]
\caption{Inner solver with primal-dual  updating  }\label{alg:primal_dual_inner} 
	\KwInput{data $\{X, y\}$; $\lambda_0$, $\lambda_1$, $\lambda_2$; step size $\omega_t$ at step $t$; initial $\mathbf{\alpha}^0$,  $\mathbf{\beta}^0$}
	\KwResult{ $\mathbf{\alpha}^t$, $\mathbf{\beta}^t$ \vspace{-0.1in}}
	\hrulefill\\	
	 $t\leftarrow 0$;\\ 
	\While{$\overline{DGap}$ decreasing}{
	\textcolor{blue}{//Super-gradient } \\
	 $g_{\mathbf{\alpha}}^{t} =[{\beta^{t-1}}^{\top} x_1-  {l^{*}_1}^{'}(\alpha^{t-1}_1), ..., {\beta^{t-1}}^{\top} x_n - {l^{*}_n}^{'}(\alpha^{t-1}_n ) ]^{\top}$; \\
	  \textcolor{blue}{//Dual ascent with feasible projection } \\
	 $\alpha^t  =\mathcal{P}_{\mathcal{F}}( \alpha^{t-1} + \omega_t g_{\mathbf{\alpha}}^{t})$ ;  \\
	 $\eta(\alpha^t) = - \frac{1}{2\lambda_2} \sum_{i=1}^n \alpha^t_i x_i$; \\
	 \textcolor{blue}{//Primal-dual relation} \\
	 $ \mathbf{\beta}^t  =   \mathfrak{B}(\eta_j(\alpha^t); \lambda_0, \lambda_1, \lambda_2 ) $;  \\
	  \textcolor{blue}{//Primal coordinate descent} \\
	 $ \mathbf{\beta}^t = T(\beta^t; \lambda_0, \lambda_1, \lambda_2 ) $; \\
	 Compute duality gap $\overline{DGap}$ with $\beta_t$ and $\alpha_t$ ; \\
	 $t\leftarrow t+1$; 
	}
\end{algorithm}

 The proposed primal-dual updating procedure is given by Algorithm~\ref{alg:primal_dual_inner}. 
The primal coordinate descent $T()$ improves the solution from primal-dual relation $ \mathfrak{B}$. $\omega_t$ is the step size at $t$, and should be decreasing with $t$.  We use Algorithm~\ref{alg:primal_dual_inner} as the backbone solver in our primal-dual algorithm, and  $\overline{DGap}$ is the  sub-problem's duality gap    achieved by the inner solver.

Moreover, according to Remark~\ref{rmk:dual_condt}, the optimal dual $\bar{\alpha} \in (l^{*'})^{-1}(X\bar{\beta}) \cap \mathcal{F}^n$. For a   solver in primal space, we can use this equation to find a point in dual space and then compute the duality gap to evaluate current solution. For linear regression, with $\beta$ we have $\alpha = X\beta -y $, then we can compute the duality gap via~\eqref{eq:dgap}.

\subsection{Improve Efficiency with Active Incremental Strategy}\label{sec:increment_strategy}

For sparse models, most of the features are redundant and they  incur extra computation costs.  The derived dual problem structure and the duality property provide an approach to implement feature screening~\citep{Fercoq2015,GAP,Ndiaye2017} and feature active incremental strategy~\citep{Celer,ren2020thunder}. According to the analysis in Section~\ref{sec:sparse_prop}, the activity of a feature $x_{\cdot j}$ depends on the value of $\eta_j(\bar{\alpha})$, i.e.,  $|\eta_j(\bar{\alpha})| = \frac{1}{2\lambda_2} |x^{\top}_{\cdot j} \bar{\alpha}|$. We use the current estimation range of $\bar{\alpha}$, i.e., $B(\bar{\alpha}; r^s)$ to approximate the value of $\eta_j(\bar{\alpha})$. Here $s$ is the step number in the outer loop of the algorithm, and $\beta^s$ and $\alpha^s$ are the primal dual solutions at step $s$. The ball radius~$r^s= \sqrt{ \frac{2(P(\beta^s) - D(\alpha^s)) } {\gamma}}$~for $\bar{\alpha}$ depends the duality gap at step $s$ according to Theorem~\ref{Thm:ball}.

\begin{algorithm}[b!]
\caption{Algorithm for Feature Adding}\label{alg:add}
\KwInput{$\alpha^s$, $r^s$, $\mathcal{A}^s$, $\mathcal{R}^s$ } 
	\KwResult{ $\mathcal{R}^{s+1}$, $\mathcal{A}^{s+1}$  }
	\vspace{-0.1in}
	\hrulefill \\
	Set  $h = \lceil c\log(p) \rceil$, \, $\mathcal{H}\leftarrow$ Select $h$ features according to the descending order of $x_{\cdot j}^{\top}\alpha^s, j\in \mathcal{R}^s$;
	
	$\mathcal{A}^{s+1} \leftarrow \mathcal{A}^{s} \cup \mathcal{H}$, \  \ 
	$\mathcal{R}^{s+1} \leftarrow \mathcal{R}^{s} \setminus \mathcal{H}$ ; \\
\end{algorithm}

\begin{algorithm}[b!]
	 \caption{Primal-dual algorithm \vspace{-0.05in} }\label{alg:imprv_primal_dual}
\KwInput{ $X$, $y$; $\lambda_0$, $\lambda_1$, $\lambda_2$; $\xi$ }
\KwResult{ $\beta^s$ }
\vspace{-0.1in}
\hrulefill \\
     Choose a small set of features $\mathcal{A}_0$  from $X$ in the descending order of $|X^{\top}l'(\mathbf{0})|$, and $\mathcal{R}_0$ represents the rest features; \\ $s \leftarrow 0$; $\beta^0 \leftarrow 0$; $\alpha^0 \leftarrow 0$; $DoAdd = True$; \\
	\While{\textbf{True}}
	{
    	\textcolor{blue}{//Sub-problem solver} \\
	    $\alpha^s, \tilde{\beta}^s \leftarrow$ Solve the sub-problem with feature set $\mathcal{A}^s$   via Algorithm~\ref{alg:primal_dual_inner}; \\  
	    $\beta^s \leftarrow$ put $\tilde{\beta}^s$ in a size $p$ vector and set entries not in $\mathcal{A}^s$  zero;
	    
		\textcolor{blue}{//Dual range estimation} \\
		  Compute  the duality gap $DGap$ and  the ball region $B(\alpha^s; r^s)$ with $r^s= \sqrt{ \frac{2(P(\beta^s) - D(\alpha^s)) } {\gamma}}$ ; 
		  
		\If{$ DGap <\xi$  }
		{
		  \textbf{Stop}; \textcolor{blue}{//Algorithm exits }
		    }
		Feature Screening with $B(\alpha^s; r^s)$ and~\eqref{eq:screen_rule};  \\
			\If{ $DoAdd$ }
    		{\textcolor{blue}{//Feature Adding} \\
    		    \If{ $\max_{j \in \mathcal{R}^s} |x^{\top}_{\cdot j} \alpha^s| + \|x_{\cdot j} \|_2 r^s  <2\lambda_2 \eta_0$ }
        		{
        		    $DoAdd = False$; Continue;
        		}
    	    	Feature Adding with $\alpha^s$ and Algorithm~\ref{alg:add}, and update $\mathcal{A}^s$ and $\mathcal{R}^s$; 
    		 }
    	 $s\leftarrow s+1$;  
}
\end{algorithm}

The proposed primal-dual algorithm for $\ell_0$ is given by Algorithm~\ref{alg:imprv_primal_dual}.  Algorithm~\ref{alg:imprv_primal_dual} starts with a small active set $\mathcal{A}$, and then increase the active set's size after solving each sub-problem. We use $\mathcal{R}$ to represent the set of features  not used by the sub-problem solver. The feature adding algorithm is given by Algorithm~\ref{alg:add}. Moreover, we can derive a gap-screening algorithm~\citep{Fercoq2015,GAP} by using the upper bound of $|\eta_j(\bar{\alpha})|$'s approximation given in Section~\ref{sec:sparse_prop}.
Based on the derivation in Section~\ref{sec:sparse_prop}, we use the following safe principle for feature screening. 
\begin{align} \label{eq:screen_rule}
&\textbf{Feature Screening:}\ |x^{\top}_{\cdot j} \alpha^s| + ||x_{\cdot j} ||_2 r^s  < 2\lambda_2 \eta_0 \implies j \notin \mathrm{supp}(\bar{\beta}), \ \text{remove}. 
\end{align}
 Here $2\lambda_2 \eta_0 = 2\sqrt{\lambda_0 \lambda_2} +  \lambda_1$.  The screening rule is safe because it is derived based on concavity of the dual problem. Base on~\eqref{eq:screen_rule}, we  derive a stopping condition for feature adding. If all features in $\mathcal{R}$ satisfy~\eqref{eq:screen_rule}, we stop Feature Adding. This allows us to avoid redundant computation  resulted from some inactive features.

 In Algorithm~\ref{alg:imprv_primal_dual}, the initialization values of $\beta^0$ and $\alpha^0$ are set to zero.  We use $DGap$ to represent the duality gap of the original problem attained by the primal-dual algorithm. Please note that we use $s$ to denote the step in Algorithm~\ref{alg:add} and Algorithm~\ref{alg:imprv_primal_dual}, in order to differ from steps ($t$s) in Algorithm~\ref{alg:primal_dual_inner}.  Empirically, feature screening does show  power to improve algorithm efficiency. But the feature active incremental strategy can significantly avoid redundant computation to achieve~the~target~duality~gap.

\section{Algorithm Analysis}\label{sec:analysis}

We discuss the algorithm convergence in this section. Firstly, we present theoretical results on the convergence and support recovery of the inner solver for sub-problems. 

\subsection{Algorithm Analysis for Inner Solver}

With the conditions in Theorem~\ref{Thm:minimax}-1, we can quantify the duality gap and develop the algorithm analysis. Different from the primal updating~\citep{Hazimeh18} or dual updating~\citep{Liu17} algorithms,   Algorithm~\ref{alg:primal_dual_inner} has both primal and dual updating steps.
Let p denotes the size of the input feature set $X$ of the sub-problem in Algorithm~\ref{alg:primal_dual_inner},  $m_1 =\max_{j:1\leq j \leq p} |y^{\top}x_{\cdot j}|$,  $m_2 =\max_{j:1\leq j \leq p} ||X^{\top}x_{\cdot j}||$, $m_3 = \max_{i, \alpha_i^t \in \mathcal{F}} |l_i^{*'}(\alpha_i^t)|$,  $\varrho = \sqrt{n}||\alpha^t||_{\infty} - \lambda_1$, and $\omega_t$ is the  decreasing step size. We have the following theorem regarding the convergence of Algorithm~\ref{alg:primal_dual_inner}.
\begin{theorem}\label{Thm:inner_convg}
Assume that $l_i$ is $1/\mu$-smooth,  $||x_{i}|| \leq \vartheta  \ \forall 1\leq i \leq n$, and $||x_{\cdot j}||= 1  \ \forall 1\leq j \leq p$.  
By choosing $w_t = \frac{1}{t \gamma}$, then the sequence generated by Algorithm~\ref{alg:primal_dual_inner} satisfies the following estimation error inequality:
 \begin{align*}
|| \alpha^t - \bar{\alpha}||^2 \leq c_1\bigg( \frac{1}{t} +  \frac{\ln t}{t}  \bigg).
\end{align*}
Here $c_1 = \frac{   c_0^2}{ \mu ^2 } $, $ c_0=  \frac{ \sqrt{np} \vartheta }{2\lambda_2(1+2\lambda_2)}(2\lambda_2m_1 +  \varrho +  \sqrt{n}m_2 \varrho - 2\lambda_1\lambda_2)+\sqrt{n}m_3$, $gamma$ is same as in Theorem~\ref{Thm:ball}. 
\end{theorem}

  We can prove the convergence of the primal variable using the results in Theorem~\ref{Thm:inner_convg}.  Let $S=\mathrm{supp}(\bar{\beta})$,  $N= \{j| \eta_j(\bar{\alpha}) = \eta_{0(j)} \}$, and $\bar{\delta} =2 \lambda_2 \min \bigg\{\min_{j:j \in S} \frac{|\eta_j(\bar{\alpha})  |  -  \eta_{0}}{||x_{\cdot j}||}, 
\min_{j:j \in S^c \setminus N} \frac{ \eta_{0} - |\eta_j(\bar{\alpha})|}{||x_{\cdot j}||}\bigg\}$. The following theorem gives the complexity for support recovery and duality gap convergence. 

\begin{theorem}\label{Thm:prim_conv}
Assume that $l_i$ is $1/\mu$-smooth,  $||x_{i}|| \leq \vartheta  \ \forall 1\leq i \leq n$, and $||x_{\cdot j}||= 1  \ \forall 1\leq j \leq p$. Let $ t_1 =\frac{3c_1}{\bar{\delta}^2}\log \frac{3c_1}{\bar{\delta}^2}$,  with $t > t_1$, we have $\mathrm{supp}(\beta(\alpha)) = \mathrm{supp}(\bar{\beta})$ and $|| \beta(\alpha) - \bar{\beta}|| \leq \frac{\sigma_{max}(X_S)}{2\lambda_2} ||\alpha - \bar{\alpha}||$. Moreover, let $t_2 = \frac{3c_1 c_2}{\epsilon} \log \frac{3c_1 c_2}{\epsilon}$,  $c_2 = c_0\bigg(1+ \frac{\sigma_{max}(X_S)}{2\mu\lambda_2} \bigg)$,  for any $\epsilon>0$ with $t >\mathrm{max}\{t_1, t_2\}$, we have $P(\beta^t) - D(\alpha^t) \leq \epsilon$.
\end{theorem}
Additional analysis on algorithms  can be found in the appendix.
Due to the large magnitudes of $c_0$ and $c_2$ and small value of $\bar{\delta}$, it could be time consuming for the solver to achieve very small duality gaps. 

\newpage

\subsection{Outer Loop Analysis}

The outer loop in Algorithm~\ref{alg:imprv_primal_dual} involves both feature screening and feature adding operations  relying on the dual variable estimation for the original problem, i.e., $|| \alpha^s -  \bar{\alpha}||_2  \leq r^s, r^s = \sqrt{ \frac{2(P(\beta^s) - D(\alpha^s)) } {\gamma}}$ defined in Theorem~\ref{Thm:ball}.  Here $\gamma = \mu + \frac{\sigma_{min}(X_S)}{2 \lambda_2}$. As discussed in  Section~\ref{sec:sparse_prop}, we set  $\gamma = \mu$ to ensure the safety of feature screening. 
\begin{remark}\label{rmk:out_loop}
The  screening operation~\eqref{eq:screen_rule} is safe, and it does not remove any features in $\mathcal{A}^s \cap \mathrm{supp}(\bar{\beta})$ at step $s$. With additional features added by the feature adding operation (Algorithm~\ref{alg:add}), the primal objective $P(\beta^s)$  always decreases after  the solution of the sub-problem regarding feature set $\mathcal{A}^s$ using the inner solver (Algorithm~\ref{alg:primal_dual_inner}).
\end{remark}
The  screening operation usually  keeps the primal objective value $P(\beta^s)$ intact. With Remark~\ref{rmk:out_loop}, $\mathcal{A}^s$ and  $P(\beta^s)$  converge  after some steps, and also Algorithm~\ref{alg:imprv_primal_dual}  converges with  $DGap$ smaller than a given threshold $\xi$. In fact, the active incremental strategy  could significantly reduce redundant operations introduced by inactive features especially when the problem is with high sparse level~\citep{Celer,ren2020thunder}. The solution sparse level (the size of $S=\mathrm{supp}(\bar{\beta})$) impacts the algorithm complexity.

\section{Experiments}\label{sec:experiment}

Experiments focus on linear regression. Our proposed algorithm can  be extended to other forms of loss~function. Via experimental studies, we show the effectiveness of our method by comparing with dual iterative hard thresholding~\citep{Liu17} and coordinate descent with spacer steps~\citep{Hazimeh18} algorithms, which is the state-of-art  for $\ell_0$ regularization problem. The experimental environment is CPU: Intel Xeon Platinum8168@2.70 GHz; OS: Windows Server~2012~R2.

\subsection{Simulation Study}

In this study, we simulate the datasets under the linear regression setting, i.e., $\mathbf{y} = X\beta + \epsilon$. The data matrix is generated according to a multi-variate Gaussian $X_{n\times p}\sim \textrm{MVN}(0, \Sigma)$, and $\Sigma = (\sigma_{i,j})$. Exponential correlation~\citep{Hazimeh18} is utilized to control feature relationship, i.e., $\sigma_{ij} = \rho^{|i-j|}$ with $\rho = 0.4$. The noise $\epsilon$ is Gaussian white noise with $SNR=\{2, 5, 20\}$ . For the true parameter $\beta$, $3\%$ entries ($0.03p$) are randomly set to the values in $[-1.0, 1.0]$, and the rest ($0.97p$) are set to zero. We generate the datasets with $p=3\,000$  and $n$ varying 
in $\{200,  300, 400, 500, 600\}$. Each setting is replicated for 50 times.

\begin{figure*}[]

\mbox{\hspace{-0.0in}
\includegraphics[width=2.2in]{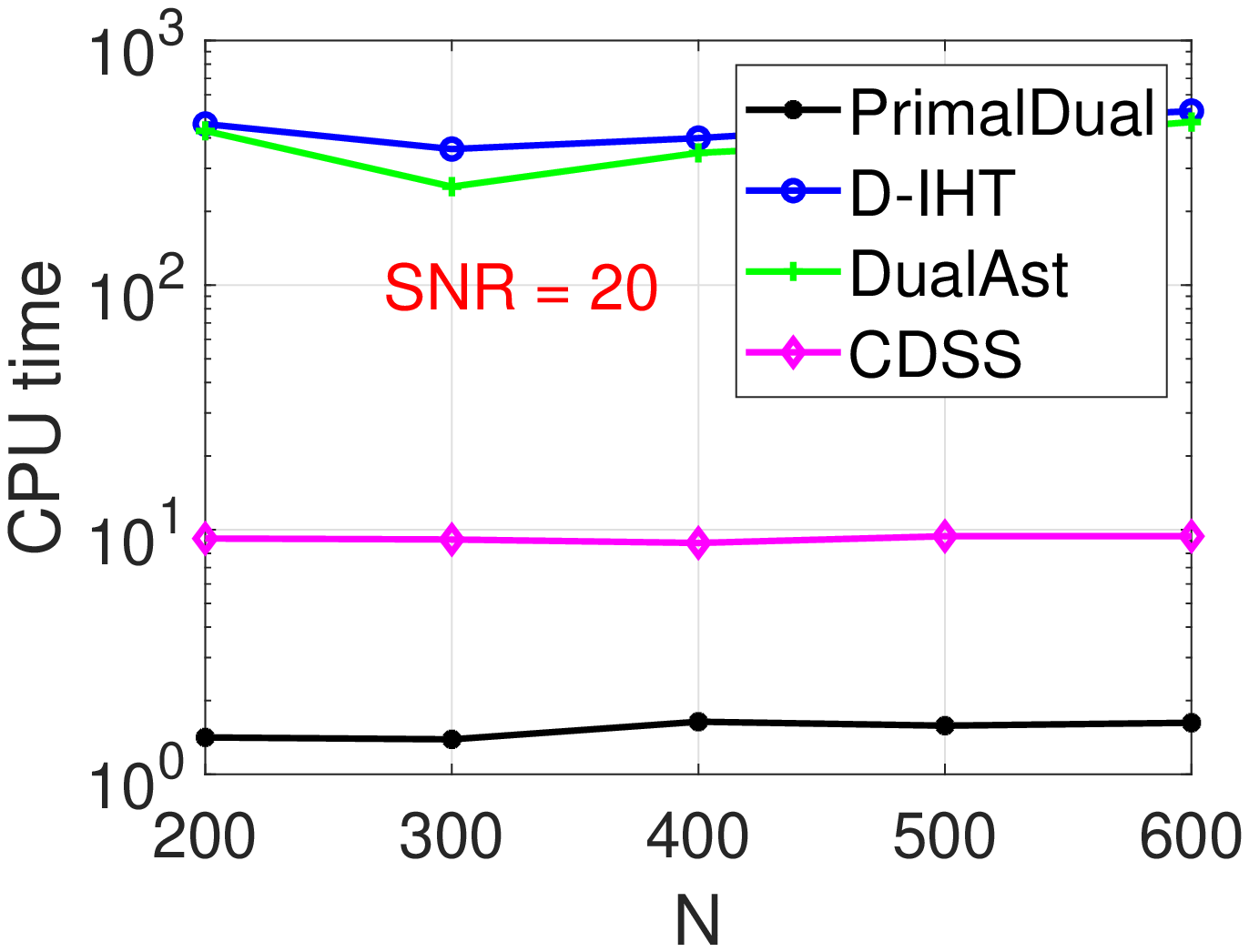}
\includegraphics[width=2.2in]{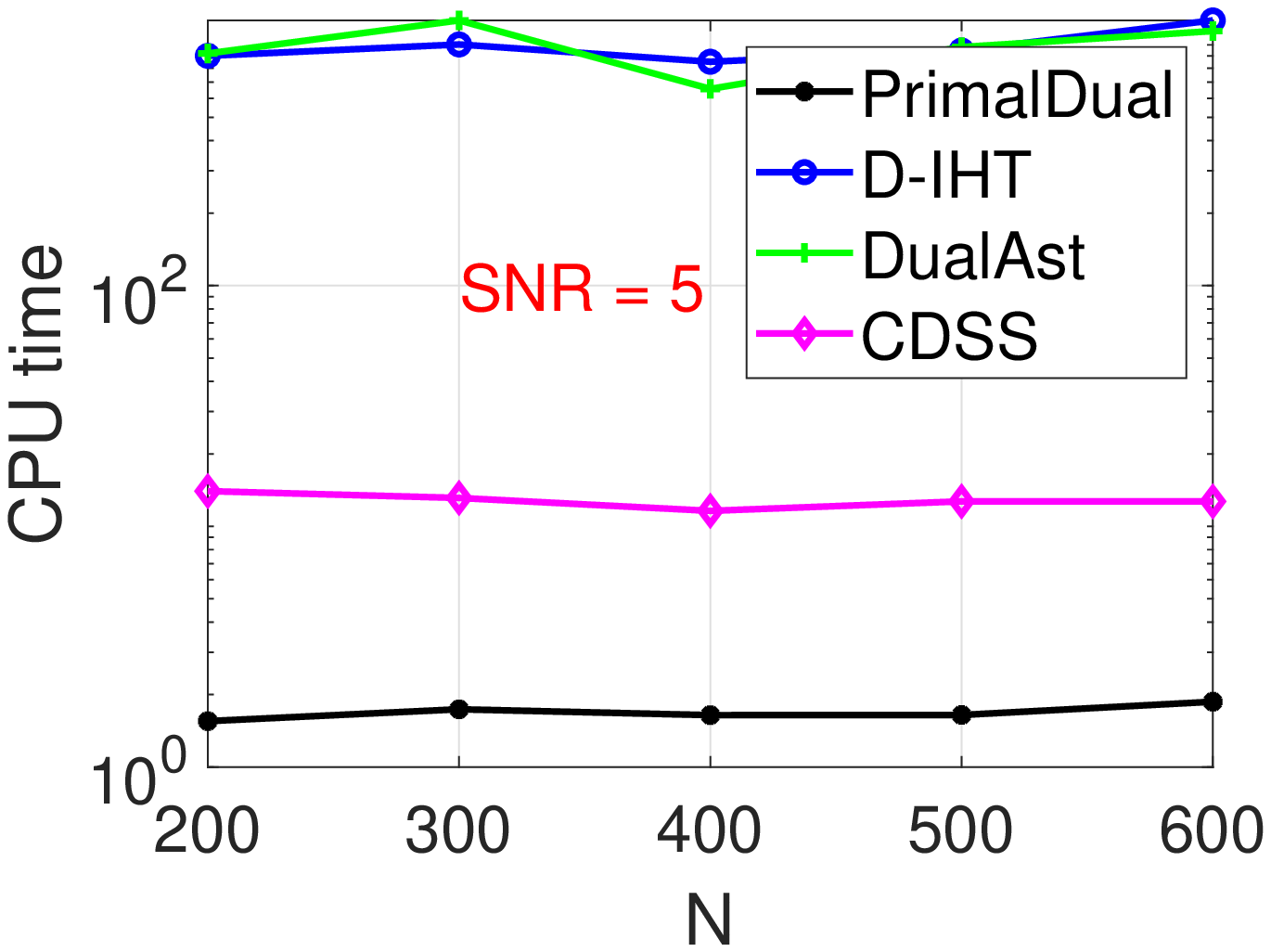} \includegraphics[width=2.2in]{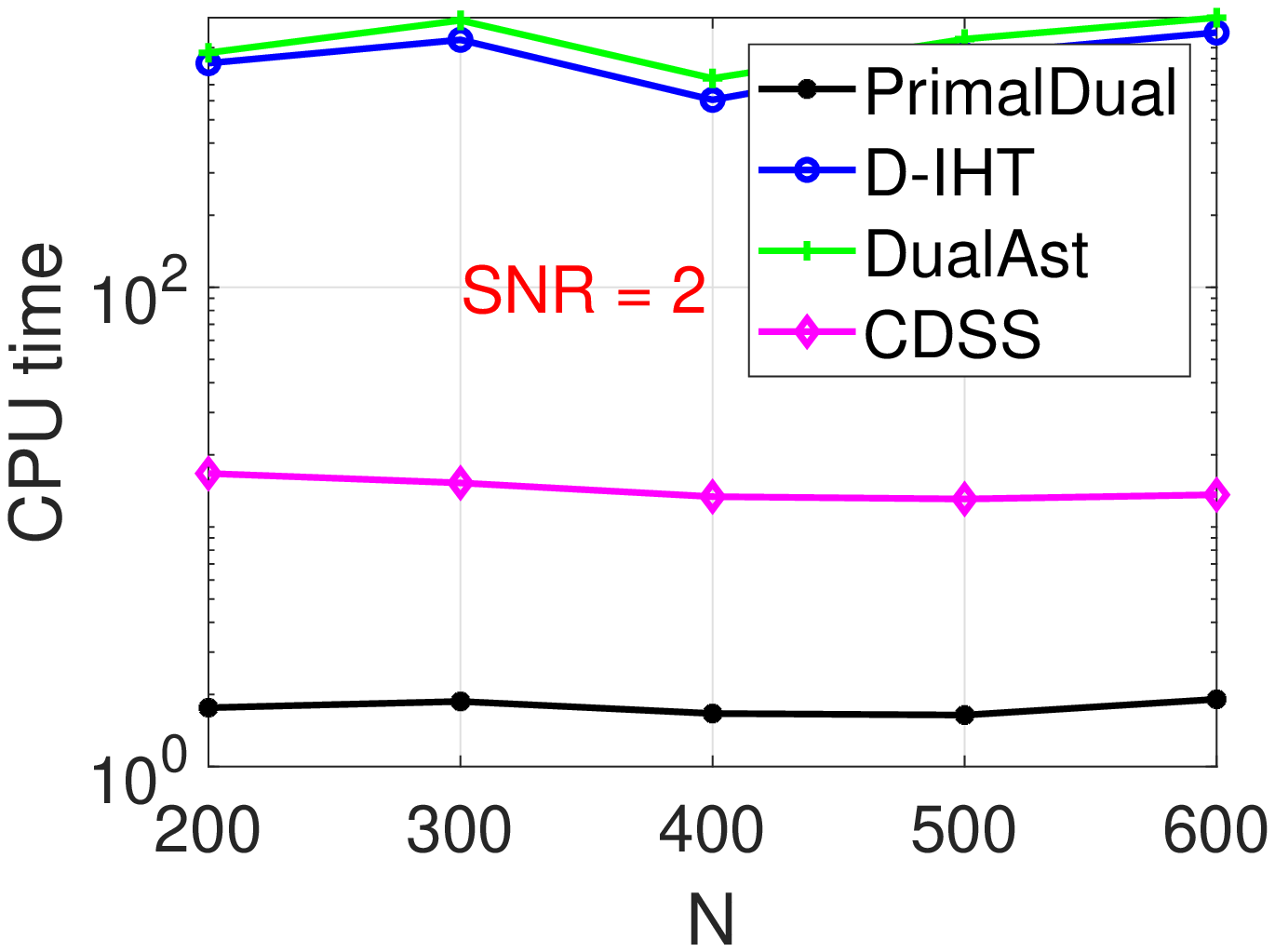} 
}

\mbox{\hspace{-0.0in}
\includegraphics[width=2.2in]{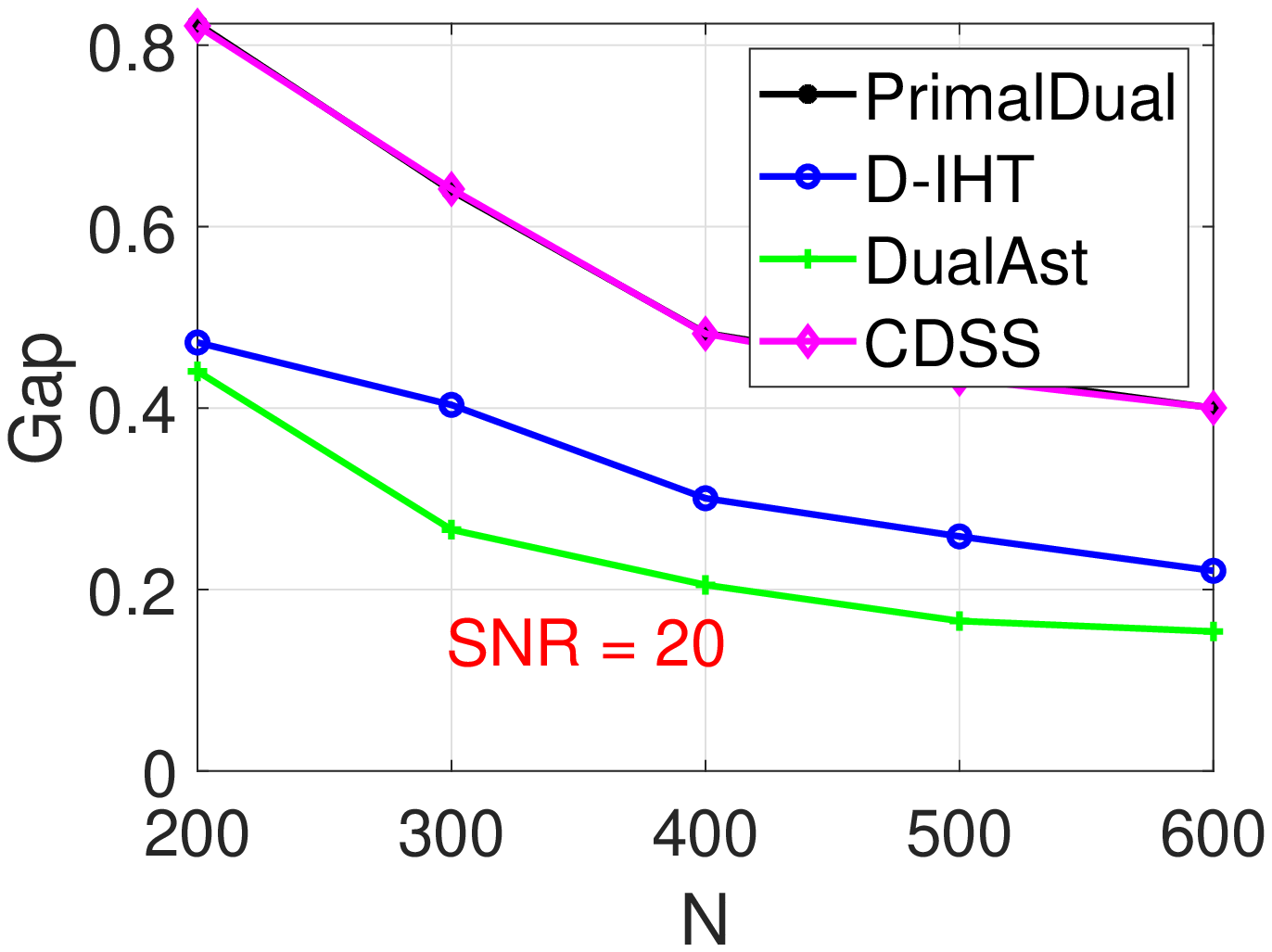}
\includegraphics[width=2.2in]{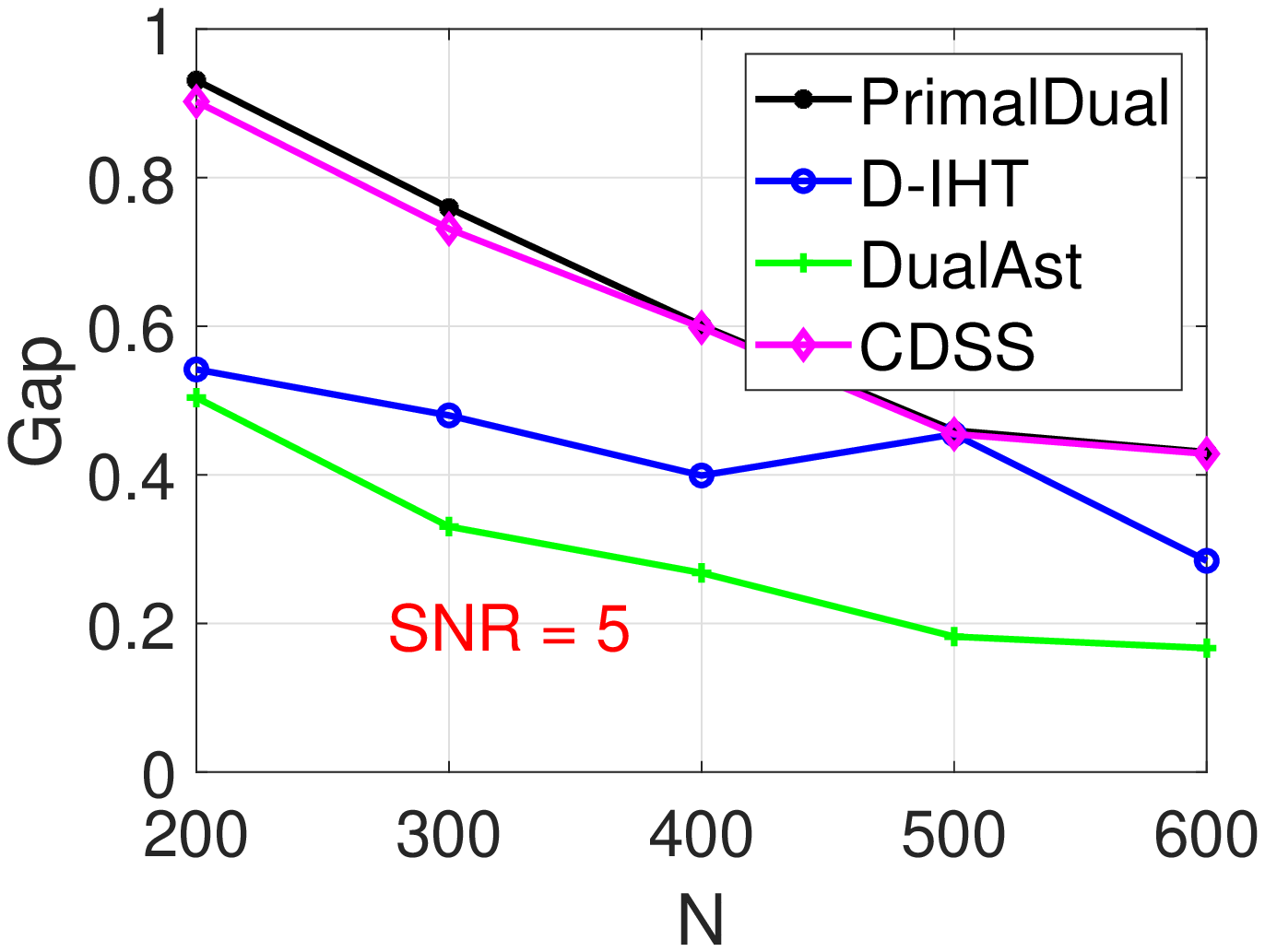} \includegraphics[width=2.2in]{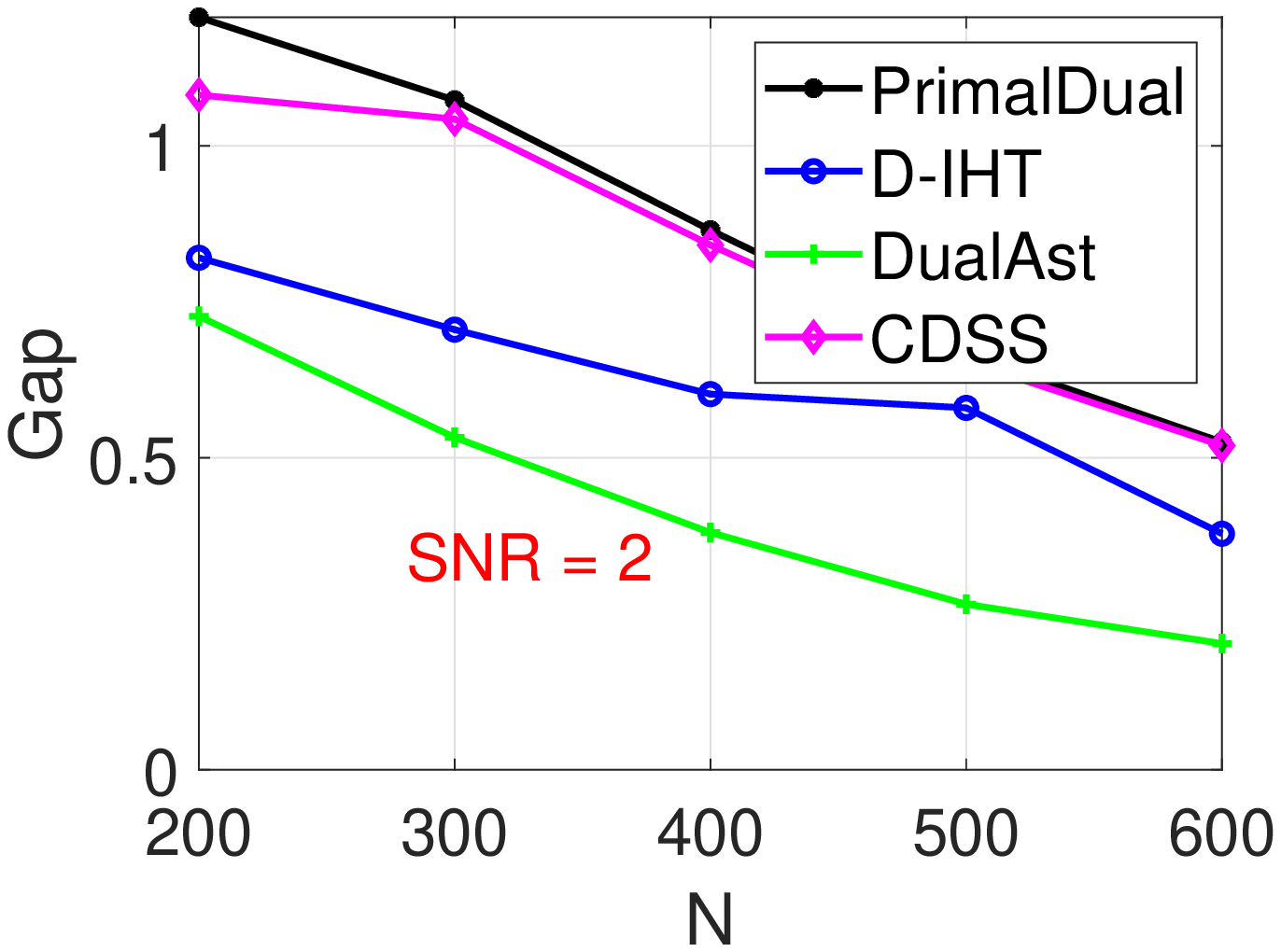} 
}

\mbox{\hspace{-0.0in}
\includegraphics[width=2.2in]{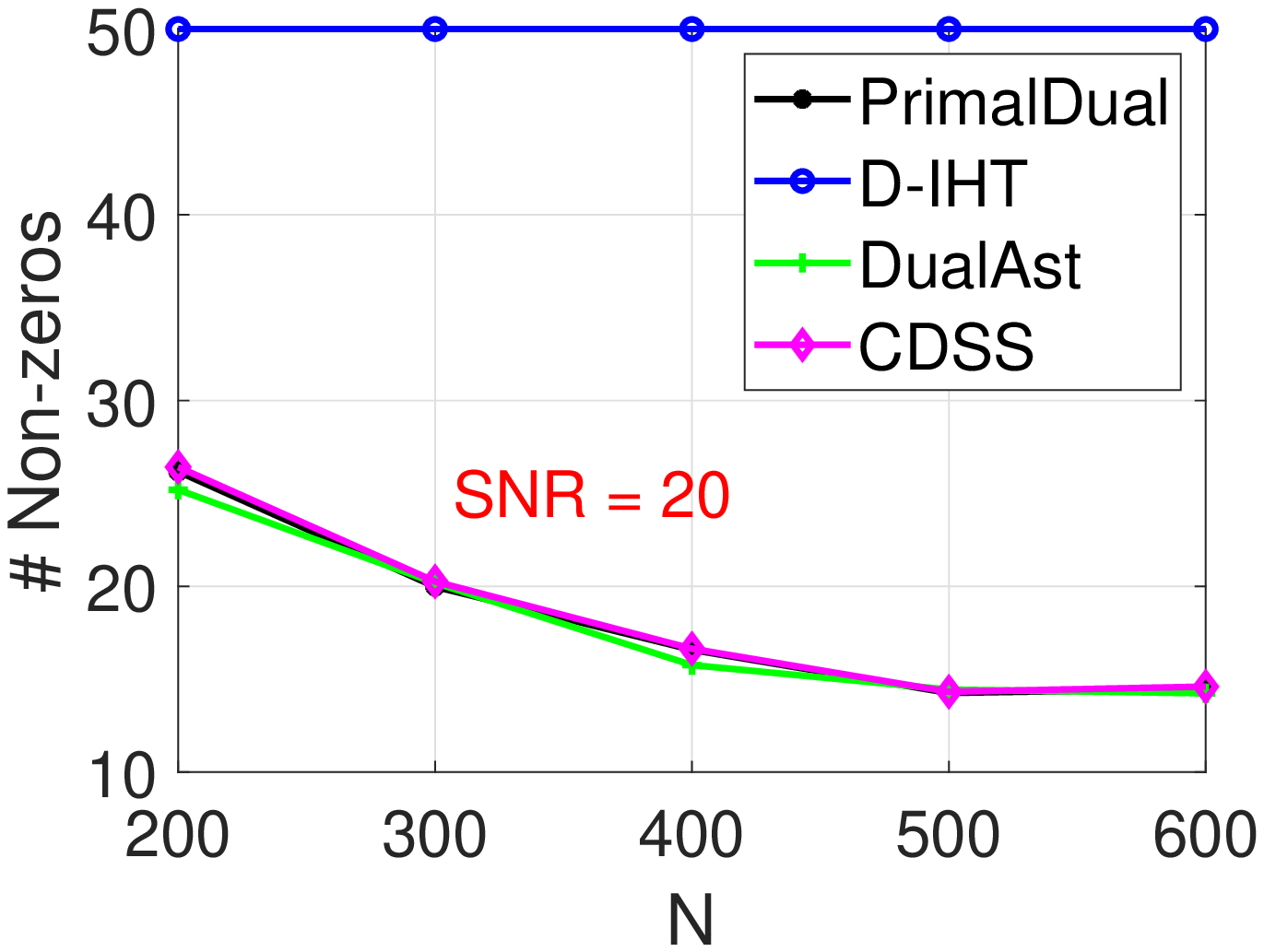}
\includegraphics[width=2.2in]{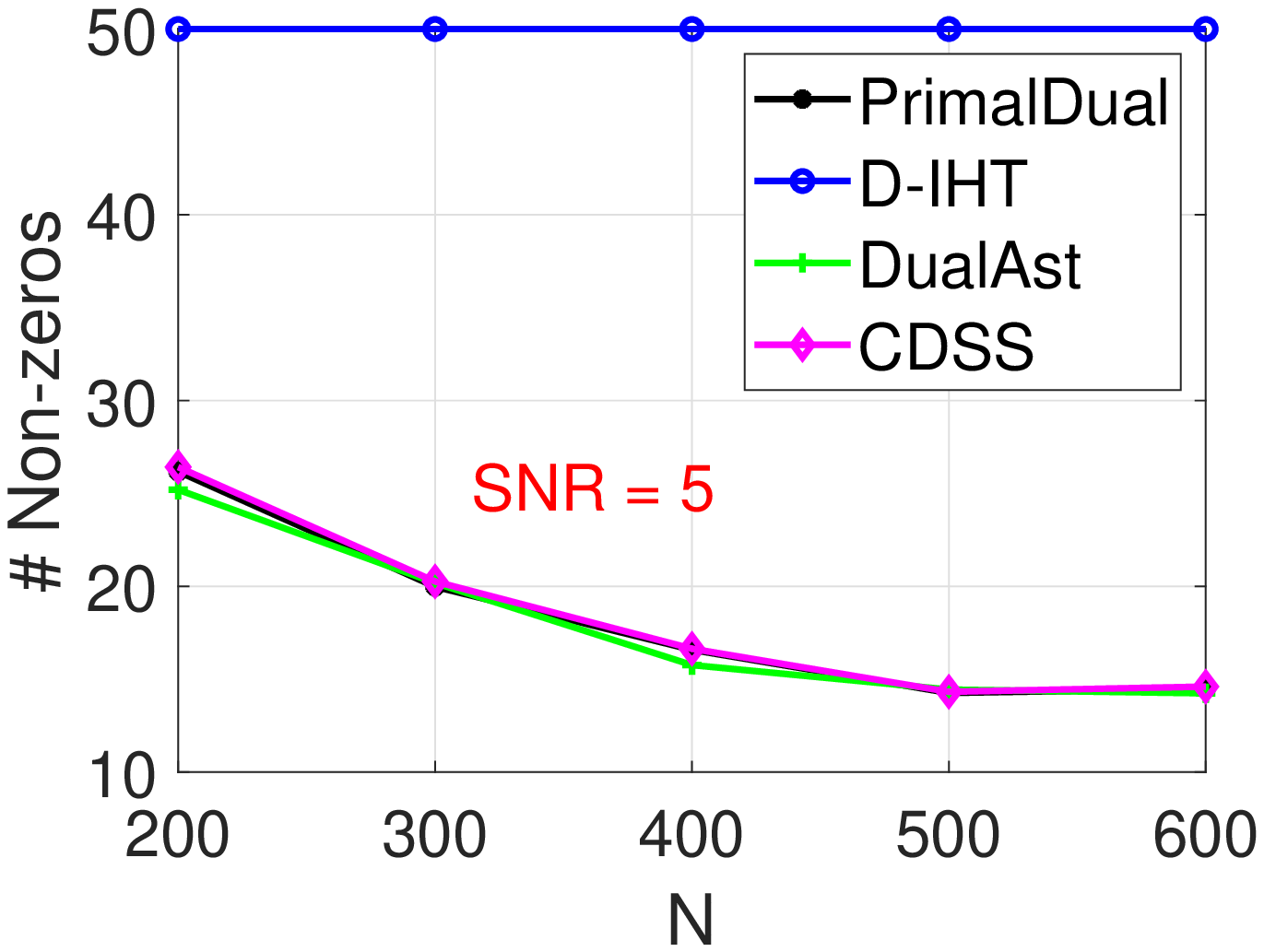} \includegraphics[width=2.2in]{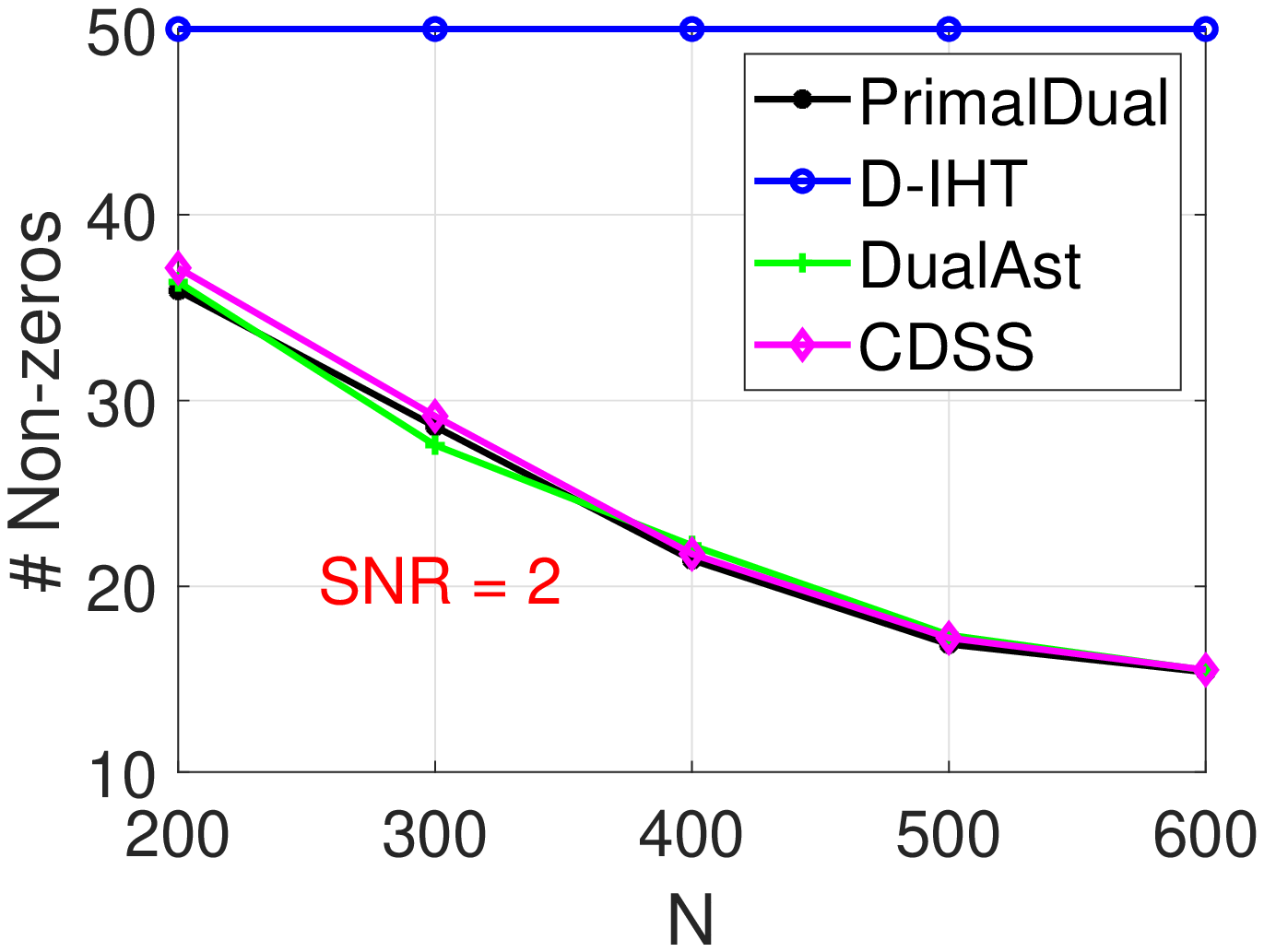} 
}

\mbox{\hspace{-0.0in}
\includegraphics[width=2.2in]{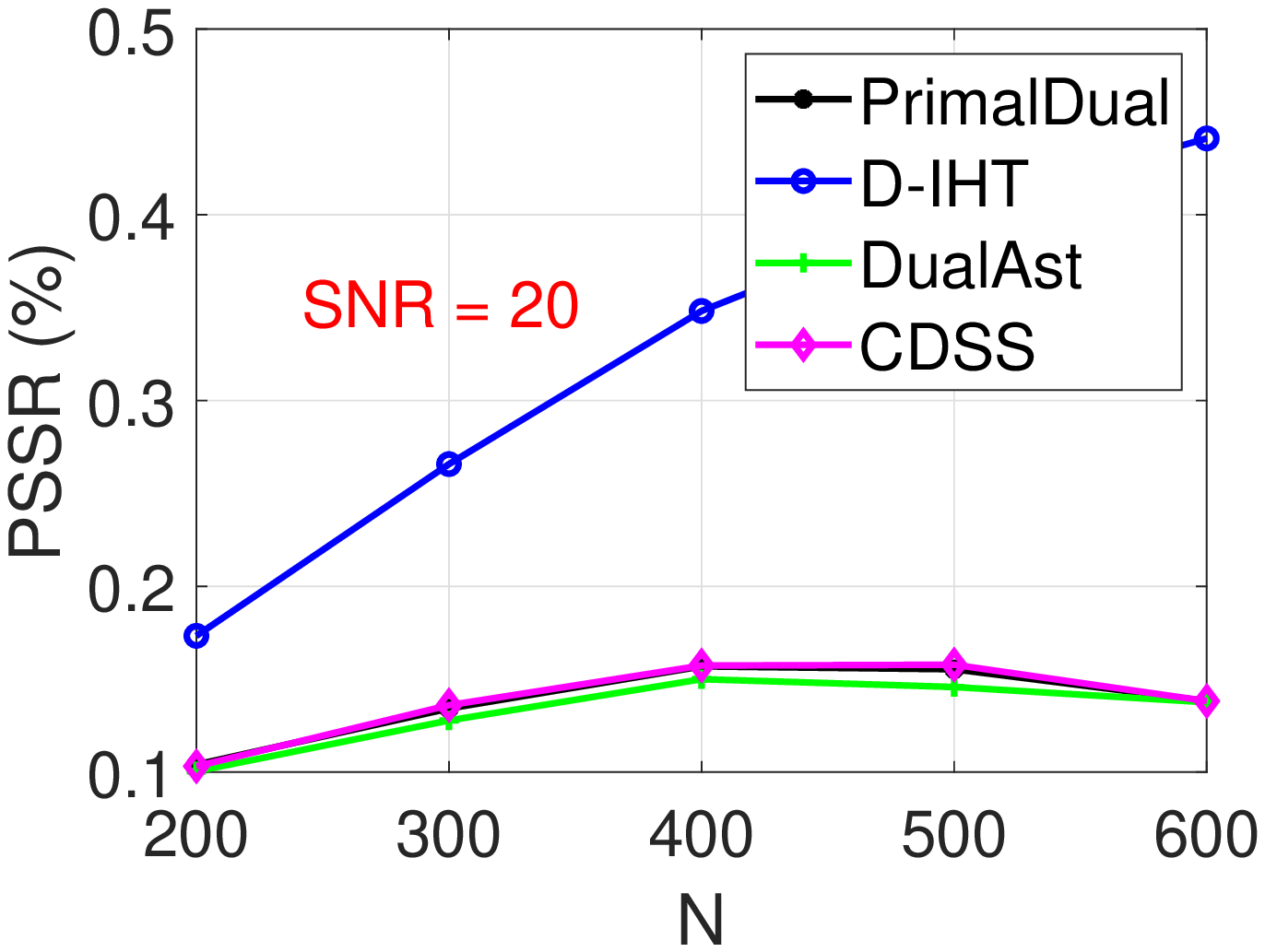}
\includegraphics[width=2.2in]{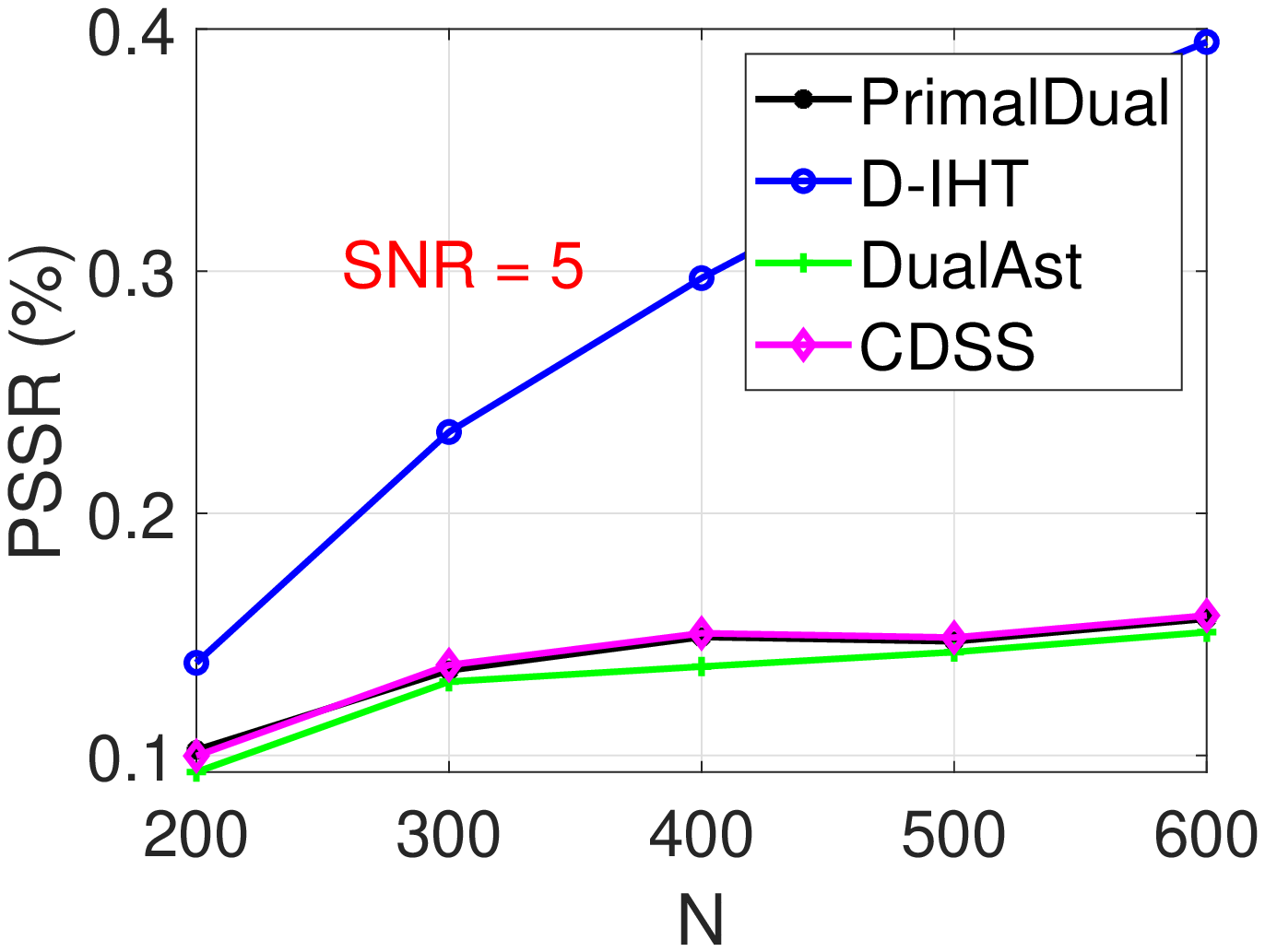}
\includegraphics[width=2.2in]{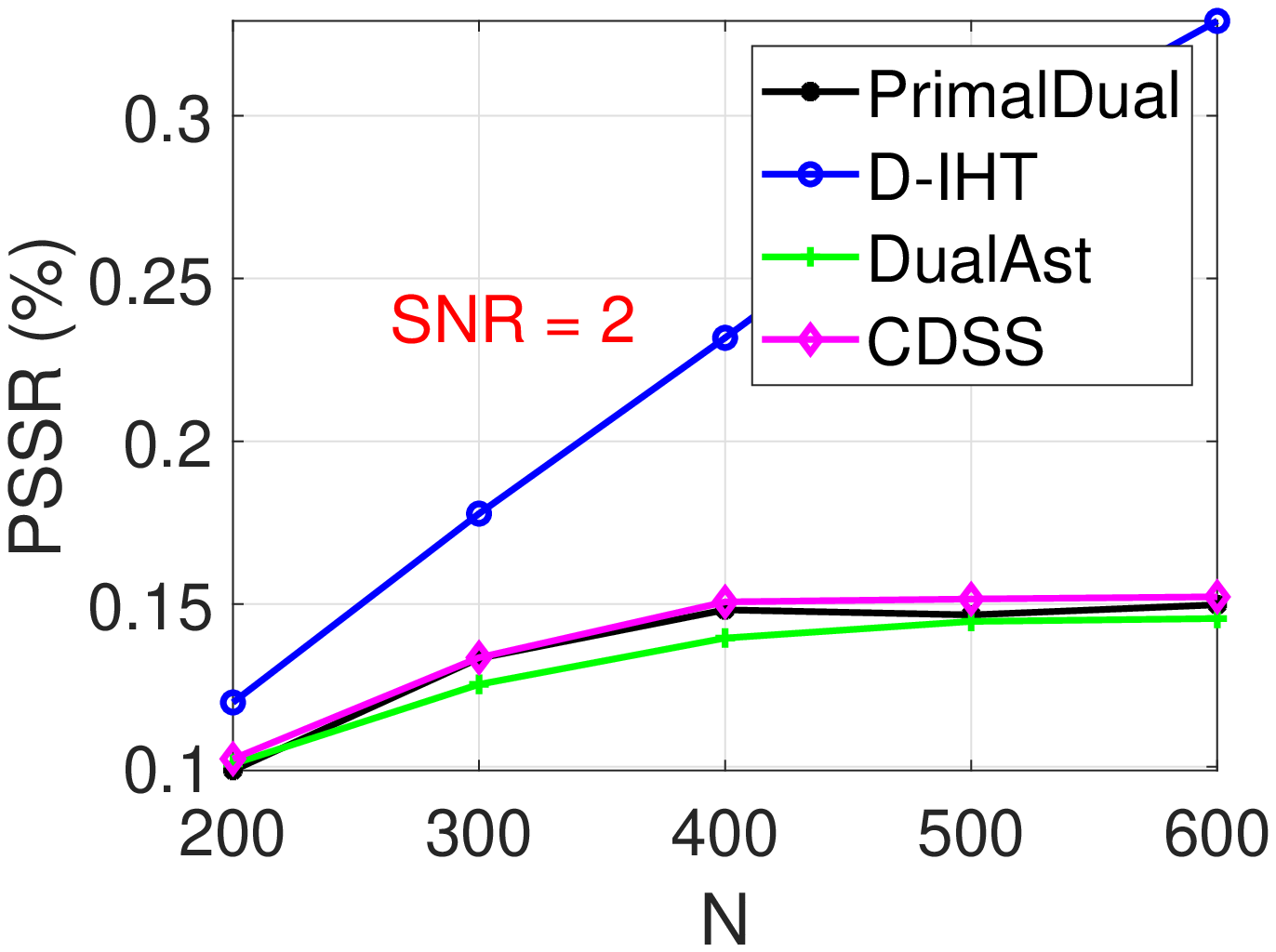} 
}

\mbox{\hspace{-0.0in}
\includegraphics[width=2.2in]{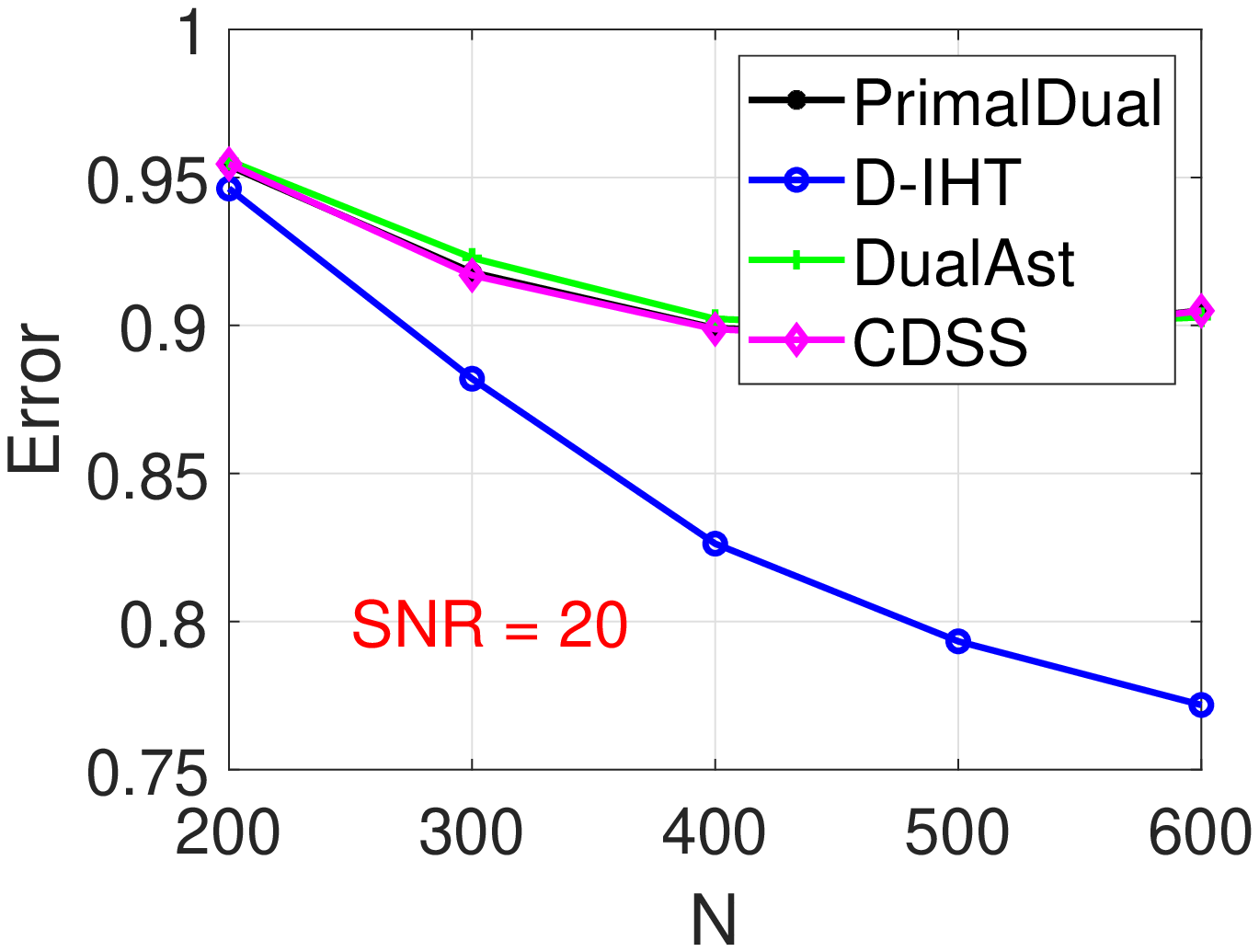}
\includegraphics[width=2.2in]{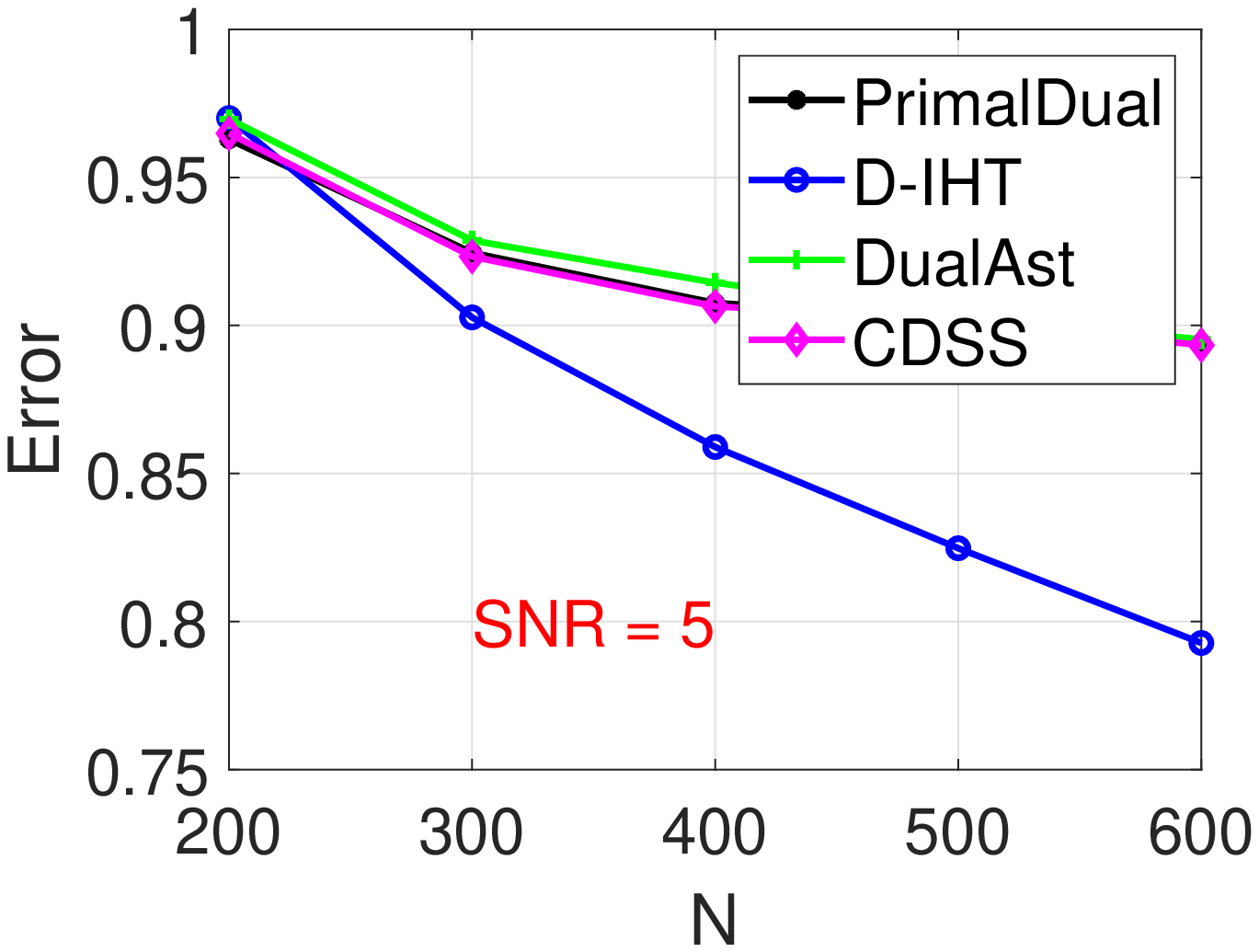}
\includegraphics[width=2.2in]{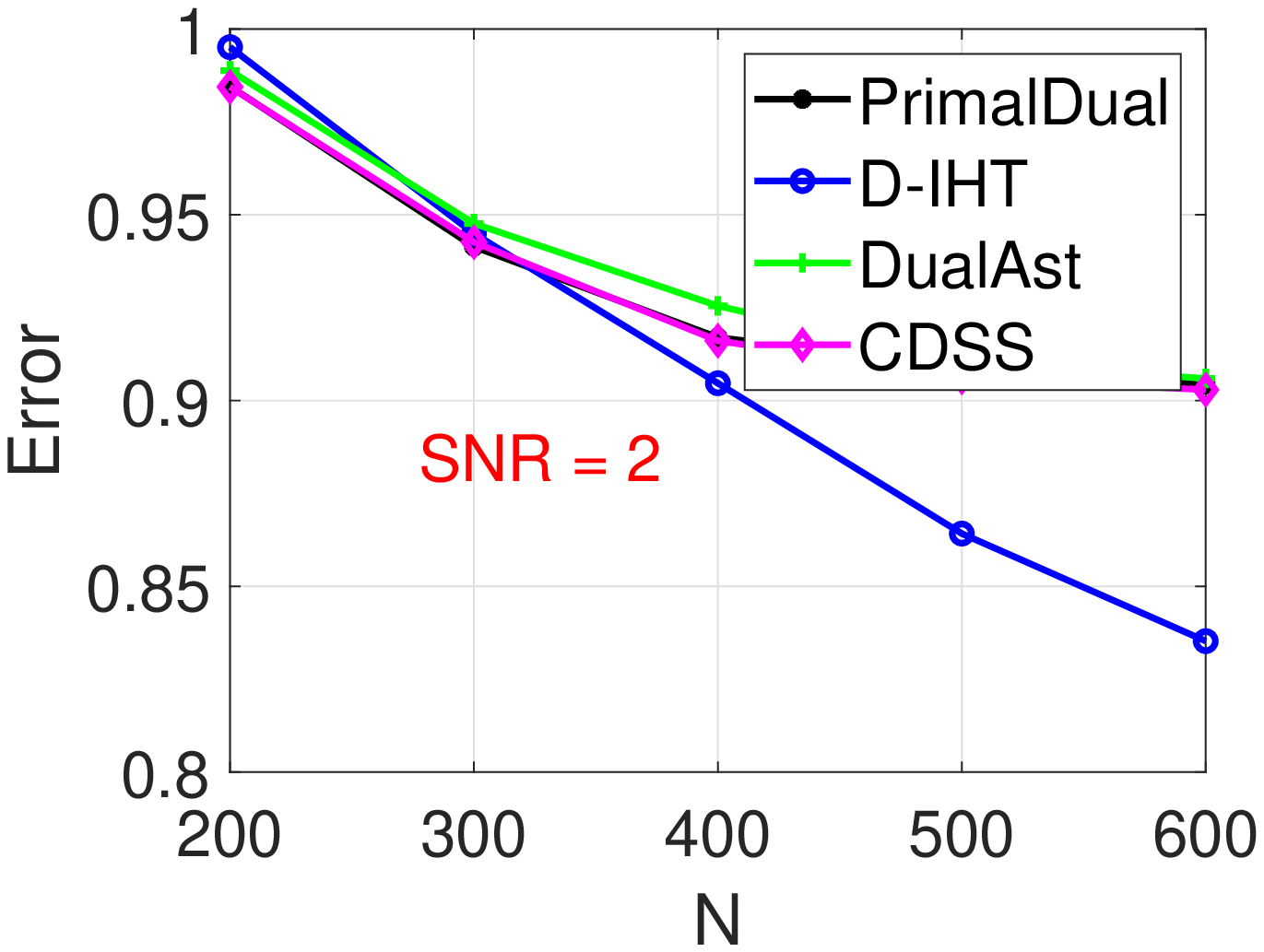} 
}

\vspace{-0.1in}

\caption{Running time, duality gap, nonzero number,  {\it PSSR}, and estimation error for four algorithms on simulated data with $SNR \in \{20, 5, 2\}$. x-axis is the number of training samples,  from 200 to 600.}\label{fig:sim}
\end{figure*}

We compare our primal-dual algorithm against  dual iterative hard thresholding~\citep{Liu17} and coordinate descent with spacer steps~\citep{Hazimeh18} algorithms. We use `Dual-IHT', and `CDSS' to represent the two algorithms, respectively. We use `PrimDual' to represent the proposed primal dual algorithm, and `DualAst' to represent dual ascent method which is Algorithm~\ref{alg:primal_dual_inner} without the primal updating steps. 
All algorithms are implemented in Matlab and run in the same environment. 
`Dual-IHT'~\citep{Liu17} is a primal-dual method with dual ascending  using hard threshold to keep $k$ largest value of $|\beta|$ in the primal space. CDSS~\citep{Hazimeh18} is a coordinate descent method operates in the primal space enhanced with PSI~(Partial Swap Inescapable) as the stopping criteria. We use  the duality gap~($DGap$) threshold $\epsilon=1.0E-6$ as the  stopping condition for Dual-IHT, DualAst and the proposed PrimDual algorithm. 
The algorithms may require extremely long time to reach a small duality gap threshold.   We also use the duality change, i.e., $\zeta = |DGap^{t-2} - DGap^{t}|$ as a stopping condition for the three algorithms, and we set $\zeta = 1.0E-6$ in the experiments.  Moreover, we use the same learning rate $\omega= 0.000\,5$ for the three algorithms.

Two indices are adopted for evaluating the performance. The first one is the  {\it percentage of successful support recovery~(PSSR)}.  The second one is parameter estimation error $|| \beta - \bar{\beta}||/||\bar{\beta}||$. Here $\bar{\beta}$ is the ground truth used in simulation. Figure~\ref{fig:sim} gives the performance of these four algorithms on datasets with different SNR values.
To achieve meaningful comparison, we choose $\lambda_0, \lambda_1$ and $\lambda_2$ to recover support number close  to the ground truth value. For the dataset with $SNR=20$, we use $\lambda_0 = 0.03, \lambda_1 = 0.02, \lambda_2 = 1.0$, and we set $\lambda_0 = 0.1, \lambda_1 = 0.2, \lambda_2 = 1.0$ for dataset with $SNR=2, 5$.
 From the plots, we can see that the proposed primal-dual algorithm can achieve similar {\it PSSR} and estimation error values (expect for Dual-IHT since its sparsity is pre-determined), but use much less time. It shows that the proposed primal-dual algorithm and incremental strategy significantly reduce the redundant operations resulted from inactive features. 

\subsection{ Experimental on Real-world Datasets}

In this section, we present additional results to compare the four algorithms on two datasets, News20 and E2006. In these experiments, we use learning rate $\omega=0.000\,5$ for Dual-IHT, DualAst, and our PrimDual method. The stopping conditions for the three algorithms are $\epsilon=1.0E-6$ and $\zeta = 1.0E-6$. We set $c=4.0$ for Algorithm~\ref{alg:add} in our experiments. It is difficult to fairly compare  D-IHT with the other three methods as it uses different objective with hard constraints, and we have to specify the hyper-parameter $k$ in advance. The values of  $k$ are set heuristically in the experiments.

\subsubsection{News20 Dataset}

After pre-processing, the commonly used News20 dataset contains 20 classes, $15\,935$ samples, and $62\,061$ features in the training set. The 20 labels in News20 dataset are transformed to response values ranging  $[-10,10]$ in the experiments.  We randomly pick up $p=2\,000$ features and to form five datasets with sample number ranging in  $\{200, 300, 400, 500, 600\}$. We use learning rate $\omega=0.0005$ for Dual-IHT, DualAst, and our PrimDual  method. The stopping conditions are $\epsilon=1.0E-6$ and $\zeta = 1.0E-6$. The hyper-parameters are set with $\lambda_0=0.1, \lambda_1=0.15, \lambda_1=1.0$. The left column of Figure~\ref{fig:news20} shows the results of different methods on News20 dataset with $p=2\,000$.
Each setting is replicated for 20 times.
We can see that under approximately the same primal objective  and duality gap values, the proposed primal-dual method uses less computation time compared against other methods when $N$ becomes larger. Though DualAst consumes similar computation cost as our method, it cannot achieve small duality gap values on all cases.  We notice that CDSS takes longest time in this case, and it could be due to that the PSI stopping condition is hard to satisfy on some real-world datasets.

 \begin{figure}[]
 
 \centering
\mbox{
\includegraphics[width=2.5in]{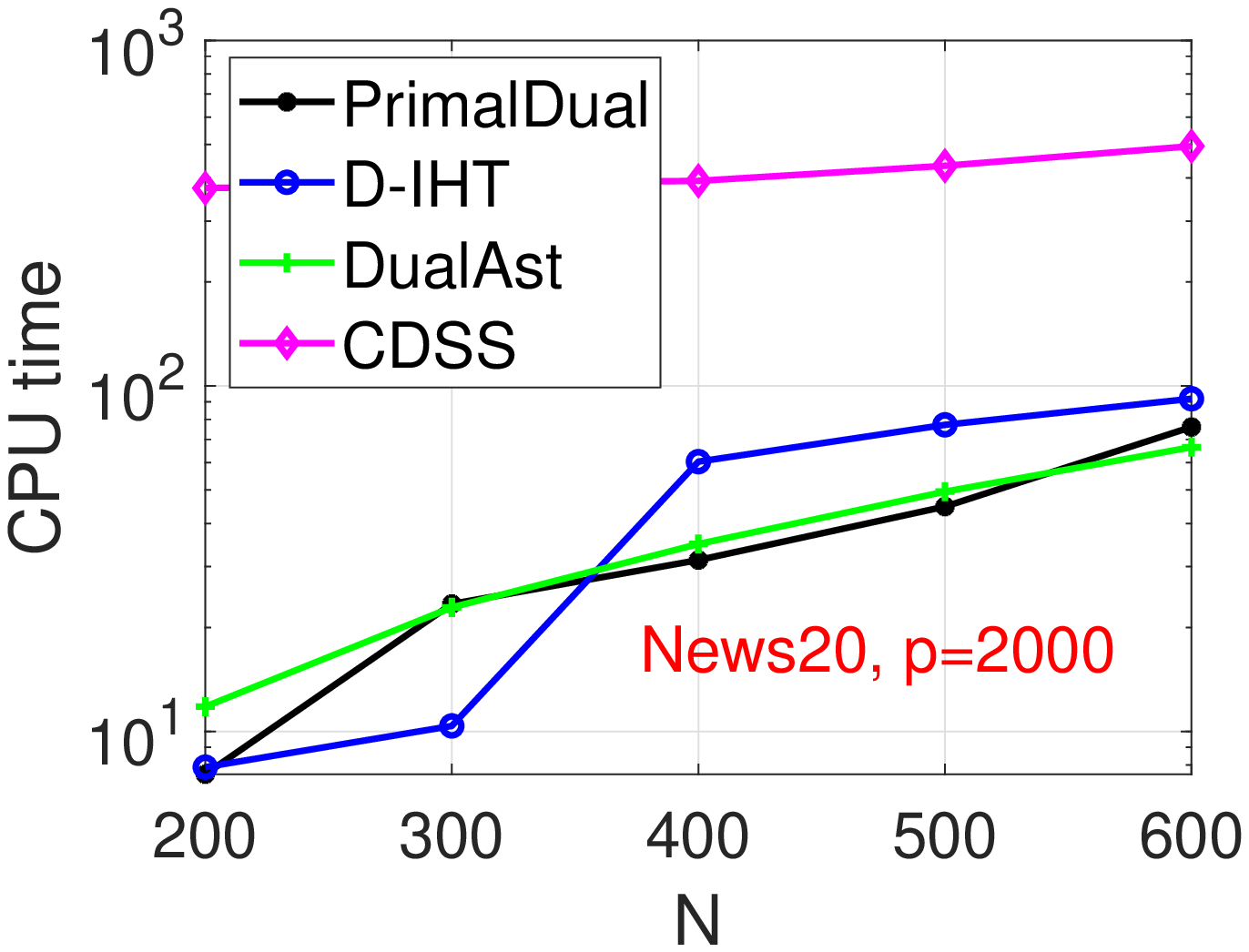}
\includegraphics[width=2.5in]{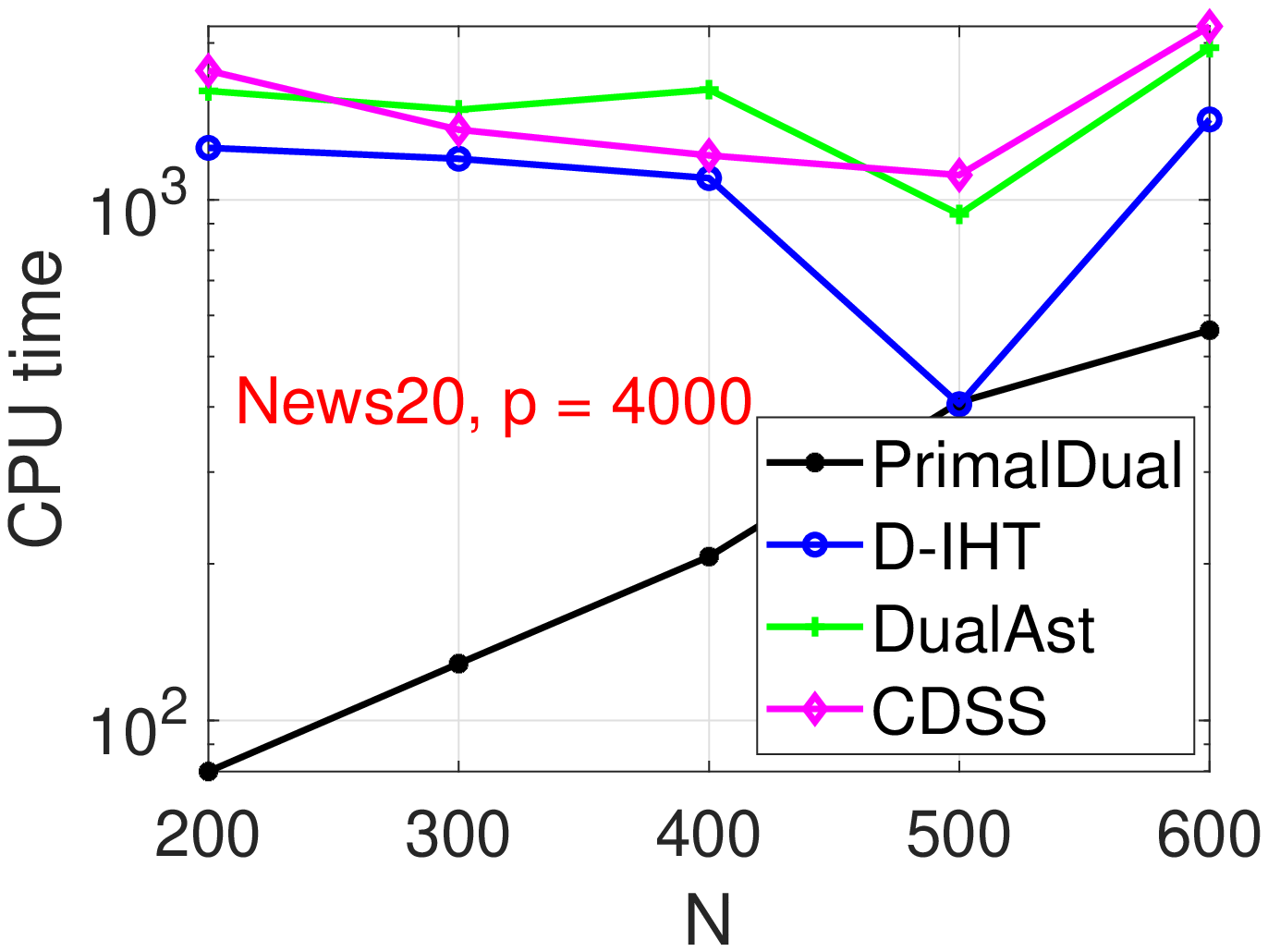}
}

\mbox{
\includegraphics[width=2.5in]{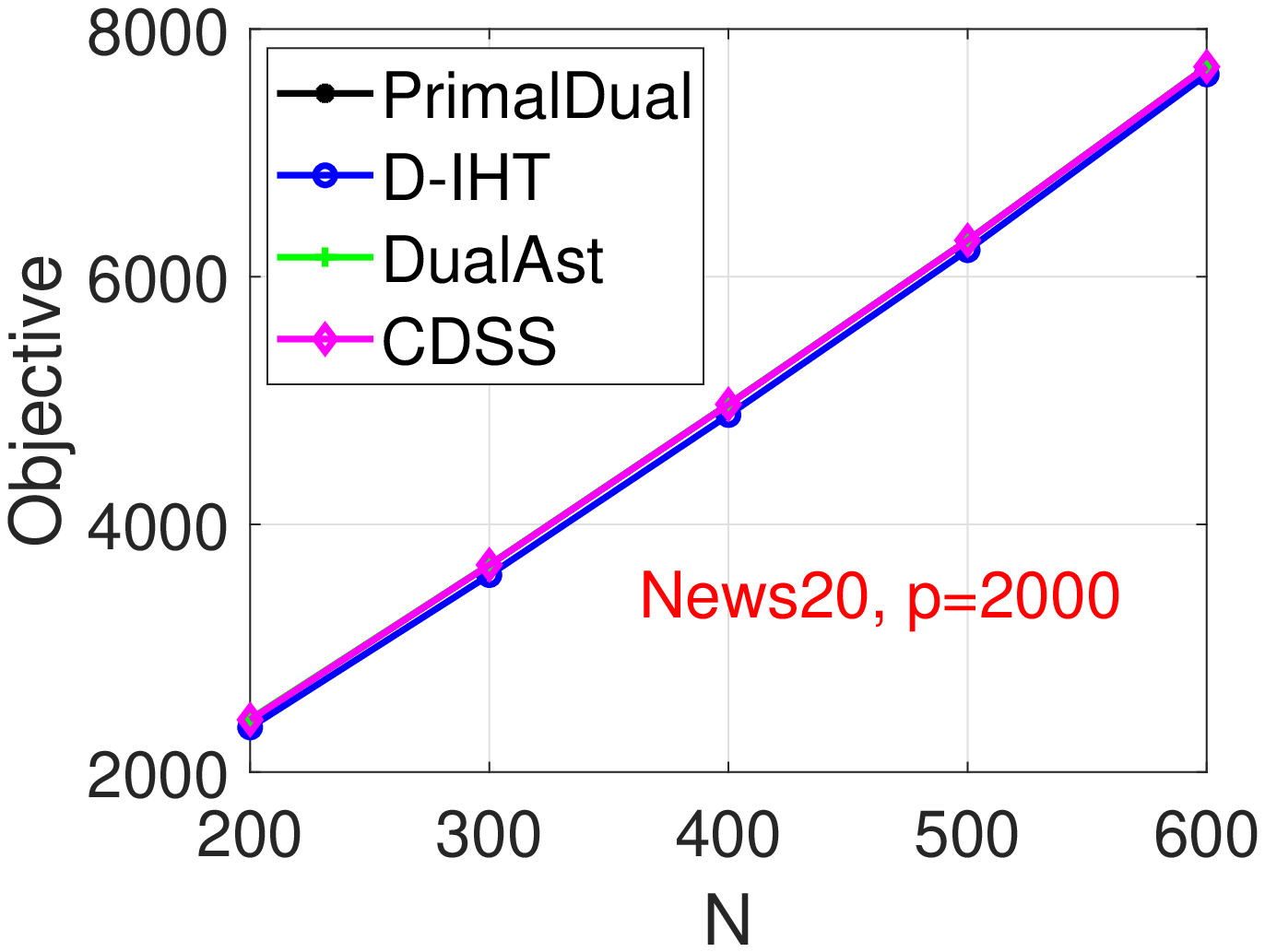}
\includegraphics[width=2.5in]{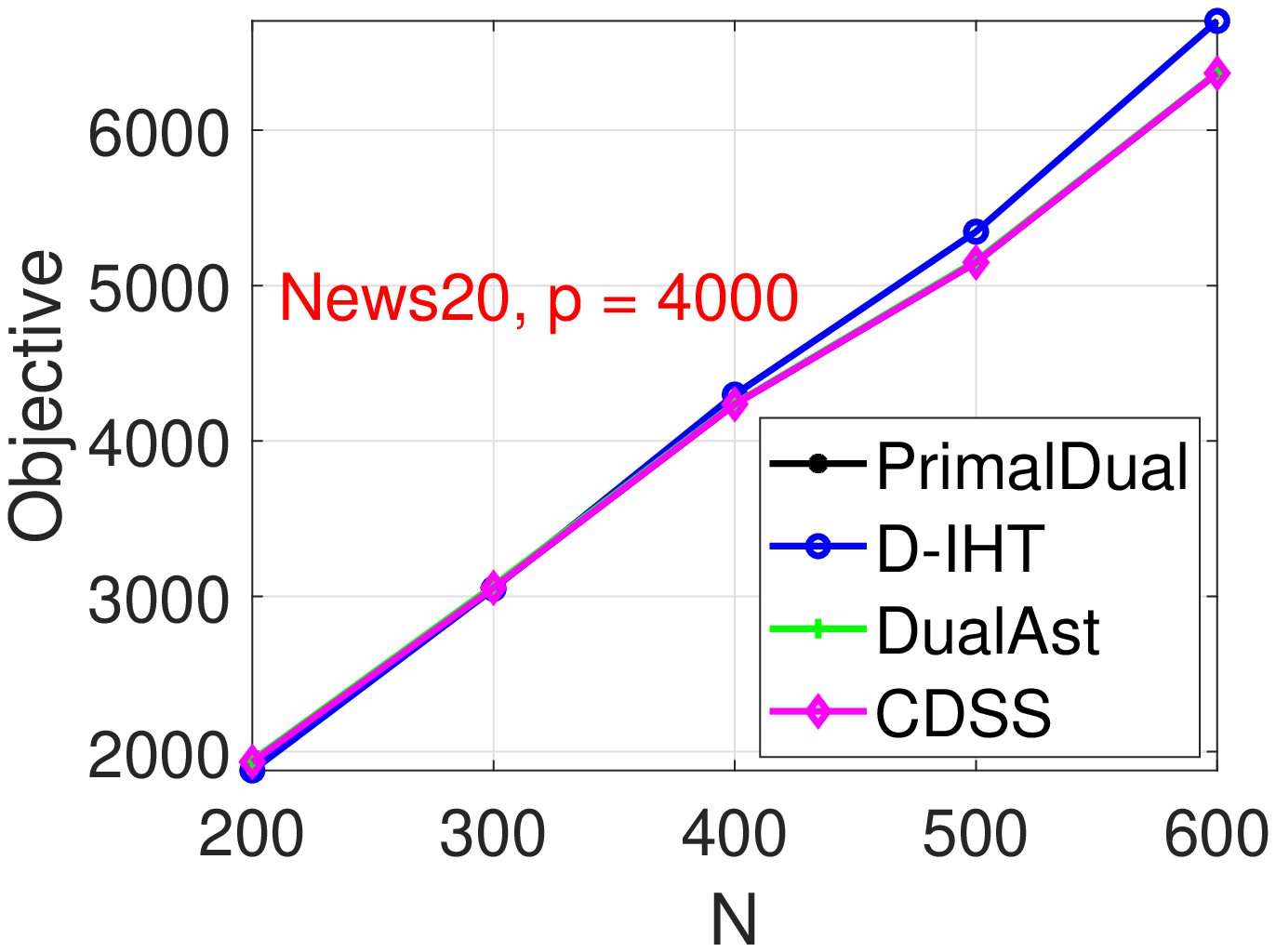}
}

\mbox{
\includegraphics[width=2.5in]{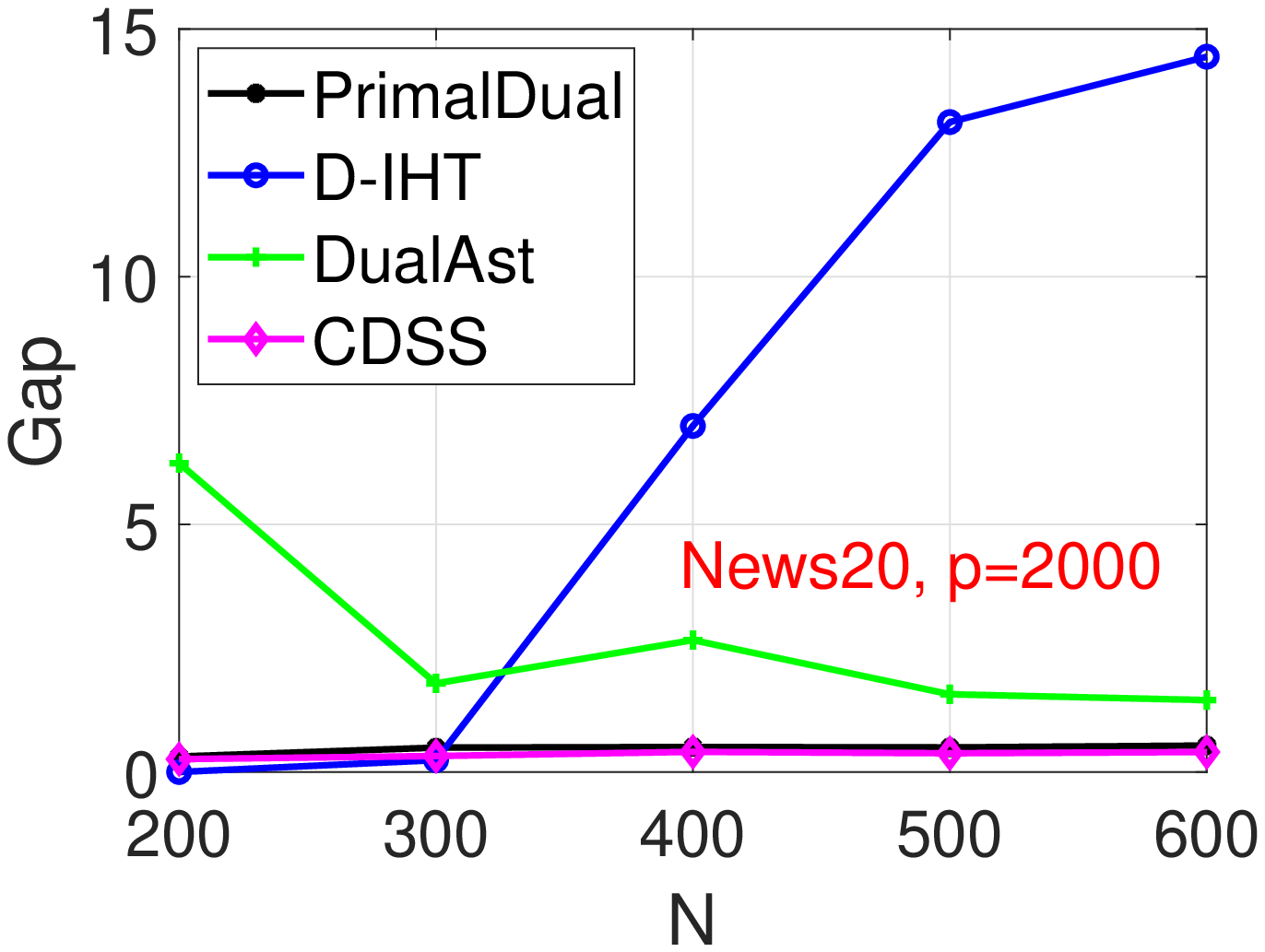}
\includegraphics[width=2.5in]{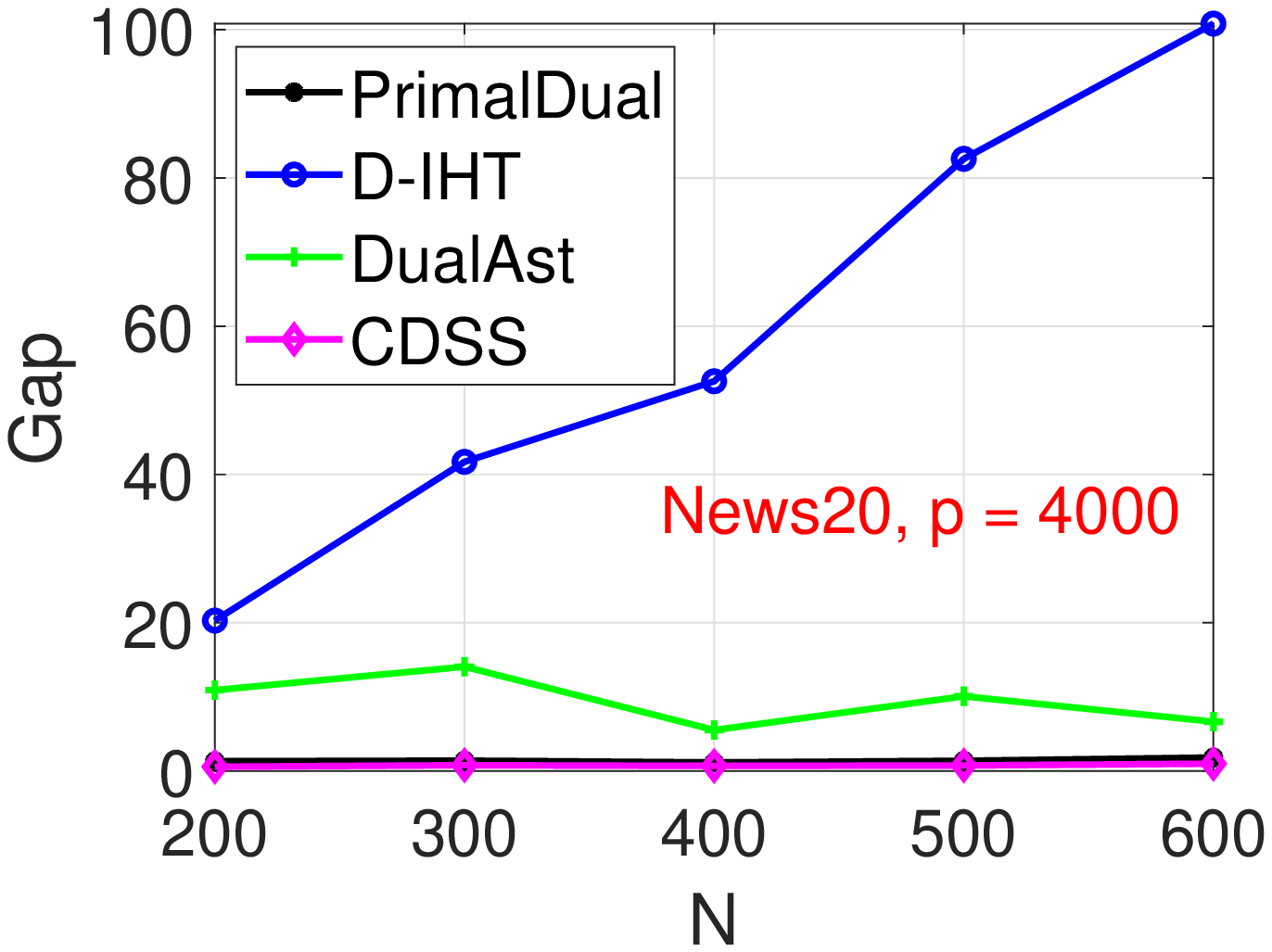}
}

\mbox{
\includegraphics[width=2.5in]{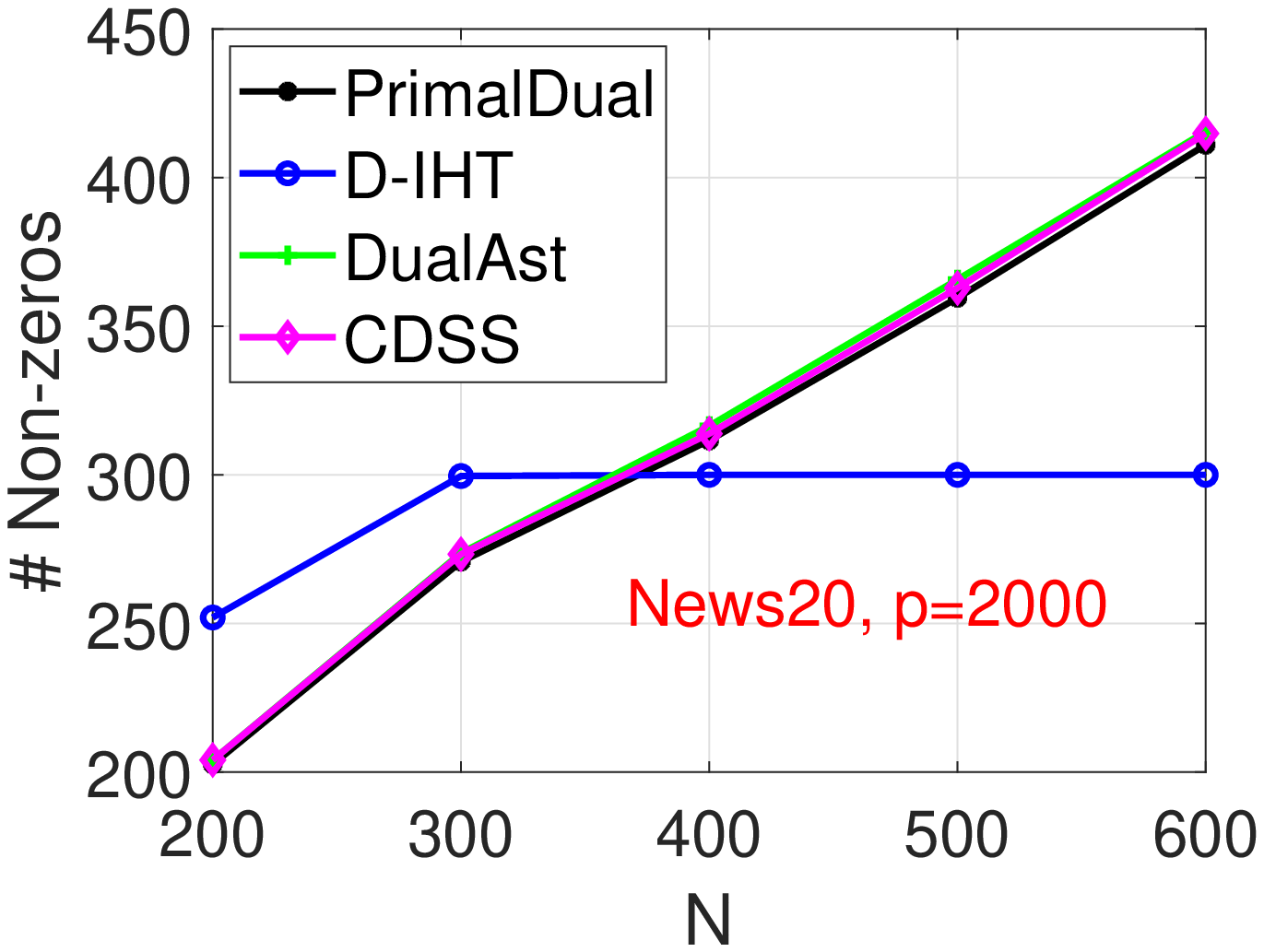}
\includegraphics[width=2.5in]{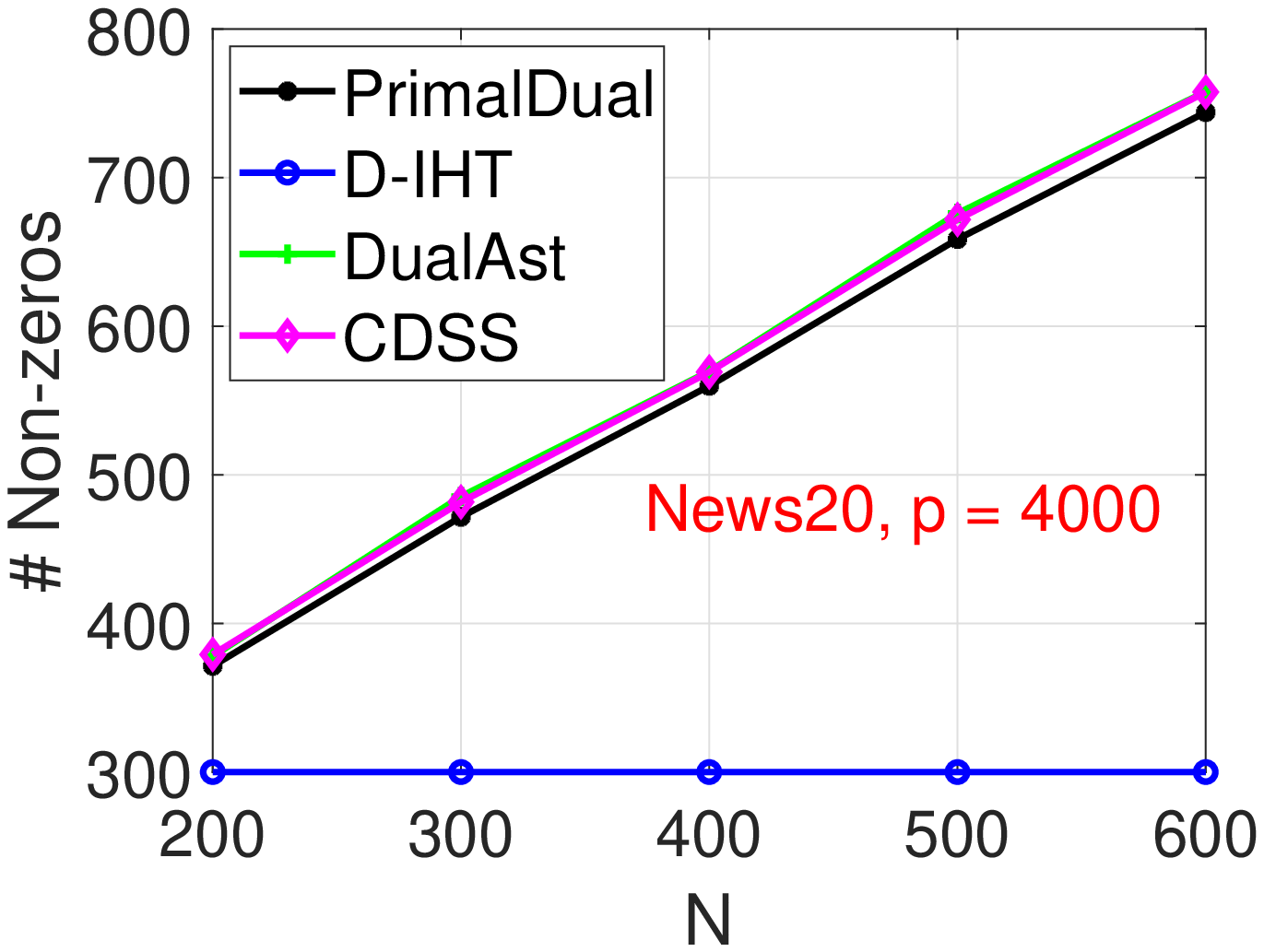}
}

\vspace{-0.15in}
\caption{The  running time, primal objective $P(\hat{\beta})$,  duality gap, and nonzero number  for different methods on News20 dataset. }\label{fig:news20}
\end{figure}

The right column of Figure~\ref{fig:news20} gives other results on News20 datasets with  a larger feature number ($p=4\,000$). Hyper-parameters  are  set as $\lambda_0=0.10, \lambda_1=0.15, \lambda_2=1.00$. We set $k=300$ for D-IHT method.  Each case is replicated for 10 times and the average result is reported.
From the plots, we can see that with longer running time, CDSS achieves smallest duality gap values. However, the proposed method  uses much less computation cost to achieve similar solutions.  

\newpage

\subsubsection{E2006}
E2006 regression dataset has 150\,360 features and 16\,087 samples for training\footnote{\url{https://www.csie.ntu.edu.tw/~cjlin/libsvmtools/datasets/regression.html}}. We also randomly select $p = 4\,000$ features and $n=\{200, 300, 400, 500, 600\}$ samples in our experiments.
 Hyper-parameters  are set to $\lambda_0=0.10, \lambda_1=0.10, \lambda_2=1.00$, and $k=300$ for D-IHT.  
 Each case is replicated for 10 times and the average result is reported.
 Figure~\ref{fig:news20_2} gives  results on E2006 datasets.
Again, the proposed primal-dual algorithm is much more efficient in solving the problems to achieve similar duality gaps compared against other methods. The results validate the duality theory, problem properties, and the proposed incremental strategy.

\begin{figure}[]

\centering
\mbox{
\includegraphics[width=2.5in]{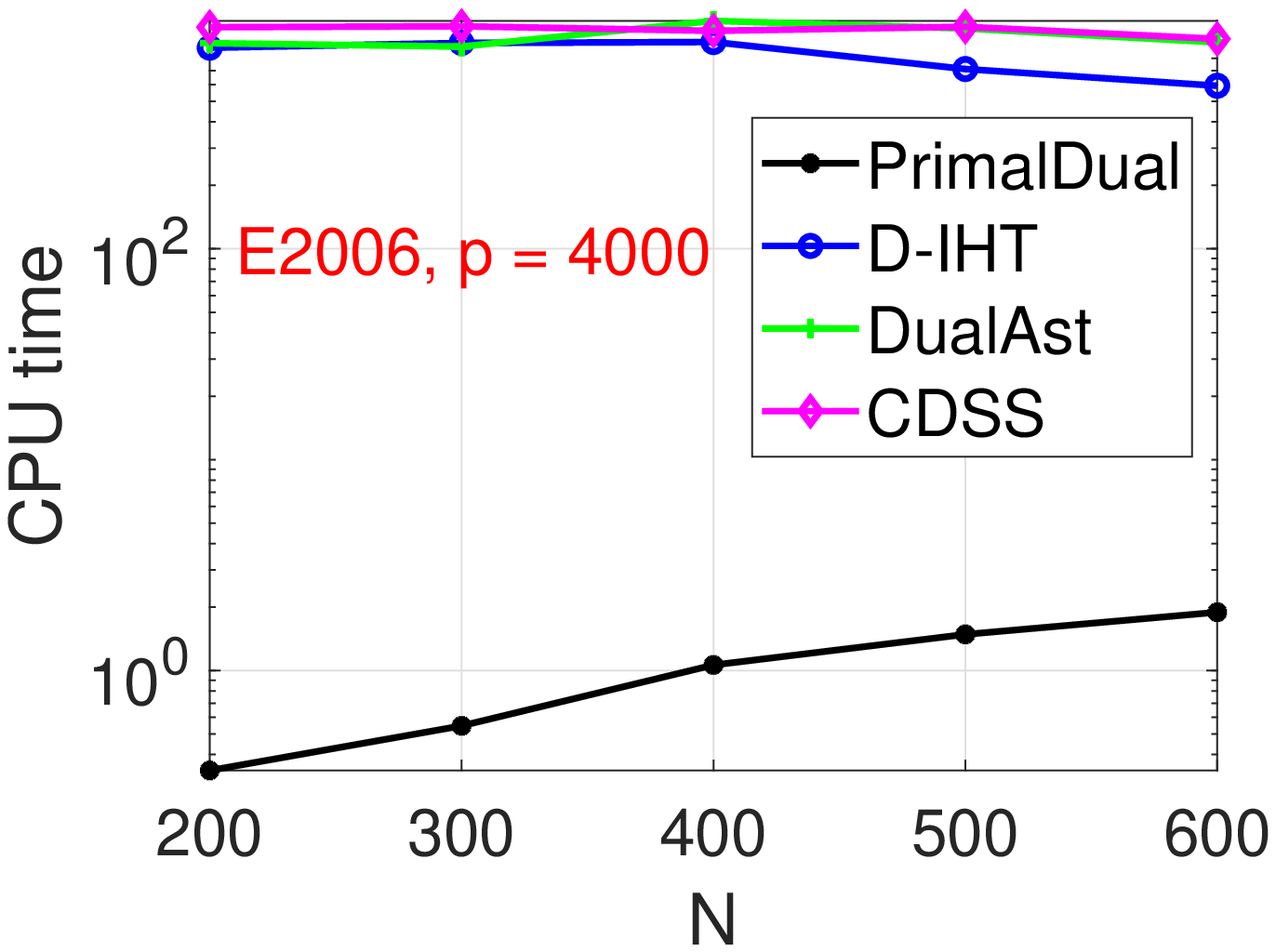}
\includegraphics[width=2.5in]{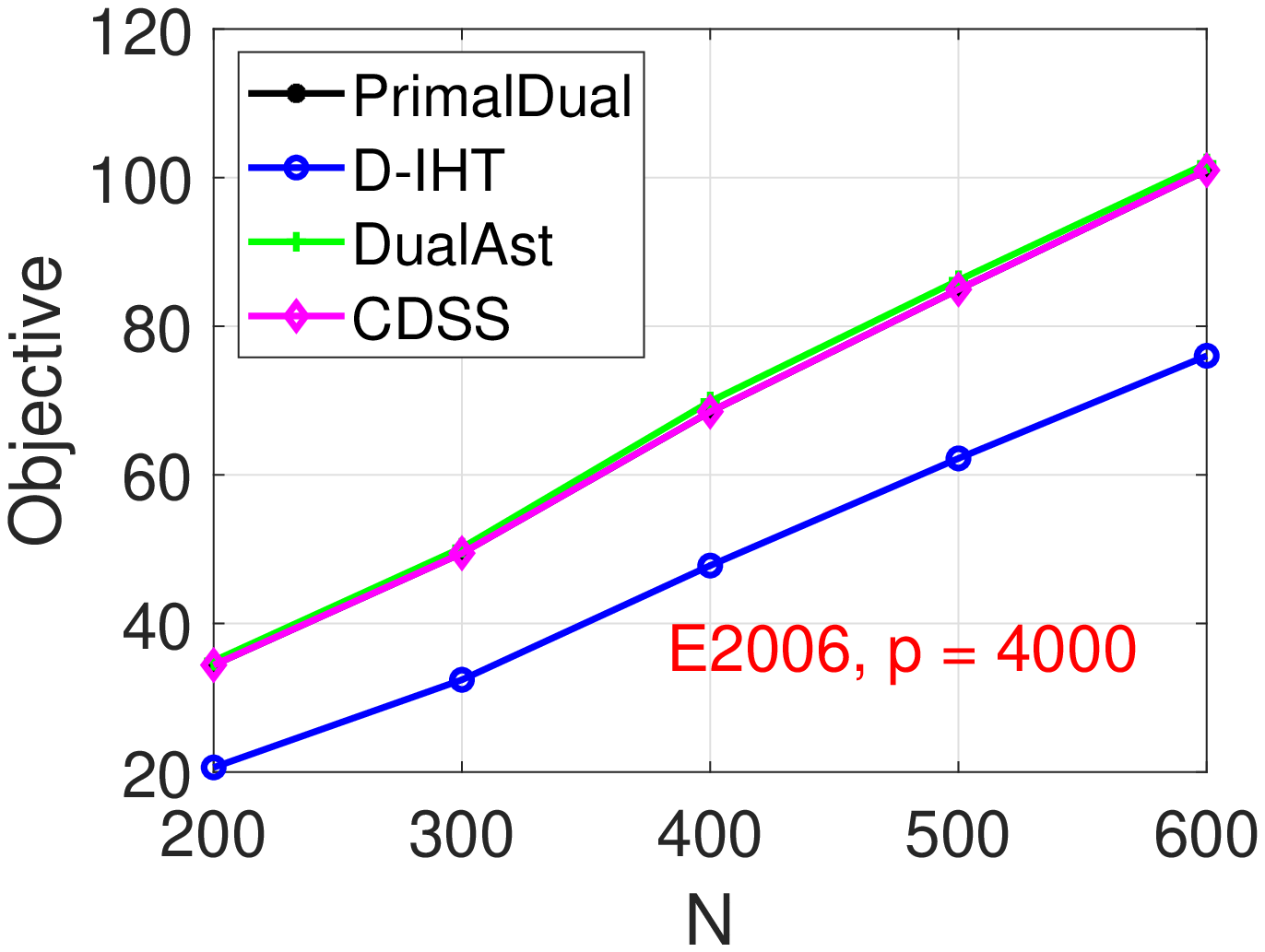}
}

\mbox{
\includegraphics[width=2.5in]{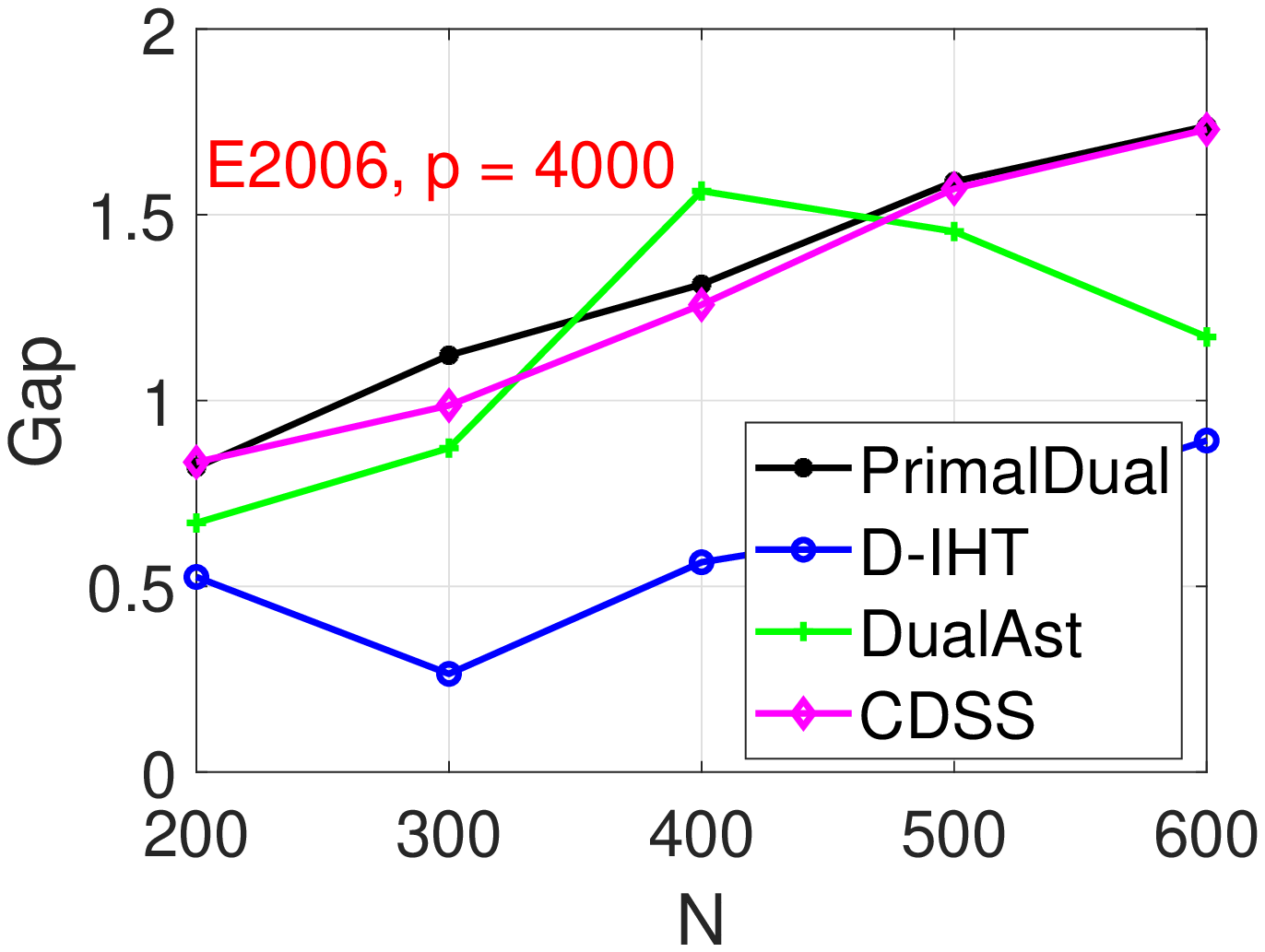}
\includegraphics[width=2.5in]{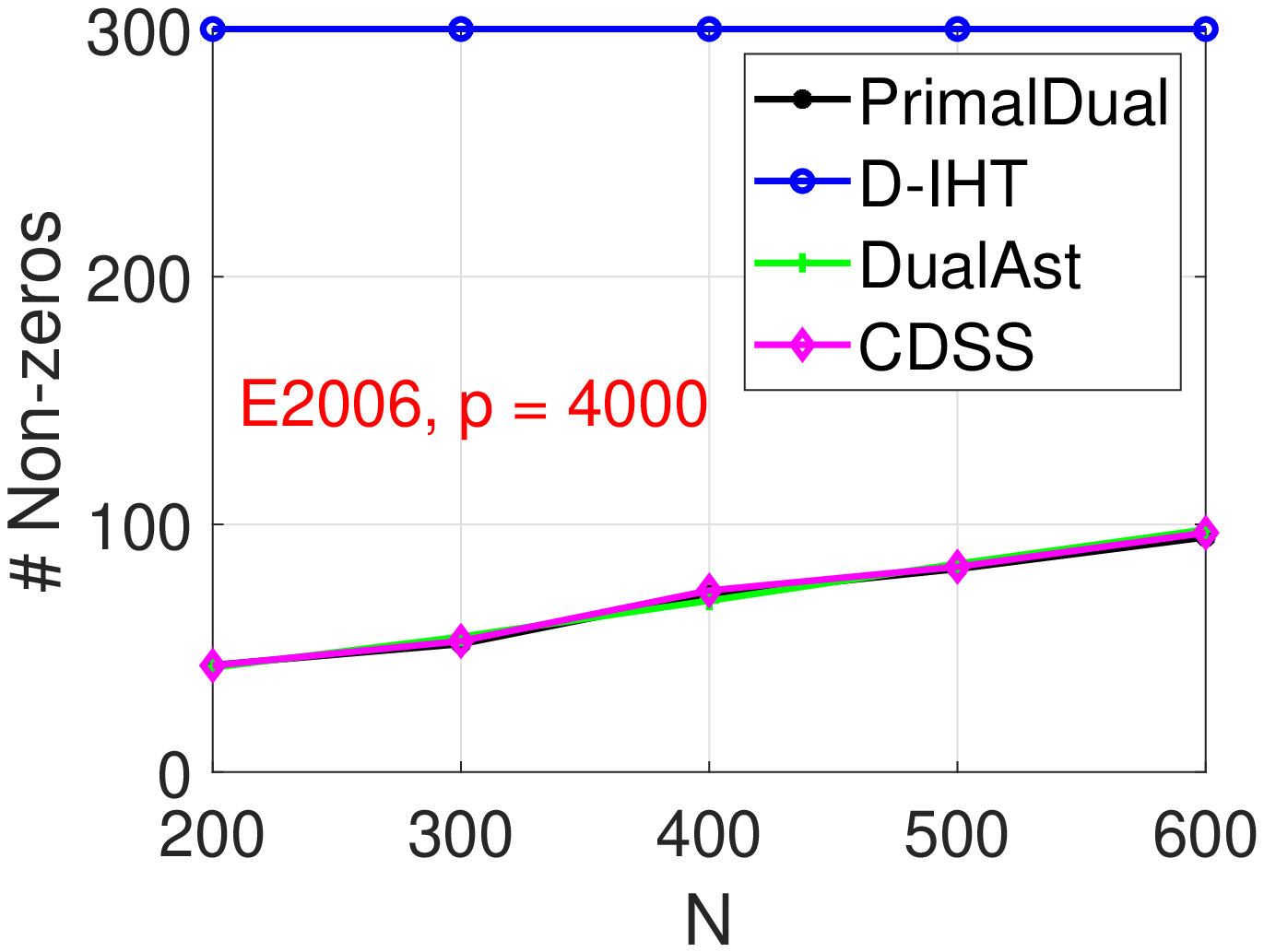}
}

\vspace{-0.1in}
\caption{The  running time, primal objective $P(\hat{\beta})$,  duality gap, and nonzero number  for different methods on E2006 datasets with $p= 4\,000$. }\label{fig:news20_2}
\end{figure}

\section{Discussion}\label{sec:discuss}
In this section, we first give several additional remarks on the method, and then we compare our approach with some existing methods. 
\subsection{Additional Remarks}\label{sec:add_remarks} 

\begin{itemize}
    \item[] \textbf{Technical Contribution:}  This paper investigates  the dual form and  strong duality of the generalized sparse learning problem~\eqref{eq:primal1}  by following~\cite{pilanci2015sparse,Liu17,yuan2020dual}. The  generalized form could overcome over-fitting issues  of $\ell_0$ regularized problems when the data  SNRs are low~\citep{pmlr-v65-david17a,mazumder2022subset}.  A  primal-dual framework  has been developed to further scale up the solutions of $\ell_0$ regularized problems based on the derived dual form. Moreover, the proposed framework considers  active coordinate incremental and screening strategies~\citep{Fercoq2015,GAP,Ndiaye2017,atamturk2020safe,Celer,ren2020thunder}  by leveraging  the duality structure properties of  problem~\eqref{eq:primal1}. The quality of solutions can be evaluated by the duality gap~\eqref{eq:dgap} with the current dual solution $\alpha$ calculated through Remark~\ref{rmk:dual_condt}. 
    
     \item[] \textbf{Saddle Point:} Different from the sparse saddle point defined in~\cite{Liu17,yuan2020dual} that requires $k$-sparse regarding the primal variable,  the saddle point in this paper can be taken as a standard saddle point.  Without the specified $k$, our  duality theory is more close to the standard duality paradigm, and hence some generic primal-dual methods can be employed to further improve the solver. The methodology developed here can be easily extended to plain $\ell_1$ or $\ell_0$ problems (with  the $\ell_2$ term),  group sparse structures, fused sparse structures, or even more  complex and  mixed sparse structures that we cannot or do not need to specify the $k$ values.
    
    \item[] \textbf{Strong Duality:}  Strong duality holds when both the primal and dual variables reach the optimal values. A closer distance between the current estimation $\alpha$ and the  optimal value $\bar{\alpha}$ gives a smaller duality gap. Once the support of $\bar \beta$ is recovered,  the objective function becomes a convex function since $\|\beta^s\|_0$ remains a constant. According to our theoretical analysis, the  saddle point of problem~\eqref{eq:primal1} could be attained within polynomial computation complexity with a decreasing step size. 
    When $l(\cdot)$ takes some special convex function, strong duality holds. However, for general $l(\cdot)$, it is hard to verify the strong duality.

\end{itemize}

\subsection{Comparisons}\label{sec:discussion}

We provide more details on the differences between our work and related works in this section to highlight our technical contributions. 

The proposed method is significantly different from Dual-IHT method~\citep{Liu17,yuan2020dual}. Firstly, the proposed primal-dual algorithm  focuses on a different problem~\eqref{eq:primal1} compared with the Dual-IHT's objective~\eqref{eq:hard_k}. Apart from using soft regularization rather than hard constraint, ~\eqref{eq:primal1} also includes the $\ell_1$ penalty that could be helpful in  the cases with low SNR values. Secondly, our primal-dual method perform updating in both primal and dual spaces to approach the  saddle point, and it can potentially attain solutions with smaller duality gaps. Finally, most important of all, our objective~\eqref{eq:primal1} and the derived dual form~\eqref{eq:dual}  allow us to employ  screening and coordinate incremental strategies~\citep{Fercoq2015,GAP,Celer,ren2020thunder,atamturk2020safe} to boost the efficiency of the algorithm.

There are several obvious differences between our primal-dual method and the coordinate descent with spacer steps~(CDSS) method~\citep{Hazimeh18}. Different from our primal-dual method, CDSS utilizes coordinate descent in the primal space for parameter updating. In CDSS~\citep{Hazimeh18}, they also rely on partial swap inescapable (PSI-$\kappa$) to improve the solution with $\kappa=1$. PSI with $\kappa \geq 2$ will introduce much more extra computation that is usually not affordable.  Our primal-dual method employs coordinate incremental strategy to save computation cost, and experimental results  indicate that the proposed method can achieve similar solution quality as CDSS but with less computation time. 

In~\cite{atamturk2020safe}, the authors propose a screening method for $\ell_0$ regularized problems. However, their objective does not include the $\ell_1$ norm. Our propose method focuses on a more generalized problem that is potentially more powerful on datasets with low SNRs. Moreover, besides the safe screening rule proposed in Section~\ref{sec:sparse_prop}, the proposed coordinate incremental strategy introduced in Section~\ref{sec:increment_strategy} is empirically effective on different datasets. The screening  methods~\citep{Fercoq2015,GAP,Ndiaye2017} and coordinate incremental strategies~\citep{ren2020thunder,Celer} used for $\ell_1$ regularized problems can be taken as special cases of the proposed method with $\lambda_0=0$.

\vspace{0.2in}

\section{Conclusion}\label{sec:conclusion}

In this paper we studied the dual forms of a broad family of $\ell_0$ regularized problems. Based on the derived dual form, a primal-dual algorithm accelerated with active coordinate selection  has been developed. Our theoretical result show that the reformed best subset selection problem can be solved within  polynomial complexity. The developed theory and the proposed framework can be integrated with many screening strategies. 
Experimental results show that the proposed primal-dual method  can  reduce  redundant operations introduced by the inactive features and hence save computation costs. The proposed framework sheds light on  primal-dual algorithms that can potentially  further  scale up the solutions of $\ell_0$ regularized non-convex sparse problems.

\vspace{0.2in}

\bibliographystyle{plainnat}
\bibliography{ref}

\appendix
\newpage

In this Appendix, we provide  theoretical proofs  supporting the results in the main context.  Section~\ref{sec:proof_duality} deliveries the proofs of  Theorem~\ref{Thm:minimax} and Theorem~\ref{Thm:ball} regarding the duality; Section~\ref{sec:supp_analysis} gives the   algorithm analysis; Section~\ref{sec:dual_of_losses} presents the dual forms of two loss functions.

\section{Proofs of Theorem~\ref{Thm:minimax} and Theorem~\ref{Thm:ball}}~\label{sec:proof_duality}

In this section, we provide the proofs to extend the duality theory~\citep{pilanci2015sparse,Liu17,yuan2020dual} to the generalized sparse learning problem~\eqref{eq:primal1}.  The derivation of duality presented here is  significantly different from  the duality of hard thresholding~\citep{Liu17,yuan2020dual} because that the generalized problem  ~\eqref{eq:primal1} uses soft-regularization  terms rather than hard constrains, and it also includes the combination of three regularization norms, i.e., $\ell_0$-, $\ell_1$-, and $\ell_2$-norms. 

\begin{lemma} \label{Lem:lemma1}
For a given $\alpha \in \mathcal{F}^n$,  let $\beta(\alpha) = \argmin_{\beta} L(\beta, \alpha) $. We have 
\begin{align*} 
\min_{\beta} L(\beta, \alpha)= -  \sum_{i=1}^n  l_i^*(\alpha_i)   + \sum_{j=1}^p \Psi(\eta_j(\alpha); \lambda_0, \lambda_1,  \lambda_2),
\end{align*} 
where 
\begin{align*} 
& \Psi(\eta_j(\alpha); \lambda_0,  \lambda_1, \lambda_2)   
=    \begin{cases}
- \lambda_2 (|\eta_j(\alpha)| - \frac{\lambda_1}{2 \lambda_2})^2 + \lambda_0   &\text{if} \  \     |\eta_j(\alpha)|  > \eta_0  \\
\big\{0, - \lambda_2 (|\eta_j(\alpha)| - \frac{\lambda_1}{2 \lambda_2})^2 + \lambda_0  \big\} \quad  &\text{if}  \  \    |\eta_j(\alpha)|  = \eta_0 \\
0  &\text{if}  \  \   |\eta_j(\alpha)|  < \eta_0
\end{cases} .
 \end{align*}
To be specific, $l^*$ is the conjugate function of $l$, $\eta(\alpha) := - \frac{1}{2\lambda_2} \sum_{i=1}^n \alpha_i x_i$ and 
 $\eta_0 := \frac{ 2\sqrt{\lambda_0 \lambda_2} +  \lambda_1}{2\lambda_2}$. The  link function between $\alpha$ and $\beta$ is
\begin{align*} 
& ~~~~ \beta_j(\alpha)  =    \mathfrak{B}(\eta_j(\alpha))  = \begin{cases}
\mathrm{sign}\big(\eta_j(\alpha)\big) \big(|\eta_j(\alpha)| -\frac{\lambda_1 }{2\lambda_2} \big) &\text{if} \  \    |\eta_j(\alpha)|  > \eta_0 \\
\big\{0, \mathrm{sign}\big(\eta_j(\alpha)\big) \big(|\eta_j(\alpha)| -\frac{\lambda_1 }{2\lambda_2} \big) \big\}  &\text{if}  \  \  |\eta_j(\alpha)|  = \eta_0  \\
0  &\text{if}  \  \   |\eta_j(\alpha)|  < \eta_0
\end{cases}
.
\end{align*}
\end{lemma}

\begin{proof}
 With $l^*$ as the conjugate of $l$, the primal problem can be rewritten as 

\begin{align*} 
&\min_{\beta \in \mathbb{R}^p}  L(\beta, \alpha) \\
=&\min_{\beta \in \mathbb{R}^p} \sum_{i=1}^n  \big(\alpha_i \beta^{\top}x_i - l_i^*(\alpha_i) \big) + \lambda_0 ||\beta||_0 + \lambda_1 ||\beta||_1+ \lambda_2 ||\beta||^2_2 \\ 
=& \min_{\beta \in \mathbb{R}^p}   \sum_{i=1}^n   \big(\alpha_i \beta^{\top}x_i - l_i^*(\alpha_i) \big) + \lambda_0 ||\beta||_0 + \lambda_1 ||\beta||_1 + \lambda_2 ||\beta||^2_2\\
=&-  \sum_{i=1}^n  l_i^*(\alpha_i)  + \min_{\beta \in \mathbb{R}^p}  \sum_{i=1}^n   \big(\alpha_i \beta^{\top}x_i \big) + \lambda_0 ||\beta||_0 + \lambda_1 ||\beta||_1 + \lambda_2 ||\beta||^2_2 \\
=&  -  \sum_{i=1}^n  l_i^*(\alpha_i)   +  \min_{\beta \in \mathbb{R}^p} \sum_{j=1}^p \big(\sum_{i=1}^n \alpha_i x_{ij} \big)\beta_j  + \lambda_0 ||\beta||_0 + \lambda_1 ||\beta||_1 + \lambda_2 ||\beta||^2_2 \\
=& -  \sum_{i=1}^n  l_i^*(\alpha_i)   + \sum_{j=1}^p \Phi( \sum_{i=1}^n \alpha_i x_{ij}; \lambda_0, \lambda_1,  \lambda_2).
\end{align*}

Let $\tau_j = \alpha^{\top}  x_{\cdot j}$,
then
\begin{align*} 
 \Phi(\tau_j; \lambda_0, \lambda_1, \lambda_2) =&  \min_{\beta_j}    \tau_j \beta_j + \lambda_0 \mathbf{I} \{\beta_j \neq 0 \}+ \lambda_1|\beta_j|+ \lambda_2 \beta_j^2 \\
 =&  \min_{\beta_j}    \lambda_2 \bigg(   \beta_j^2  + \frac{\tau_j}{\lambda_2} \beta_j\bigg)^2 +  \lambda_0 \mathbf{I} \{\beta_j \neq 0 \}+ \lambda_1|\beta_j| \\
 =& \min_{\beta_j}   \lambda_2 \bigg( \beta_j + \frac{\tau_j}{2\lambda_2}\bigg)^2  +   \lambda_0 \mathbf{I} \{\beta_j \neq 0 \}+ \lambda_1|\beta_j| - \frac{\tau_j^2}{4 \lambda_2} .
\end{align*}
Let
\begin{align*} 
h(u) =   \lambda_2 \bigg( u + \frac{\tau_j}{2\lambda_2}\bigg)^2  +   \lambda_0 \mathbf{I} \{u \neq 0 \}+ \lambda_1|u| - \frac{\tau_j^2}{4 \lambda_2} .
\end{align*}

If $u \neq 0$, we have the soft-thresholding 

\begin{align*} 
\hat{u} = -\frac{1}{2\lambda_2} \mathrm{sign}(\tau_j)\max(|\tau_j| - \lambda_1,0) .
\end{align*}

With $|\tau_j| > \lambda_1$,  
\begin{align*} 
h(\hat{u}) &=  \lambda_2 \bigg( -\frac{\tau_j - \lambda_1 \mathrm{sign}(\tau_j) }{2\lambda_2}  + \frac{\tau_j}{2\lambda_2}\bigg)^2  +   \lambda_0 +\frac{ \lambda_1}{2\lambda_2}(|\tau_j| - \lambda_1) - \frac{\tau_j^2}{4 \lambda^2_2}  \\
&= \lambda_0 + \frac{\lambda_1 |\tau_j|}{2\lambda_2} - \frac{\lambda_1^2}{4\lambda_2} -  \frac{\tau_j^2}{4 \lambda_2} \\
&=  \lambda_0  - \frac{1}{4\lambda_2} \big( \lambda_1^2 + \tau_j^2 - 2\lambda_1|\tau_j|\big) \\
&=  \lambda_0  -   \frac{1}{4\lambda_2} \big( |\tau_j| -  \lambda_1 \big) ^2 .
\end{align*}
With $h(\hat{u}) < h(0)=0$, we get 
\begin{align*} 
 \lambda_0  <   \frac{1}{4\lambda_2} \big( |\tau_j| -  \lambda_1 \big) ^2  \implies   |\tau_j|  > 2\sqrt{\lambda_0 \lambda_2} +  \lambda_1 .
\end{align*}
Hence, $\hat{u}$ is the minimizer when $ |\tau_j|  > 2\sqrt{\lambda_0 \lambda_2} +  \lambda_1$.  We have
\begin{align*} 
 \Phi(\tau_j; \lambda_0, \lambda_1, \lambda_2) =  -   \frac{1}{4\lambda_2} \big( |\tau_j| -  \lambda_1 \big) ^2 +  \lambda_0  .
\end{align*}

When $2\sqrt{\lambda_0 \lambda_2} +  \lambda_1  = |\tau_j|$, both $\hat{u}$ and 0 are 
minimizer. Then 
\begin{align*} 
 \Phi(\tau_j; \lambda_0, \lambda_1, \lambda_2)  = \{ 0, -\frac{1}{4\lambda_2} (|\tau_j| - \lambda_1)^2 + \lambda_0 \}.
\end{align*}

 If $2\sqrt{\lambda_0 \lambda_2} +  \lambda_1  > |\tau_j| $, 0 is the minimizer, and $ \Psi(\tau_j, \lambda_0, \lambda_1, \lambda_2)  = 0$.

 The optimal primal can be written as 

\begin{align*} 
\beta^*_j & =    \mathfrak{A}(\tau_j; \lambda_0, \lambda_1, \lambda_2)  =  \begin{cases}
\mathrm{sign}(\tau_j) \frac{\lambda_1 - |\tau_j|}{2\lambda_2} &\text{if} \  \    |\tau_j|  > 2\sqrt{\lambda_0 \lambda_2} +  \lambda_1 \\
\{0,  \mathrm{sign}(\tau_j) \frac{\lambda_1 - |\tau_j|}{2\lambda_2}  \} \quad  &\text{if}  \  \   |\tau_j|  = 2\sqrt{\lambda_0 \lambda_2} +  \lambda_1 \\
0  &\text{if}  \  \   |\tau_j|   < 2\sqrt{\lambda_0 \lambda_2} +  \lambda_1
\end{cases} \\
&= \mathfrak{B}(\eta_j(\alpha)) .
\end{align*}
Here $\eta_0 = \frac{ 2\sqrt{\lambda_0 \lambda_2} +  \lambda_1}{2\lambda_2} $.
With the optimal $\beta(\alpha)$, $L(\beta, \alpha)$can be written as 
\begin{align*} 
 \Phi(\tau_j; \lambda_0,  \lambda_1, \lambda_2)& =   \begin{cases}
-\frac{1}{4\lambda_2} (|\tau_j| - \lambda_1)^2  + \lambda_0   &\text{if} \  \    |\tau_j|  > 2\sqrt{\lambda_0 \lambda_2} +  \lambda_1 \\
\{0,  -\frac{1}{4\lambda_2} (|\tau_j| - \lambda_1)^2 + \lambda_0   \} \quad  &\text{if}  \  \   |\tau_j|  = 2\sqrt{\lambda_0 \lambda_2} +  \lambda_1 \\
0  &\text{if}  \  \   |\tau_j|   < 2\sqrt{\lambda_0 \lambda_2} +  \lambda_1
\end{cases} .
\end{align*}
Alternatively, 
\begin{align*} 
 \Psi(\eta_j(\alpha); \lambda_0,  \lambda_1, \lambda_2)   
&= \Phi(\tau_j; \lambda_0,  \lambda_1, \lambda_2)  \\
 & =   \begin{cases}
- \lambda_2 (|\eta_j(\alpha)| - \frac{\lambda_1}{2 \lambda_2})^2 + \lambda_0   &\text{if} \  \     |\eta_j(\alpha)|  >  \eta_0 \\
\big\{0, - \lambda_2 (|\eta_j(\alpha)| - \frac{\lambda_1}{2 \lambda_2})^2 + \lambda_0  \big\} \quad  &\text{if}  \  \    |\eta_j(\alpha)|  = \eta_0 \\
0  &\text{if}  \  \   |\eta_j(\alpha)|  < \eta_0
\end{cases} .
\end{align*}
Here $\eta(\alpha) = - \frac{1}{2\lambda_2} \sum_{i=1}^n \alpha_i x_i$, and $\eta_0 = \frac{ 2\sqrt{\lambda_0 \lambda_2} +  \lambda_1}{2\lambda_2} $. It concludes the lemma. 
\end{proof}

\begin{lemma} \label{Lem:saddle_point}
 (Saddle Point). Let $\bar{\beta}$ be a primal vector and $\bar{\alpha}$ a  dual vector. Then $ (\bar{\alpha}, \bar{\beta})$ is a   saddle point of $L$ if and only if the following conditions hold:
 
 a) $\bar{\beta}$ solves the primal problem;
 
 b) $\bar{\alpha} \in [\partial l_1(\bar{\beta}^{\top}x_1),  \partial l_2(\bar{\beta}^{\top}x_2), ..., \partial l_n(\bar{\beta}^{\top}x_n)]^{\top}$;
 
 c) $\bar{\beta}_j = \mathfrak{B}(\eta_j(\bar{\alpha}))$.
\end{lemma}

\begin{proof}
$\Leftarrow$:  If the pair $(\bar{\beta}, \bar{\alpha})$ is a  saddle point for $L$, then from the definition of conjugate convexity and inequality in the definition of  saddle point we have 
\begin{align*} 
P(\bar{\beta}) = \max_{\alpha \in \mathcal{F}^n} L(\bar{\beta}, \alpha) \leq  L(\bar{\beta}, \bar{\alpha})  \leq \min_{\beta \in \mathbb{R}^p } L(\beta, \bar{\alpha}) .
 \end{align*}
On the other hand, we know that for any $\beta \in \mathbb{R}^p$ and $\alpha \in \mathcal{F}^n$
\begin{align*} 
L(\beta, \alpha) \leq \max_{\alpha' \in \mathcal{F}^n} L(\beta, \alpha') = P(\beta).
 \end{align*}
By combining the preceding two inequalities we obtain
\begin{align*} 
P(\bar{\beta}) \leq \min_{\beta \in \mathbb{R}^p} L(\beta, \bar{\alpha}) \leq \min_{\beta \in \mathbb{R}^p} P(\beta) \leq P(\bar{\beta}).
 \end{align*}
Therefore $P(\bar{\beta}) = \min_{ \beta \in \mathbb{R}^p } P(\beta)$, i.e., $\bar{\beta}$ solves the primal problem , which proves the necessary condition~(a). Moreover, the above arguments lead to 
\begin{align*} 
P(\bar{\beta}) = \max_{\alpha  \in \mathcal{F}^n } L(\bar{\beta}, \alpha) = L(\bar{\beta}, \bar{\alpha}) .
 \end{align*}
Then from the maximizing argument property of convex conjugate we have $\bar{\alpha}_i \in \partial l_i(\bar{\beta}^{\top} x_i) $, and it concludes condition b).
Note that 
\begin{align} \label{eq:L_Thm1}
L(\beta, \bar{\alpha}) = \lambda_2\bigg|\bigg| \beta + \frac{1}{ 2\lambda_2} \sum_{i=1}^n \bar{\alpha}_i x_i \bigg|\bigg|^2 -   \sum_{i=1}^n l_i^*(\bar{\alpha}_i) + \lambda_1||\beta||_1 + \lambda_0 ||\beta||_0+C .
\end{align}
Let $\bar{F} = \mathrm{supp}(\bar{\beta})$. Since the above analysis implies $L(\bar{\beta}, \bar{\alpha})=\min_{\beta} L(\beta, \bar{\alpha})$, with $\eta(\bar{\alpha}) = - \frac{1}{2\lambda_2} \sum_{i=1}^n \bar{\alpha}_i x_i$, it must hold that (more details refer to the proof of Lemma~\ref{Lem:lemma1})
$\bar{\beta}_j = \mathfrak{B}(\eta_j(\bar{\alpha}))$.
This validates the condition (c).

$\Rightarrow$: Inversely, let us assume that $\bar{\beta}$ is a solution to the primal problem~(condition a)), and $\bar{\alpha}_i \in \partial l_i(\bar{\beta}^{\top}x_i)$~(condition b)). Again from the maximizing argument property of convex conjugate we know that $l_i(\bar{\beta}^{\top} x_i) = \bar{\alpha} \bar{\beta}^{\top}x_i - l_i^*(\bar{\alpha}_i)$. This leads to 
\begin{align} 
L(\bar{\beta}, \alpha) \leq  P(\bar{\beta}) = \max_{\alpha \in \mathcal{F}^n} L(\bar{\beta}, \alpha) = \ L(\bar{\beta}, \bar{\alpha}). \label{eq:ThOne1_2}
 \end{align}
The sufficient condition~(c) guarantees that based on the expression of \eqref{eq:L_Thm1}, for any $\beta$, we have 
\begin{align} 
L(\bar{\beta}, \bar{\alpha}) \leq  L(\beta, \bar{\alpha}) \label{eq:ThOne2}.
\end{align}
By combining the inequalities~\eqref{eq:ThOne1_2} and~\eqref{eq:ThOne2} we get that for any $\beta$ and $\alpha$ 
 \begin{align*}
L(\bar{\beta}, \alpha) \leq L(\bar{\beta}, \bar{\alpha})  \leq L( \beta, \bar{\alpha}) .
\end{align*}
This shows that $(\bar{\beta}, \bar{\alpha})$ is a  saddle point of the Lagrangian $L$. 
\end{proof}

\noindent\textbf{Theorem~\ref{Thm:minimax}} \  {\it 
Let $\bar{\beta} \in \mathbb{R}^p $ be a primal vector and $\bar{\alpha} \in \mathcal{F}^n$  regarding L, then
\vspace{-0.1in}

\begin{enumerate} 
\item $ (\bar{\alpha}, \bar{\beta})$ is a  saddle point of $L$ if and only if the following conditions hold:
\begin{enumerate}
\item $\bar{\beta}$ solves the primal problem; 
 
\item  $\bar{\alpha} \in [\partial l_1(\bar{\beta}^{\top}x_1),  \partial l_2(\bar{\beta}^{\top}x_2), ..., \partial l_n(\bar{\beta}^{\top}x_n)]^{\top}$;
\item
$\bar{\beta}_j =   \mathfrak{B}(\eta_j(\bar{\alpha}))$.
\end{enumerate}

 \item The mini-max relationship 
\begin{align} \label{eq:th2}
\max_{\alpha \in \mathcal{F}^n} \min_{\beta} L(\beta, \alpha) = \min_{\beta} \max_{\alpha \in  \mathcal{F}^n} L(\beta, \alpha) .
\end{align}
holds if and only if there exists a saddle point $(\bar{\beta}, \bar{\alpha})$ for L. 

\item The corresponding dual problem of~\eqref{eq:primal1} is  written as
\begin{align} \label{eq:dual}
  &\max_{\alpha \in \mathcal{F}^n} D(\alpha) =  \max_{\alpha \in \mathcal{F}^n}-  \sum_{i=1}^n  l_i^*(\alpha_i)   + \sum_{j=1}^p \Psi(\eta_j(\alpha); \lambda_0,  \lambda_1, \lambda_2),
\end{align} 
where  $l^*$ is the conjugate function of $l$.  The primal dual link is written as $\beta_j(\alpha)  =  \mathfrak{B}(\eta_j(\alpha))$.

\item (Strong duality) $\bar{\alpha}$ solves the dual problem in~\eqref{eq:dual}, i.e., $D(\bar{\alpha}) \geq D(\alpha), \alpha \in \mathcal{F}^n$, and $P(\bar{\beta}) = D(\bar{\alpha})$ if and only if the pair $(\bar{\beta}, \bar{\alpha})$ satisfies the  three conditions given by (a)$\sim$(c).

\end{enumerate}
}

\begin{proof}
According to Lemma~\ref{Lem:saddle_point}, statement-1 can be proved. We focus on statements 2-4. The following Part (I),  Part (II), and Part (III) are proofs of statement-2, statement-3, and statement-4, respectively. 

Part (I): the mimi-max relationship in  statement-2.

$\Rightarrow$:
Let $(\bar{\beta}, \bar{\alpha})$ be a  saddle point for $L$. On one hand, note that the following holds for any $\beta'$ and $\alpha'$,
\begin{align*} 
\min_{\beta} L(\beta, \alpha') \leq  L(\beta', \alpha')  \leq \max_{\alpha \in \mathcal{F}^n} L(\beta', \alpha) 
\end{align*}
which implies
\begin{align} 
\max_{\alpha \in \mathcal{F}^n }\min_{\beta} L(\beta, \alpha)  \leq \min_{\beta} \max_{\alpha \in \mathcal{F}^n} L(\beta, \alpha) . \label{eq:th2_1}
\end{align}
On the other hand, since $(\bar{\beta}, \bar{\alpha})$ is a  saddle point for $L$, the following is true:
\begin{align}
& ~ \max_{\alpha \in \mathcal{F}^n}\min_{\beta} L(\beta, \alpha)  \leq  \max_{\alpha \in \mathcal{F}^n } L(\bar{\beta}, \alpha)  \nonumber \\
& \leq  \max_{\alpha \in \mathcal{F}^n} L(\bar{\beta},  \bar{\alpha})  \leq \min_{\beta} L(\beta,\bar{ \alpha}) 
 \leq \max_{\alpha \in \mathcal{F}^n }\min_{\beta} L(\beta, \alpha) \label{eq:th2_2} .
\end{align}
By combining~\eqref{eq:th2_1} and~\eqref{eq:th2_2} we prove the equality~\eqref{eq:th2}. 

$\Leftarrow$: 
Assume that the equality in~\eqref{eq:th2} holds. Let us define $\bar{\beta}$ and $\bar{\alpha}$ such that 
\[\max_{\alpha \in \mathcal{F}^n} L(\bar{\beta}, \alpha) = \min_{\beta}  \max_{\alpha \in \mathcal{F}^n } L(\beta, \alpha) \] 
and 
\[\min_{\beta} L(\beta, \bar{\alpha}) = \max_{\alpha \in \mathcal{F}^n}  \min_{\beta}  L(\beta, \alpha) .
\]
Then we can see that for any $\alpha$, $L(\bar{\beta}, \bar{\alpha}) \geq \min_{\beta} L(\beta, \bar{\alpha}) = \max_{\alpha' \in \mathcal{F}^n} L(\bar{\beta}, \alpha')$, 
where the ``='' is due to~\eqref{eq:th2}. In the meantime, for any $\beta$
 \begin{align*} 
L(\bar{\beta}, \bar{\alpha}) \leq \max_{\alpha \in \mathcal{F}^n} L(\bar{\beta}, \alpha) =  \min_{\beta} L(\beta', \bar{\alpha}) \leq L(\beta, \bar{\alpha}).
\end{align*}
This shows that $(\bar{\beta}, \bar{\alpha})$ is a  saddle point for L.

Part (II): the dual form in statement-3.

According to Lemma~\ref{Lem:lemma1}, for any $\alpha$, the $\beta$ that minimizes $L(\beta, \alpha)$ is 
\begin{align} \label{eq:beta_j}
\beta_j (\alpha)  =    \mathfrak{B}(\eta_j(\alpha))  = \begin{cases}
\mathrm{sign}\big(\eta_j(\alpha)\big) \big(|\eta_j(\alpha)| -\frac{\lambda_1 }{2\lambda_2} \big) &\text{if} \  \    |\eta_j(\alpha)|  > \eta_0 \\
\big\{0, \mathrm{sign}\big(\eta_j(\alpha)\big) \big(|\eta_j(\alpha)| -\frac{\lambda_1 }{2\lambda_2} \big) \big\} \quad  &\text{if}  \  \  |\eta_j(\alpha)|  = \eta_0  \\
0  &\text{if}  \  \   |\eta_j(\alpha)|  < \eta_0
\end{cases}
.
\end{align}

Then we have 
\begin{align} \label{eq:dual1}
D(\alpha) = -  \sum_{i=1}^n  l_i^*(\alpha_i)   + \sum_{j=1}^p \Psi(\eta_j ; \lambda_0, \lambda_1,  \lambda_2),
\end{align} 
where 
\begin{align} \notag
 \Psi(\eta_j(\alpha); \lambda_0,  \lambda_1, \lambda_2) &=   \begin{cases}
- \lambda_2 (|\eta_j(\alpha)| - \frac{\lambda_1}{2 \lambda_2})^2 + \lambda_0   &\text{if} \  \     |\eta_j(\alpha)|  > \eta_0  \\
\big\{0, - \lambda_2 (|\eta_j(\alpha)| - \frac{\lambda_1}{2 \lambda_2})^2 + \lambda_0  \big\} \quad  &\text{if}  \  \    |\eta_j(\alpha)|  = \eta_0 \\
0  &\text{if}  \  \   |\eta_j(\alpha)|  < \eta_0
\end{cases} \\ \label{eq:dual_normal}
&= - \lambda_2 ||\beta(\alpha)||^2_2 + \lambda_0 ||\beta(\alpha)||_0 .
 \end{align}
Here $\eta_0 = \frac{ 2\sqrt{\lambda_0 \lambda_2} +  \lambda_1}{2\lambda_2} $. Assume we have two arbitrary dual variables $\alpha_1, \alpha_2 \in \mathcal{F}^n$ and any $g(\alpha_2) \in [\beta(\alpha_2)^{\top} x_1 - {l_1^*}^{'}(\alpha_{2(1)}), ..., \beta(\alpha_2)^{\top} x_n - {l_n^*}^{'}(\alpha_{2(n)})] $. Here $\alpha_{2(n)}$ is the $n$th entry of $\alpha_{2}$. $L(\beta, \alpha)$ is concave in terms of $\alpha$ given any fixed $\beta$. According to the definition of $D(\alpha)$, we have 
 \begin{align*}
 D(\alpha_1) = L(\beta(\alpha_1), \alpha_1) \leq L(\beta(\alpha_2), \alpha_1)  \leq L(\beta(\alpha_2), \alpha_2)  +\langle g(\alpha_2) , \alpha_1 - \alpha_2 \rangle .
 \end{align*}
Hence $D(\alpha)$ is concave and the super gradient is as given. 

Part (III): Strong duality.
$\Rightarrow$:
Given the conditions a)-c), we can see that the pair~($\bar{\beta}$, $\bar{\alpha}$) forms a  saddle point of $L$. Thus based on the definitions of  saddle 
point and dual function  $D(\alpha)$, we can show that 
\begin{align*} 
D(\bar{\alpha}) = \min_{\beta} L(\beta, \bar{\alpha}) \geq  L(\bar{\beta}, \bar{\alpha}) \geq  L(\bar{\beta}, \alpha) \geq  D(\alpha) .
\end{align*}
This implies that $\bar{\alpha}$ solves the dual problem. Furthermore, Theorem~\ref{Thm:minimax}-2 guarantees the following 
\begin{align*} 
D(\bar{\alpha}) =\max_{\alpha \in \mathcal{F}^n} \min_{\beta} L(\beta, \alpha) =\min_{\beta} \max_{\alpha \in \mathcal{F}^n}  L(\beta, \alpha) = P(\bar{\beta}) . 
\end{align*}
This indicates that the primal and dual optimal values are equal to each other.

$\Leftarrow$: 
Assume that $\bar{\alpha}$ solves the dual problem in  and $D(\bar{\alpha}) =  P(\bar{\beta}) $. Since $D(\bar{\alpha}) \leq  P(\beta) $ holds for any $\beta$, $\bar{\beta}$ must be the sparse
minimizer of $P(\beta)$. It follows that 
\begin{align*} 
\max_{\alpha \in \mathcal{F}^n} \min_{\beta} L(\beta, \alpha) = D(\bar{\alpha})  = P(\bar{\beta}) =\min_{\beta} \max_{\alpha \in \mathcal{F}^n}  L(\beta, \alpha) . 
\end{align*}
From the $\Leftarrow$ argument the proof of Theorem~\ref{Thm:minimax}-2 and we get that conditions a)-c)in Theorem~\ref{Thm:minimax}-1 should be satisfied for $(\bar{\beta}, \bar{\alpha})$. This completes the proof.
\end{proof}

Compared to the dual problem developed in~\cite{Liu17,yuan2020dual} regarding  hard thresholding, the soft thresholding in \eqref{eq:beta_j}  corresponds to the combination of $\ell_1$ and $\ell_0$ penalties, and it is helpful on datasets with low SNR values~\citep{pmlr-v65-david17a,mazumder2022subset}. \\

\noindent\textbf{Theorem~\ref{Thm:ball}} \  {\it
Assume that the primal loss functions $\{l_i(\cdot)\}_{i=1}^n$ are $1/\mu$-strongly smooth.  The range of the dual variable is  bounded via the duality gap value, i.e.,  $\forall \alpha \in \mathcal{F}^n, \beta \in \mathbb{R}^p$, $\{B(\alpha; r) : || \alpha -  \bar{\alpha}||_2  \leq r, r = \sqrt{ \frac{2(P(\beta) - D(\alpha)) } {\gamma}} \} $. 
Here $\gamma$ is a positive constant and $\gamma \geq \mu$. 
}

\begin{proof}
 With the assumption $l_i$ being $1/\mu$-smooth, its  conjugate function $l^*_i$ is $\mu$-strongly convex. 
With the dual problem 
\begin{align*} 
D(\alpha)& = -  \sum_{i=1}^n  l_i^*(\alpha_i)   + \sum_{j=1}^p \Psi(\eta_j(\alpha) ; \lambda_0, \lambda_1,  \lambda_2)  .
\end{align*} 
Here
\begin{align*} 
&\Psi(\eta_j(\alpha); \lambda_0,  \lambda_1, \lambda_2) 
 =   \begin{cases}
- \lambda_2 (|\eta_j(\alpha)| - \frac{\lambda_1}{2 \lambda_2})^2 + \lambda_0   &\text{if} \  \     |\eta_j(\alpha)|  \geq \eta_0 \\
0  &\text{if}  \  \   |\eta_j(\alpha)|  < \eta_0
\end{cases} .
\end{align*}
With $\eta_0 = \frac{ 2\sqrt{\lambda_0 \lambda_2} +  \lambda_1}{2\lambda_2}$ and  $\eta_j(\alpha) = -\frac{1}{2\lambda_2} x^{\top}_{\cdot j}  \alpha$, we get $\Psi(\eta_j(\alpha); \lambda_0,  \lambda_1, \lambda_2)$ is concave for $j, 1 \leq j \leq p$ regarding $\alpha$, hence $D(\alpha)$ is concave. For given hyper-parameters $\lambda_0, \lambda_1, \lambda_2$, let $\Psi(\alpha)=\sum_{j=1}^p\Psi(\eta_j(\alpha); \lambda_0,  \lambda_1, \lambda_2)$, and $\Psi_j(\alpha)=\Psi(\eta_j(\alpha); \lambda_0,  \lambda_1, \lambda_2)$.  For $1 \leq j \leq p$,
\begin{align*}
\Psi_j'(\alpha) = & \begin{cases}
\mathrm{sign}(\eta_j(\alpha)) (|\eta_j(\alpha)| - \frac{\lambda_1}{2\lambda_2}) x_{\cdot j}  &\text{if} \  \     |\eta_j(\alpha)|  \geq \eta_0 \\
0  &\text{if}  \  \   |\eta_j(\alpha)|  < \eta_0
\end{cases} \\
=& \begin{cases}
\beta_j x_{\cdot j}  &\text{if} \  \     |\eta_j(\alpha)|  \geq \eta_0 \\
0  &\text{if}  \  \   |\eta_j(\alpha)|  < \eta_0
\end{cases}  .
\end{align*} 
Hence,
\begin{align*}
\Psi_j''(\alpha) = & \begin{cases}
-\frac{1}{2\lambda_2} x_{\cdot j}  x^{\top}_{\cdot j}  &\text{if} \  \     |\eta_j(\alpha)|  \geq \eta_0 \\
0  &\text{if}  \  \   |\eta_j(\alpha)|  < \eta_0
\end{cases}  .
\end{align*} 
Therefore, 
\begin{align*}
\Psi''(\alpha) =  -\frac{1}{2\lambda_2} \sum_{j\in S}x_{\cdot j}x^{\top}_{\cdot j}.
\end{align*} 
Here $S$ is the support set regarding $\alpha$, i.e. $S=\mathrm{supp}\big(\beta(\alpha)\big)=\{j \big| |\eta_j(\alpha)| \geq \eta_0\}=\{j \big| |x_{\cdot j}^{\top} \alpha| \geq 2 \lambda_2 \eta_0\}$.  
The smallest eigenvalue of Hessian matrix  $\Psi''(\alpha)$ depends on $X_S$. Let $\sigma_{min}(X_S)$ be the smallest eigenvalue of $X_S$, then $\Psi(\alpha)$ is concave, and it is also $\nu = \frac{\sigma_{min}(X_S)}{2 \lambda_2}$-strongly concave at point $\alpha$. 

 
With $l^*_i$ $\mu$-strongly convex, $D(\alpha)$ is $\gamma$-strongly concave with $\gamma = \mu + \inf_{S} \frac{\sigma_{min}(X_S)}{2 \lambda_2} \geq \mu$. Then we have
\begin{align*}
D(\alpha_1) \leq  D(\alpha_2) + \langle\nabla_{\alpha} D(\alpha_2),  \alpha_1 - \alpha_2 \rangle - \frac{\gamma}{2} || \alpha_1 - \alpha_2 ||^2 .
\end{align*}
Let $\alpha_2 = \bar{\alpha}$, and $\alpha_1 = \alpha \in \mathcal{F}^n$. As $\bar{\alpha}$ maximizes $D(\alpha)$, $\langle\nabla_{\alpha} D(\bar{\alpha}),  \alpha - \bar{\alpha} \rangle \leq 0$. It implies
 \begin{align*}
D(\alpha) \leq  D(\bar{\alpha}) - \frac{\gamma}{2} || \alpha - \bar{\alpha} ||^2 .
\end{align*}
Thus we have a ball range for the dual variable
\begin{align*}
 \forall \alpha \in \mathcal{F}^n, \beta \in \mathbb{R}^p,  || \alpha -  \bar{\alpha}||_2  \leq \sqrt{ \frac{2(P(\beta) - D(\alpha)) } {\gamma}} =: r.
 \end{align*}
 This completes the proof.
\end{proof}


\newpage

\section{Algorithm Analysis}\label{sec:supp_analysis}
In this section, we  present the complexity analysis of the inner updating Algorithm~\ref{alg:primal_dual_inner}. 

\subsection{Convergence of Inner Primal-dual Updating Algorithm}
 We will show that under certain conditions $\beta(\alpha)$ is locally smooth around $\bar{\beta} = \beta(\bar{\alpha})$. For a given set of parameters $\lambda = \{\lambda_0, \lambda_1, \lambda_2\}$, $\beta(\alpha)$ corresponds to a set of support features $\mathrm{supp}(\beta(\alpha))$. We use $\bar{\delta}$ to represent the set in the dual feasible space that $\mathrm{supp}(\beta(\alpha)) = \mathrm{supp}(\bar{\beta})$.


\begin{lemma}\label{Lem:support}
Let $X = [x_1, ..., x_n]^T \in \mathbb{R}^{n \times p}$ be the data matrix,  $\eta_{j}(\bar{\alpha})$ be the $j$th entry of $\eta(\bar{\alpha})$, $S=\mathrm{supp}(\bar{\beta})$, and $N= \{j| \eta_j(\bar{\alpha}) = \eta_{0(j)} \}$. Assume that $\{l_i\}_{i =1,..., n}$ are differentiable, and  let
\begin{align*}
\bar{\delta} =: 2 \lambda_2 \min \bigg\{\min_{j:j \in S} \frac{|\eta_j(\bar{\alpha})  |  -  \eta_{0}}{||x_{\cdot j}||}, 
\min_{j:j \in S^c \setminus N} \frac{ \eta_{0} - |\eta_j(\bar{\alpha})|}{||x_{\cdot j}||}\bigg\},
\end{align*}
with $|| \alpha - \bar{\alpha}|| \leq \bar{\delta}$, we have $\mathrm{supp}(\beta(\alpha)) = \mathrm{supp}(\bar{\beta})$, and $|| \beta(\alpha) - \bar{\beta}|| \leq \frac{\sigma_{max}(X_S)}{2\lambda_2} ||\alpha - \bar{\alpha}||$.
\end{lemma}

\begin{proof}
For any $\alpha $, we have
\begin{align}\label{eq:eta_alp}
\eta(\alpha) = - \frac{1}{2\lambda_2} \sum_{i=1}^n \alpha_i x_i  = - \frac{1}{2\lambda_2}  X^{\top}\alpha .
\end{align}
For a feature $j \in S= \mathrm{supp}(\bar{\beta})$,  we have $|\eta_j(\bar\alpha)| - \eta_{0} > 0 $, which is
\begin{align*}
& |\eta_j(\bar{\alpha}) + \eta_j(\alpha) - \eta_j(\bar{\alpha}))|  >  \eta_{0} .
\end{align*}
We try to find the space for $\alpha$s that have the same support as $\bar{\alpha}$. We use the lower bound of above inequality,
\begin{align*}
& |\eta_j(\bar{\alpha}) + \eta_j(\alpha) - \eta_j(\bar{\alpha}))| \geq |\eta_j(\bar{\alpha})  | - |\eta_j(\alpha) - \eta_j(\bar{\alpha})|   >  \eta_{0},
\end{align*}
yields 
\begin{align*}
& |\eta_j(\alpha) - \eta_j(\bar{\alpha})| < |\eta_j(\bar{\alpha})|-\eta_{0}.
\end{align*}
With~\eqref{eq:eta_alp},
\begin{align*}
&  \frac{\big|\big| \alpha - \bar{\alpha} \big|\big|}{2\lambda_2} \big|\big| x_{\cdot j} \big|\big|   < |\eta_j(\bar{\alpha})  |  -  \eta_{0}.
\end{align*}
Hence
\begin{align*}
&\big|\big| \alpha - \bar{\alpha} \big|\big| < \min_{j:j \in S} \frac{2 \lambda_2 (|\eta_j(\bar{\alpha})  |  -  \eta_{0})}{||x_{\cdot j}||}.
\end{align*}
Similarly, for features $j$, $j \notin S $ and  $j \notin N$, 
\begin{align*}
&|\eta_j(\bar\alpha)|  < \eta_{0}, \end{align*}
yields
\begin{align*}
& |\eta_j(\alpha) - \eta_j(\bar{\alpha})| < \eta_{0} - |\eta_j(\bar{\alpha})|.
\end{align*}
With all $j$s, $j \in S^c \setminus N$
\begin{align*}
&\big|\big| \alpha - \bar{\alpha} \big|\big| < \min_{j:j \in S^c \setminus N} \frac{2 \lambda_2 (  \eta_{0} - |\eta_j(\bar{\alpha})|)}{||x_{\cdot j}||}. 
\end{align*}
Therefore, if $\alpha \in \bar{\delta}$,with 
\begin{align*}
|| \alpha - \bar{\alpha}|| \leq  \bar{\delta} = 2 \lambda_2 \min \bigg\{\min_{j:j \in S} \frac{|\eta_j(\bar{\alpha})  |  -  \eta_{0}}{||x_{\cdot j}||}, 
\min_{j:j \in S^c \setminus N} \frac{ \eta_{0} - |\eta_j(\bar{\alpha})|}{||x_{\cdot j}||}\bigg\}
\end{align*}
we have $\mathrm{supp}(\beta(\alpha)) = \mathrm{supp}(\bar{\beta})$. With $|| \alpha - \bar{\alpha}|| \leq \bar{\delta}$,  the primal problem  becomes a convex $\ell_1$ regularization problem without any redundant features. The super-gradient in Remark~\ref{rmk:sup-grad} becomes $g_{\alpha} = X_{S}\beta_S(\alpha) - l^{*'}(\alpha)$. With   $||g_{\alpha}||\rightarrow 0$,  $\beta_S(\alpha) = (X_{S}^{\top}X_{S})^{-1} X_{S}^{\top} l^{*'}(\alpha)$. As $S$ is fixed, with $||\alpha - \bar{\alpha}||$ a small value, we have  $\mathrm{sign}(\beta(\alpha)) = \mathrm{sign}(\bar{\beta})$. It means
\begin{align*}
|| \beta(\alpha) - \bar{\beta}|| &= ||  \mathfrak{B}(\eta(\alpha)) -  \mathfrak{B}(\eta(\bar{\alpha}))|| \\
&\leq || \eta(\alpha) -  \eta(\bar{\alpha})|| = \frac{1}{2\lambda_2} || X_S(\alpha - \bar{\alpha})|| \\
&\leq \frac{\sigma_{max}(X_S)}{2\lambda_2} ||\alpha - \bar{\alpha}|| .
\end{align*}
It finishes the proof of the lemma.
\end{proof} 

Note that the above lemma can be extended to any pair of $\alpha_1, \alpha_2 \in \mathcal{F}^n$, and if they are close enough, they have the same support set. Let $\Psi( \alpha; \lambda_0, \lambda_1,  \lambda_2) = \sum_{j=1}^p \Psi(\eta_j(\alpha) ; \lambda_0, \lambda_1,  \lambda_2)  $, and it is easy to verify that  $\Psi( \alpha; \lambda_0, \lambda_1,  \lambda_2)$ is concave. 


\begin{lemma}\label{Lem:dual_range}
Assume that the primal loss functions $\{l_i()\}_{i=1}^n$ are $1/\mu$-strongly smooth. Then the following inequality holds for any $\alpha_1, \alpha_2 \in \mathcal{F}^n$,  and $g(\alpha_2) \in \partial D(\alpha_2)$:
 \begin{align*}
D(\alpha_1) \leq D(\alpha_2) +  \langle g_{\alpha_2}, \alpha_1 -   \alpha_2  \rangle -  \frac{  \gamma }{2 } ||  \alpha_1 - \alpha_2 ||^2.
 \end{align*}
 Moreover, $\forall \alpha \in \mathcal{F}^n$, and $g_{\alpha} \in \partial D(\alpha)$, $||\alpha - \bar{\alpha}|| \leq \sqrt{ \frac{2   \langle g_{\alpha}, \bar{ \alpha}- \alpha \rangle  }{ \gamma}} .$
 Here, $\gamma$ is the same in Theorem~\ref{Thm:ball}.
\end{lemma}
\begin{proof}
With the assumption $l_i$ being $1/\mu$-smooth, its  conjugate function $l^*_i$ is $\mu$-strongly convex. 
With the dual problem 
\begin{align*} 
D(\alpha)& = -  \sum_{i=1}^n  l_i^*(\alpha_i)   + \sum_{j=1}^p \Psi(\eta_j(\alpha) ; \lambda_0, \lambda_1,  \lambda_2)  .
\end{align*} 
Here
\begin{align*} 
&\Psi(\eta_j(\alpha); \lambda_0,  \lambda_1, \lambda_2) 
 =   \begin{cases}
- \lambda_2 (|\eta_j| - \frac{\lambda_1}{2 \lambda_2})^2 + \lambda_0   &\text{if} \  \     |\eta_j(\alpha)|  \geq \eta_0 \\
0  &\text{if}  \  \   |\eta_j(\alpha)|  < \eta_0
\end{cases} .
\end{align*}
According to the proof of Theorem~\ref{Thm:ball}, $\Psi(\alpha)= \sum_{j=1}^p \Psi(\eta_j(\alpha) ; \lambda_0, \lambda_1,  \lambda_2)$ is $\nu$-strongly concave with  $\nu = \inf_{\alpha} \nu(\alpha) = \frac{\sigma_{min}(X_{S_{\alpha}})}{2 \lambda}$, and $S_{\alpha} = \mathrm{supp}\big(\beta(\alpha)\big)$. 
When $S_{\alpha} = \emptyset$, we have $\nu(\alpha)=0$. 
 Now let us consider two arbitrary dual variables $\alpha_1, \alpha_2 \in \mathcal{F}^n$,
\begin{align*} 
\Psi(\alpha_1) \leq  \Psi(\alpha_2) + \Psi'(\alpha_2)^{\top}(\alpha_1 - \alpha_2)   -  \frac{\nu}{2}||\alpha_1  -  \alpha_2||^2   .
\end{align*} 
Hence,
\begin{align}\label{eq:dual1}
D(\alpha_1) = &-  \sum_{i=1}^n  l_i^*(\alpha_{1(i)})   + \sum_{j=1}^p \Psi(\eta_j(\alpha_1) ; \lambda_0, \lambda_1,  \lambda_2)  \notag \\
\leq & \sum_{i=1}^n\big( -l_i^*(\alpha_{2(i)}) - l_i^{*'}(\alpha_{2(i)}) (\alpha_{1(i)} - \alpha_{2(i)}) - \frac{\mu}{2}(\alpha_{1(i)}  - \alpha_{2(i)})^2\big) \notag \\
&+ \Psi(\alpha_2) + \Psi'(\alpha_2)^{\top}(\alpha_1 - \alpha_2)   - \frac{\nu}{2}||\alpha_1  - \alpha_2||^2   \notag \\
 \leq & D(\alpha_2) +  \langle g_{\alpha_2}, \alpha_1 -   \alpha_2  \rangle -  \frac{  \gamma }{2 } ||  \alpha_1 - \alpha_2 ||^2.
\end{align} 
Here $\alpha_{1(i)}$ is the $i$th entry of $\alpha_1$. This proves the first desirable inequality in the lemma. With the above inequality and using the fact $D(\alpha) \leq D(\bar{\alpha})$ we get that 
 \begin{align*}
D(\bar{\alpha})  \leq& D(\alpha) +  \langle g_{\alpha}, \bar{\alpha}-   \alpha \rangle -    \frac{  \gamma }{2 } ||  \alpha- \bar{\alpha} ||^2 
\leq   D(\bar{\alpha}) +  \langle g_{\alpha}, \bar{\alpha}-   \alpha \rangle -   \frac{  \gamma}{2 } ||  \alpha- \bar{\alpha} ||^2,
 \end{align*}
which leads to the second desired bound,
\begin{align*}
||\alpha - \bar{\alpha}|| \leq \sqrt{ \frac{2   \langle g_{\alpha}, \bar{ \alpha}- \alpha \rangle  }{ \mu + \nu(\alpha)}} .
\end{align*}
It concludes the proof of the lemma
\end{proof}


Different from the primal updating~\citep{Hazimeh18} or dual updating~\citep{Liu17} algorithms,   Algorithm~\ref{alg:primal_dual_inner} has both primal and dual updating steps.
Let $m_1 =\max_{j:1\leq j \leq p} |y^{\top}x_{\cdot j}|$, and $m_2 =\max_{j:1\leq j \leq p} ||X^{\top}x_{\cdot j}||$, $m_3 = \max_{i, \alpha_i^t \in \mathcal{F}} |l_i^{*'}(\alpha_i^t)|$, and  $\varrho = \sqrt{n}||\alpha^t||_{\infty} - \lambda_1$. We have the following theorem regarding the convergence of Algorithm~\ref{alg:primal_dual_inner}.\\

\noindent\textbf{Theorem~\ref{Thm:inner_convg}} 
{\it Assume that $l_i$ is $1/\mu$-smooth,  $||x_{i}|| \leq \vartheta  \ \forall 1\leq i \leq n$, and $||x_{\cdot j}||= 1  \ \forall 1\leq j \leq p$. 
By choosing $w_t = \frac{1}{t \gamma}$, then the sequence generated by Algorithm~\ref{alg:primal_dual_inner} satisfies the following estimation error inequality:
 \begin{align*}
|| \alpha^t - \bar{\alpha}||^2 \leq c_1\bigg( \frac{1}{t} +  \frac{\ln t}{t}  \bigg).
\end{align*}
Here $c_1 = \frac{   c_0^2}{ \mu ^2 } $. 
$ c_0=  \frac{ \sqrt{np} \vartheta }{2\lambda_2(1+2\lambda_2)}(2\lambda_2m_1 +  \varrho +  \sqrt{n}m_2 \varrho - 2\lambda_1\lambda_2)+\sqrt{n}m_3$,
$gamma$ is same as in Theorem~\ref{Thm:ball}. 
}


\begin{proof}
 Let us consider $g^t$, $g_i^t = x_i^{\top} \beta^t - l_i^{*'}(\alpha_i^t)$. After computing the primal $\beta^t$  with the primal-dual relation~\eqref{eq:bigB}, Algorithm~\ref{alg:primal_dual_inner} also performs primal coordinate descent starting with $\beta^t$ using~\eqref{eq:thresh} to the improve super-gradient $g^t$. 
 
 Let $\breve{\beta}^t$ be the output of operation~\eqref{eq:bigB} at step $t$.  
  From the expression of $\beta^t$~\eqref{eq:bigB}, if $\breve{\beta}_j^t \neq 0$, $\breve{\beta}^t_j (\alpha^t)=\mathrm{sign}\big(\eta_j(\alpha^t)\big) \big(|\eta_j(\alpha^t)| -\frac{\lambda_1 }{2\lambda_2} \big) $. With $||x_{\cdot j}|| = 1$, $|\eta_j(\alpha^t)|= \frac{|x_{\cdot j}^{\top}\alpha^t|}{2\lambda_2}\leq \frac{\sqrt{n}\varphi}{2\lambda_2}$. Here $\varphi = ||\alpha^t||_{\infty}$. Then we have 
\begin{align}\label{eq:bound_beta_1}
 |\breve{\beta}_j^t| \leq |\eta_j(\alpha^t)| - \frac{\lambda_1}{2\lambda_2} \leq \frac{\sqrt{n}\varphi - \lambda_1}{2\lambda_2}= \frac{\varrho}{2\lambda_2} , 
\end{align}
with  $\varrho = \sqrt{n}\varphi - \lambda_1$.

 According to~\eqref{eq:thresh}, with $\beta$ as the input,  the non-zero output at entry $j$
\begin{align*}
 \grave{\beta}_j = T(\beta; \lambda_0, \lambda_1,\lambda_2)= \mathrm{sign}(\tilde{\mathbf{\beta}}_j) \frac{|\tilde{\mathbf{\beta}}_j| - \lambda_1}{1 + 2\lambda_2},
\end{align*}
 with
\begin{align*}
\tilde{\mathbf{\beta}}_j = \big \langle y - \sum_{i:i\ne j, i\in S} x_{\cdot i}\mathbf{\beta}_i, x_{\cdot j}  \big\rangle 
= (y - X\mathbf{\beta})^{\top} x_{\cdot j} + \mathbf{\beta}_j  x_{\cdot  j}^{\top}x_{\cdot j} 
=y^{\top} x_{\cdot j} -  \beta^{\top}X^{\top} x_{\cdot j}  + \mathbf{\beta}_j 
\end{align*}
and $ S=\mathrm{supp}(\beta)$. With $\beta_j \neq 0$
we have $|\tilde{\mathbf{\beta}}_j| - \lambda_1 \geq 0$
.
Then
 \begin{align*}
  |\grave{\beta}_j| &= \frac{1}{1+2\lambda_2} \big( | \tilde{\mathbf{\beta}}_j| -\lambda_1\big) \leq \frac{1}{1+2\lambda_2} \big( |y^{\top} x_{\cdot j}| +|  \beta^{\top}X^{\top} x_{\cdot j}|  + |\mathbf{\beta}_j|  -\lambda_1\big) \\
  &\leq \frac{1}{1+2\lambda_2} \big(  |y^{\top} x_{\cdot j}| +\sqrt{\beta^{\top}\beta  x_{\cdot j}^{\top}XX^{\top}x_{\cdot j} }+ |\beta_j| -\lambda_1\big) .
\end{align*}
 Let input $\beta = \breve{\beta}^t$, with~\eqref{eq:bound_beta_1}  the upper bound of the output after one round coordinate descent will be  
  \begin{align*}
  |\beta^t_j| & \leq \frac{1}{1+2\lambda_2} \big(  |y^{\top} x_{\cdot j}| +\sqrt{\breve{\beta}^{\top}\breve{\beta}  x_{\cdot j}^{\top}XX^{\top}x_{\cdot j} }+ |\breve{\beta}_j| -\lambda_1\big) \\
  &\leq   \frac{1}{1+2\lambda_2} \big(  |y^{\top} x_{\cdot j}| +  \frac{1 + \sqrt{n}||X^{\top}x_{\cdot j}||}{2\lambda_2}\varrho -\lambda_1\big) \\
  &\leq \frac{1}{2\lambda_2(1+2\lambda_2)}(2\lambda_2m_1 +  \varrho +  \sqrt{n}m_2 \varrho - 2\lambda_1\lambda_2):=\psi .
\end{align*}
 
Here $m_1 =\max_{j:1\leq j \leq p} |y^{\top}x_{\cdot j}|$, and $m_2 =\max_{j:1\leq j \leq p} ||X^{\top}x_{\cdot j}||$. 
 Let $m_3 = max_{i, \alpha_i^t \in \mathcal{F}} |l_i^{*'}(\alpha_i^t)|$. Then 
 \begin{align*}
  |g_i^t| \leq |x_i^{\top} \beta^t| + |l_i^{*'}(\alpha_i^t)|\leq  \sqrt{||x_i||^2||\beta^t||^2}+ m_3\leq   \sqrt{p}\psi \vartheta + m_3 \ , \quad  \forall 1\leq i \leq n.
\end{align*}
Hence,
 \begin{align}\label{eq:c0}
     ||g^t|| \leq c_0 := \sqrt{np}\psi \vartheta + \sqrt{n}m_3 .
 \end{align}

Let $h^t = ||\alpha^t - \bar{\alpha}||$ and $v^t = \langle g^t, \bar{\alpha} - \alpha^t \rangle$. The concavity of $D$ implies $v^t \geq 0$. According to Lemma~\ref{Lem:dual_range}, 
\begin{align}\label{eq:h_inequ}
h^t = ||\alpha^t - \bar{\alpha}|| \leq \sqrt{ \frac{2   v^t  }{ \gamma}} .
\end{align}

Let $\omega^t$ be the step size of dual variable at step $t$. Then we have 
\begin{align*}
(h^t)^2  =&||\alpha^t - \bar{\alpha}||^2\\
=& ||P_{\mathcal{F}^n}\big(\alpha^{t-1} + \omega^{t-1} g^{t-1} \big) - \bar{\alpha}||^2 \\
\leq & || \alpha^{t-1} +  \omega^{t-1} g^{t-1} - \bar{\alpha} ||^2  \\
=&  (h^{t-1})^2 - 2\omega^{t-1}v^{t-1} + (\omega^{t-1})^2 ||g^{t-1}||^2   \\
\leq & (h^{t-1})^2 - \omega^{t-1}(    \gamma ) (h^{t-1})^2 + (\omega^{t-1})^2c_0^2 . 
\end{align*}
The last step is due to~\eqref{eq:h_inequ} and~\eqref{eq:c0}. 
Let $\omega^{t-1} = \frac{ 1 }{   \gamma t} $. Then we get 
\begin{align*}
(h^t)^2  \leq \big(1 - \frac{1}{t} \big) (h^{t-1})^2 + \frac{c_0^2}{\big(  \gamma\big)^2 t^2} .
\end{align*}
Recursively applying the above inequality we get
\begin{align*}
(h^t)^2  \leq   \frac{c_0^2}{\gamma^2}  \bigg( \frac{1}{t} + \frac{\ln t}{t}\bigg) \leq  c_1 \bigg( \frac{1}{t} + \frac{\ln t}{t}\bigg).
\end{align*}
Here $c_0=  \frac{\vartheta \sqrt{np}}{2\lambda_2(1+2\lambda_2)}(2\lambda_2m_1 +  \varrho +  \sqrt{n}m_2 \varrho - 2\lambda_1\lambda_2)+\sqrt{n}m_3$.
This proves the bound in the theorem. 
\end{proof} 

Different from the complexity analysis in~\cite{Liu17} regarding hard thresholding, the analysis presented here is based on the primal-dual problem structures in \eqref{eq:primal1} and \eqref{eq:dual2}. Moreover, the study in~\cite{Liu17} focus on the dual updating steps regarding hard thresholding. Whereas Theorem~\ref{Thm:inner_convg} includes the complexity  of both primal and dual updating steps given in Algorithm~\ref{alg:primal_dual_inner}. We further prove the convergence of primal  variable and the duality gap. 

\begin{lemma}\label{Lem:dual_gap_range}
For $\alpha \in \mathcal{F}^n$, with $\beta = \beta(\alpha)$ the primal-dual gap can be written as 
\begin{align} \label{eq:gap_eq}
P(\beta) - D(\alpha)= \sum_{i=1}^n \bigg(l(\beta^{\top} x_i, y_i)  + l_i^*(\alpha_i)  \bigg) -  \alpha^{\top} X \beta \ .
\end{align}
Moreover, with $\theta \in [\partial l_1(\beta^{\top}x_1),  \partial l_2(\beta^{\top}x_2), ..., \partial l_n(\beta^{\top}x_n)]^{\top}$, we have 
\begin{align}
 P(\beta) - D(\alpha) \leq  \langle g_{\alpha}, \theta - \alpha \rangle.
\end{align}
\end{lemma}
\begin{proof}
With the definitions of $P(\beta)$, $D(\alpha)$, and~\eqref{eq:dual_normal} 
\begin{align} \notag
&P(\beta) - D(\alpha)\\ \notag
= &\sum_{i=1}^n l(\beta^{\top} x_i, y_i)   +  \lambda_1 ||\beta||_1 +\lambda_2 ||\beta||_2^2 + \lambda_0 ||\beta||_0 -\bigg( -  \sum_{i=1}^n  l_i^*(\alpha_i)   + \sum_{j=1}^p \Psi(\eta_j(\alpha); \lambda_0,  \lambda_1, \lambda_2) \bigg) \\ \notag
= &\sum_{i=1}^n l(\beta^{\top} x_i, y_i)   +  \lambda_1 ||\beta||_1 +\lambda_2 ||\beta||_2^2 + \lambda_0 ||\beta||_0 -\bigg( -  \sum_{i=1}^n  l_i^*(\alpha_i)   - \lambda_2 ||\beta||^2_2 + \lambda_0 ||\beta||_0 \bigg) \\ \label{eq:l1l2}
=&\sum_{i=1}^n \bigg(l(\beta^{\top} x_i, y_i)  + l_i^*(\alpha_i)  \bigg) + \lambda_1 ||\beta||_1 +2\lambda_2 ||\beta||_2^2 \ .
\end{align}
Let $S = \mathrm{supp}(\beta)$, the last two terms in above equation can rewritten as 
\begin{align} \notag
\lambda_1 ||\beta||_1 +2\lambda_2 ||\beta||_2^2
=&2\lambda_2\bigg(\frac{\lambda_1}{2\lambda_2} ||\beta||_1 +   ||\beta||_2^2\bigg) \\ \notag
=&2\lambda_2\sum_{j\in S}\bigg(\frac{\lambda_1}{2\lambda_2} (|\eta_j(\alpha)| -\frac{\lambda_1 }{2\lambda_2}) +   (|\eta_j(\alpha)| -\frac{\lambda_1 }{2\lambda_2})^2\bigg) \\ \notag
=&2\lambda_2\sum_{j\in S}\bigg(   (|\eta_j(\alpha)| -\frac{\lambda_1 }{4\lambda_2})^2  -  \big(\frac{\lambda_1}{4\lambda_2}\big)^2 \bigg) \\ \notag
=&2\lambda_2 \sum_{j\in S} |\eta_j(\alpha)|  \bigg(|\eta_j(\alpha)| -\frac{\lambda_1 }{2\lambda_2}\bigg) \\ \notag
\overset{}{=} &2\lambda_2 \sum_{j\in S} \eta_j(\alpha) \mathrm{sign}(\eta_j(\alpha)) \bigg(|\eta_j(\alpha)| -\frac{\lambda_1 }{2\lambda_2} \bigg) \\ \label{eq:aXb}
=& -  \alpha^{\top} X \beta  \ .
\end{align}
The last step is due to the definitions of $\eta(\alpha)$ and $\beta(\alpha)$~\eqref{eq:B_frak}. 
With equations~\eqref{eq:l1l2} and \eqref{eq:aXb}, we prove equation~\eqref{eq:gap_eq}.

For $i\in [n]$, with the maximizing argument property of convex conjugate 
\begin{align*}
&l_i(\beta^{\top} x_i)= \beta^{\top} x_i l'_i(\beta^{\top} x_i) - l^*_i(l'_i(\beta^{\top} x_i) ) , \\
& l^*_i(\alpha_i)= \alpha_i l^{*'}_i(\alpha_i) - l_i(l^{*'}_i(\alpha_i) ) \ .
\end{align*}
Adding above two equations, and applying Fenchel-Young inequality, we get 
\begin{align}\notag
l_i(\beta^{\top} x_i)+  l^*_i(\alpha_i) =& \beta^{\top} x_i l'_i(\beta^{\top} x_i)  + \alpha_i l^{*'}_i(\alpha_i) - \bigg(l^*_i(l'_i(\beta^{\top} x_i) )+ l_i(l^{*'}_i(\alpha_i) ) \bigg) \\ \label{eq:conj_inequ}
\leq & \beta^{\top} x_i l'_i(\beta^{\top} x_i)  + \alpha_i l^{*'}_i(\alpha_i)- l^{*'}_i(\alpha_i)l'_i(\beta^{\top} x_i).
\end{align}
Moreover, with $\theta \in [\partial l_1(\beta^{\top}x_1),  \partial l_2(\beta^{\top}x_2), ..., \partial l_n(\beta^{\top}x_n)]^{\top}$
\begin{align} \notag
&\langle g_{\alpha}, \theta - \alpha \rangle \\ 
\notag
=& \sum_{i=1}^n \big(\beta^{\top}x_i - l_i^{*'}(\alpha_i) \big) \big(l_i'(\beta^{\top}x_i) - \alpha_i\big)\\ \notag
=& \sum_{i=1}^n \bigg(\beta^{\top} x_i l'_i(\beta^{\top} x_i)  + \alpha_i l^{*'}_i(\alpha_i)- l^{*'}_i(\alpha_i)l'_i(\beta^{\top} x_i) -  \alpha_i\beta^{\top} x_i \bigg) \\ \label{eq:inequ1}
\geq &  \sum_{i=1}^n \bigg(l_i(\beta^{\top} x_i)+  l^*_i(\alpha_i) -  \alpha_i\beta^{\top} x_i \bigg) \\ \label{eq:dgap2}
=& P(\beta) - D(\alpha).
\end{align}
Step~\eqref{eq:inequ1} is due to~\eqref{eq:conj_inequ}, and step~\eqref{eq:dgap2} is by~\eqref{eq:gap_eq}.  It concludes the lemma.
\end{proof}

\newpage

\noindent\textbf{Theorem~\ref{Thm:prim_conv}} {\it
Assume that $l_i$ is $1/\mu$-smooth,  $||x_{i}|| \leq \vartheta  \ \forall 1\leq i \leq n$, and $||x_{\cdot j}||= 1  \ \forall 1\leq j \leq p$. Let $ t_1 =\frac{3c_1}{\bar{\delta}^2}\log \frac{3c_1}{\bar{\delta}^2}$,  with $t > t_1$, we have $\mathrm{supp}(\beta(\alpha)) = \mathrm{supp}(\bar{\beta})$ and $|| \beta(\alpha) - \bar{\beta}|| \leq \frac{\sigma_{max}(X_S)}{2\lambda_2} ||\alpha - \bar{\alpha}||$. Moreover, let $t_2 = \frac{3c_1 c_2}{\epsilon} \log \frac{3c_1 c_2}{\epsilon}$,  $c_2 = c_0\bigg(1+ \frac{\sigma_{max}(X_S)}{2\mu\lambda_2} \bigg)$,  for any $\epsilon>0$ with $t >\mathrm{max}\{t_1, t_2\}$, we have $P(\beta^t) - D(\alpha^t) \leq \epsilon$.
}
\begin{proof}
With Theorem~\ref{Thm:inner_convg}, let $|| \alpha - \bar{\alpha}||^2 \leq c_1 \bigg( \frac{1}{t} + \frac{\ln t}{t}\bigg)  \leq \bar{\delta}^2$, and it implies that 
 $t\geq \frac{3c_1}{\bar{\delta}^2}\log \frac{3c_1}{\bar{\delta}^2}$. Hence, with $t\geq t_1 :=\frac{3c_1}{\bar{\delta}^2}\log \frac{3c_1}{\bar{\delta}^2}$, $\mathrm{supp}(\beta(\alpha)) = \mathrm{supp}(\bar{\beta})$ and $|| \beta(\alpha) - \bar{\beta}|| \leq \frac{\sigma_{max}(X_S)}{2\lambda_2} ||\alpha - \bar{\alpha}||$ according to Lemma~\ref{Lem:support}.

Let $\theta^t =  [\partial l_1((\beta^{t})^{\top}x_1),  \partial l_2((\beta^{t})^{\top}x_2), ..., \partial l_n((\beta^{t})^{\top}x_n)]^{\top}$. According to Lemma~\ref{Lem:dual_gap_range}, 
\begin{align} \notag
P(\beta^t) - D(\alpha^t)
\leq & \langle g^t_{\alpha^t}, \theta^t - \alpha^t \rangle \\
\leq & || g^t_{\alpha^t}|| \big( || \theta^t - \bar{\alpha}|| + || \bar{\alpha} -  \alpha^t || \big) .
\end{align}
As $\bar{\alpha} \in [\partial l_1(\bar{\beta}^{\top}x_1),  \partial l_2(\bar{\beta}^{\top}x_2), ..., \partial l_n(\bar{\beta}^{\top}x_n)]^{\top}$,   $l_i()$ is $\frac{1}{\mu}$ smooth, and $||x_{\cdot j}||=1$,  we have 
\begin{align} \notag
|| \theta^t - \bar{\alpha}|| \leq \frac{1}{\mu} ||\beta^t - \bar{\beta}|| \leq \frac{\sigma_{max}(X_S)}{2\mu\lambda_2} ||\alpha^t - \bar{\alpha}||.
\end{align}
Therefore,
\begin{align} \notag
P(\beta^t) - D(\alpha^t)
\leq & || g^t_{\alpha^t}|| \big( || \theta^t - \bar{\alpha}|| + || \bar{\alpha} -  \alpha^t || \big)\\ \label{eq:dgap3}
\leq & c_0\bigg(1+ \frac{\sigma_{max}(X_S)}{2\mu\lambda_2} \bigg)||\alpha^t - \bar{\alpha}|| .
\end{align}
If $t >\mathrm{max}\{t_1, t_2\}, t_2 := \frac{3c_1 c_2^2}{\epsilon^2} \log \frac{3c_1 c_2^2}{\epsilon^2}$, and $c_2 := c_0\bigg(1+ \frac{\sigma_{max}(X_S)}{2\mu\lambda_2} \bigg)$,  we have $||\alpha^t - \bar{\alpha}|| \leq \frac{\epsilon}{c_0\bigg(1+ \frac{\sigma_{max}(X_S)}{2\mu\lambda_2} \bigg)}$, and with~\eqref{eq:dgap3},  $P(\beta^t) - D(\alpha^t) \leq \epsilon$.  
\end{proof}




\newpage

\section{Dual Problems of Some Loss Functions}\label{sec:dual_of_losses}

\subsection{Logistic Loss}

The primal form of logistic regression is given by
\begin{align}\label{eq:primal_logistic_reg}
\min_{\beta \in \mathbb{R}^p} F(\beta) &= f(\beta) + \lambda_0 ||\beta||_0,\\ \notag
~~~~f(\beta) &= \sum_{i=1}^n \big(-y_i x_i^{\top}\beta + \log (1+ \exp(x_i^{\top} \beta) ) \big)   +  \lambda_1 ||\beta||_1 +\lambda_2 ||\beta||_2^2.
\end{align}
Here $l_i(u; y_i) =  -y_i u + \log (1+ \exp(u) ) $, and $l_i^*(\alpha_i) = (\alpha_i + y_i) \log (\alpha_i + y_i) + (1-\alpha_i - y_i) \log(1-\alpha_i - y_i), \alpha_i + y_i \in [0,1]$, here $u=\beta^{\top}x_i$. 
The dual objective is 
\begin{align*}
D(\alpha) =  & -  (\alpha + y)\log(\alpha + y) - (1-\alpha - y)\log(1-\alpha - y)  +  \sum_{j=1}^p \Psi( - \frac{1}{2\lambda_2} \sum_{i=1}^n \bar{\alpha}_i x_i; \lambda_0, \lambda_1, \lambda_2).
\end{align*}
Here $ \Psi()$ is given by~\eqref{eq:psi}. With $\alpha_i + y_i \in [0,1]$, the dual feasible project operator is $\mathcal{P}_{\mathcal{F}}(\alpha_i) = 
\begin{cases}
\alpha_i \quad \ \quad -y_i \leq \alpha \leq 1-y_i \\
0  \quad \ \    \quad   \ \alpha \leq -y_i \\
1  \quad \ \   \quad \  \alpha \geq 1-y_i \\
\end{cases}.$
The super gradient regarding logistic regression is 
\begin{align*}
g_{\alpha} =  \nabla_{\alpha} D(\alpha) = & \bigg[\beta^{\top}(\alpha)x_1 - \log\big(\frac{\alpha_1+y_1}{1- \alpha_1 - y_1}\big), ...,  \beta^{\top}\big(\alpha)x_n - \log(\frac{\alpha_n + y_n}{1- \alpha_n - y_n}\big) \bigg]^{\top} .
\end{align*}
There is no closed form of updating formula with coordinate descent regarding the primal problem~\eqref{eq:primal_logistic_reg}. We can apply proximal algorithm~\citep{parikh2014proximal} to this type of primal loss functions. 
\subsection{Huber Loss}
Consider a regression problem with Huber loss, i.e.,
\begin{align*}
&\min_{\beta \in \mathbb{R}^p} F(\beta)  = \sum_{i=1}^n \left\{ l_{Huber}(y_i x_i^T \beta) \right\} + \lambda_0 \|\beta\|_0   + \lambda_1 \|\beta\|_1 + \lambda_2 \|\beta\|_2,\\ 
&l_{Huber}(y_i x_i^T \beta) =\begin{cases}
0, & y_i x_i^T \beta \geq 1\\
1 - y_i x_i^T \beta - \gamma/2, & y_i x_i^T \beta  < 1 - \gamma\\
\frac{1}{2\gamma}(1 - y_i x_i^T \beta )^2, & \textrm{otherwise}
\end{cases}
\end{align*}
with $\gamma$ being some hyper tuning parameter.
The dual function of $l_{Huber}(\cdot)$ is 
\begin{align*}l^{\ast}_{Huber}(\alpha_i) =
\begin{cases}
y_i \alpha_i + \frac{\gamma}{2}\gamma \alpha^2, & -1 \leq y_i x_i^T \leq 0 \\
+\infty, & \textrm{otherwise}
\end{cases}.\end{align*}
Therefore, the corresponding Lagrangian function is
$L(\beta, \alpha) = \sum_i \{ \alpha_i x_i^T \beta - l_{Huber}^{\ast}(\alpha_i)\}  
+ \lambda_0 ||\beta||_0   +  \lambda_1 ||\beta||_1 + \lambda_2 ||\beta||^2 . \nonumber $
The dual problem can be written as 
\begin{align*}
\max_{\alpha \in \mathcal{F}^n} -  \sum_{i=1}^n  l_{Huber}^*(\alpha_i)   + \sum_{j=1}^p \Psi(\eta_j(\alpha); \lambda_0,  \lambda_1, \lambda_2).
\end{align*}

\end{document}